\documentclass[a4paper,12pt]{amsart}

\usepackage{kantlipsum} 
\calclayout{}

\usepackage[ruled,vlined]{algorithm2e}
\usepackage{amsfonts}
\usepackage{amsmath}    
\usepackage{amssymb}    
\usepackage{bm}         
\usepackage{bbm}
\usepackage{color}
\usepackage{enumitem}
\usepackage{graphicx}   
\usepackage{natbib}    
\usepackage{epsf}
\usepackage{lscape}
\usepackage{enumitem}
\usepackage{hyperref}   
\usepackage{longtable}
\bibpunct{(}{)}{;}{a}{,}{,}

\usepackage[normalem]{ulem}

\usepackage{subcaption}

\usepackage{bm}

\theoremstyle{thmstyleone}%
\newtheorem{theorem}{Theorem}
\newtheorem{proposition}{Proposition}
\newtheorem{corollary}{Corollary}%
\newtheorem{lemma}{Lemma}%

\theoremstyle{thmstyletwo}%

\theoremstyle{thmstylethree}%

\newcommand{\bgamma}{\bm{\gamma}}
\newcommand{\bGamma}{\bm{\Gamma}}
\newcommand{\bpi}{\bm{\pi}}
\newcommand{\btheta}{\bm{\theta}}
\newcommand{\bomega}{\bm{\omega}}

\newcommand{\R}{\mathcal{R}}
\newcommand{\T}{\mathcal{T}}

\newcommand{\bx}{\bm{x}}
\newcommand{\by}{\bm{y}}
\newcommand{\bz}{\bm{z}}

\newcommand{\bA}{\bm{A}}
\newcommand{\bD}{\bm{D}}
\newcommand{\bP}{\bm{P}}
\newcommand{\bU}{\bm{U}}
\newcommand{\bV}{\bm{V}}
\newcommand{\bX}{\bm{X}}
\newcommand{\bZ}{\bm{Z}}

\begin{document}

  \title{\bf Empirical Bayes for Data integration}
  \author{Paul Rognon-Vael,
          Pompeu Fabra University\\
          and David Rossell, 
          Pompeu Fabra University\\
}

\bigskip
\begin{abstract}
We discuss the use of empirical Bayes for data integration, in the sense of transfer learning. Our main interest is in settings where one wishes to learn structure (e.g., feature selection) and one only has access to incomplete data from previous studies, such as summaries, estimates or lists of relevant features. We discuss differences between full Bayes and empirical Bayes, and develop a computational framework for the latter. We discuss how empirical Bayes attains consistent variable selection under weaker conditions (sparsity and betamin assumptions) than full Bayes and other standard criteria do, and how it attains faster convergence rates. Our high-dimensional regression examples show that fully Bayesian inference enjoys excellent properties, and that data integration with empirical Bayes can offer moderate yet meaningful improvements in practice.
\end{abstract}

\noindent%
\vfill

  \maketitle

\section{Introduction}\label{sec:intro}

Data integration is the task of combining multiple datasets or sources of information into a single analysis. 
It has the potential to improve inference and/or predictions,
particularly in settings where data are scarce relative to the problem's complexity (number of parameters, models or questions being considered).
This is a standard statistical problem that is receiving renewed attention within the context of transfer learning.
Therein, one wishes to transfer findings obtained from a setting where one has large or cheaply obtained data, to another setting where data are scarcer or costlier.
We discuss the role of empirical Bayes methodology for data integration.
In contrast to parameter estimation or prediction problems, which have been better studied, we focus on structural learning.
By structural learning we refer to model selection tasks where one seeks to learn about structural features of the data-generating truth, say whether a parameter is zero, and we use variable selection as a canonical example.

There are two main scenarios where one may seek to integrate data, illustrated in Figure \ref{fig:plate}. 
The first (left panel) is when one has full access to multiple datasets. For example, one has a separate dataset for multiple cancer types, or for multiple schools.
We refer to this scenario as {\it data integration with full data}.
There is a well-developed literature for this setting. A popular approach is to use hierarchical models that share information across datasets, see for instance \cite{gelman:2007} for a monograph. 
For data integration in covariance models see \cite{avalospacheco:2022,chandra:2024}, in clustering see \cite{lock:2013} and in functional data analysis see \cite{yang_jingjing:2015}.
See also \cite{jacob:2017,nott:2023} and references therein for data integration strategies when the assumed model is misspecified.

Our main focus is in the second scenario (Figure \ref{fig:plate}, right), where one is interested in a single dataset $\by$, and one only has full access to that data. One also has other information that could be useful, for example summaries of previous datasets (denoted $\bZ=(\bz_1,\ldots,\bz_q)$ in Figure \ref{fig:plate}).
This second scenario is very common and more challenging, because one does not have full access to the previous data and the provided summaries may be poor.
An application that we illustrate later is a gene expression study in human colon cancer, where one has a list of genes that were found to be important in mice. We have full access to the human data, but only the list of genes from the mice experiment.
We refer to this scenario as {\it data integration with meta-covariates},
i.e. with summaries $\bZ$ that may be informative as to what parameters are more likely to be important. 
Literature focusing on this scenario includes \cite{stingo:2011, cassese:2014}, who proposed Bayesian variable selection 
where prior inclusion probabilities depend on biological knowledge and meta-covariates. 
\cite{chen_tinghuei:2021} predicted disease outcomes by allowing LASSO penalties to depend on functional annotations. 
In causal analysis, \cite{belloni:2012,antonelli:2021,papaspiliopoulos:2025} 
penalized the inclusion of control covariates using their degree of association with the treatment of interest.
See \cite{vandewiel:2019} for an overview on the use of meta-covariates in regression.
Beyond regression, in Gaussian graphical models \cite{peterson:2016} and \cite{jewson:2023} used node and edge covariates to drive prior edge inclusion probabilities.
In factor models, \cite{schiavon:2022} used meta-covariates to determine non-zero loadings.

\begin{figure}
\begin{center}
\begin{tabular}{cc}
\includegraphics[width=0.49\textwidth]{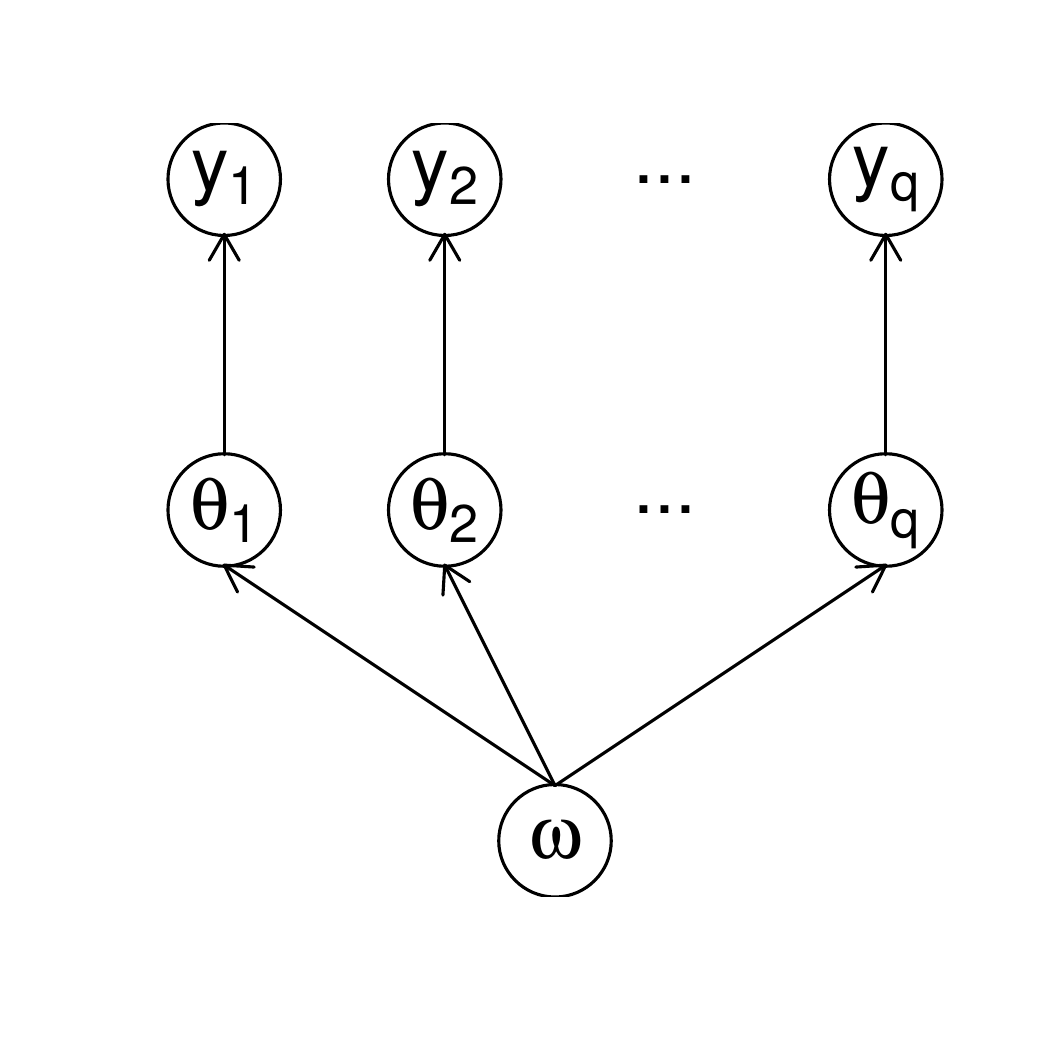} &
\includegraphics[width=0.49\textwidth]{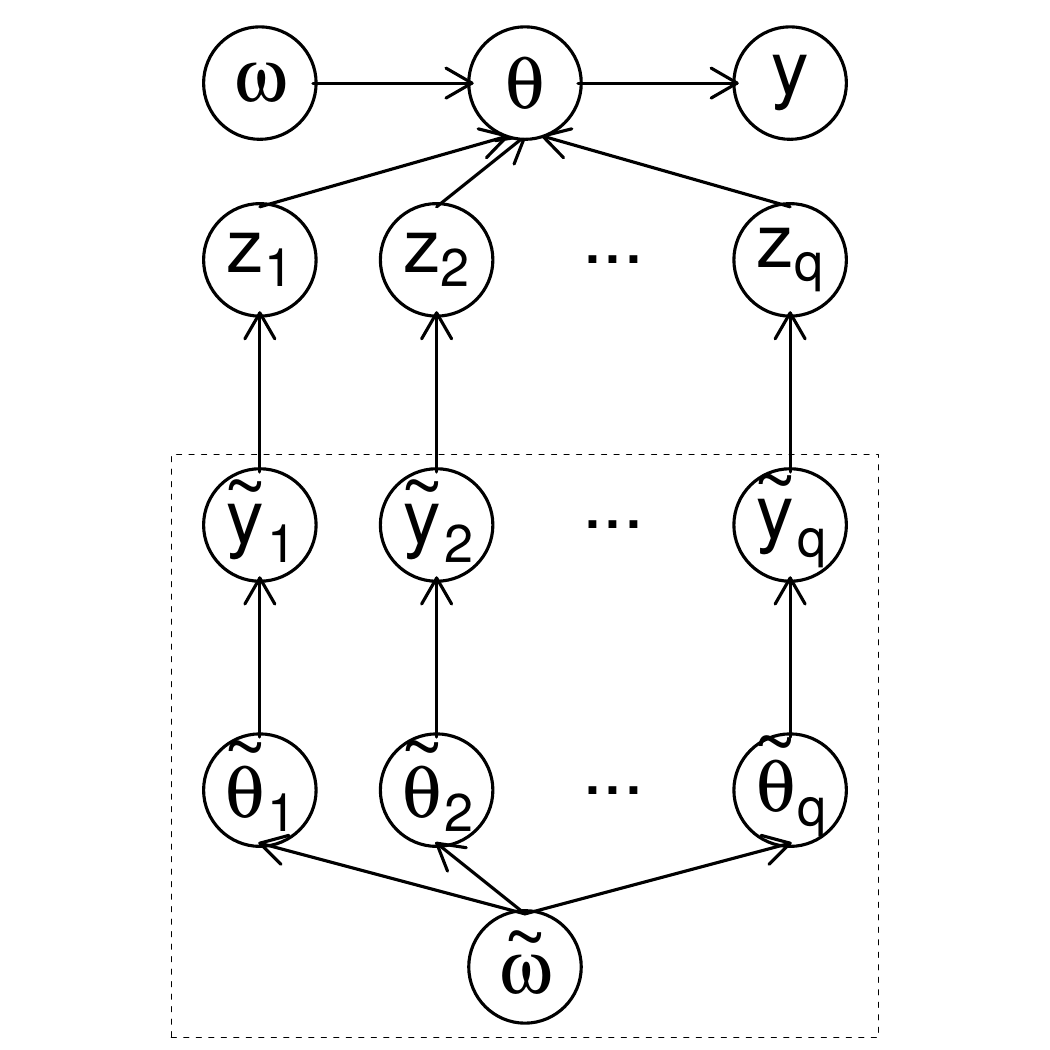}
\end{tabular}
\end{center}
\caption{Data integration with full data (left) and with meta-covariates (right). With full data one observes $q$ datasets with underlying parameters $\btheta_1,\ldots,\btheta_q$. The hyper-parameter $\bomega$ allows sharing information.
With meta-covariates one has access to summaries $\bz_1,\ldots,\bz_q$ extracted from unobserved data $\tilde{\by}_1,\ldots,\tilde{\by}_q$. These summaries inform the parameter $\btheta$ governing the dataset of interest $\by$}
\label{fig:plate}
\end{figure}

Bayesian statistics provides a natural framework to integrate prior data such as meta-covariates $\bZ$ into the analysis.
In principle, it suffices to set a prior on the parameters of interest ($\btheta$ in Figure \ref{fig:plate}, right panel) that depends on $\bZ$ and some hyper-parameters $\bomega$.
There may be situations where an expert is willing to set such a prior subjectively. For example a biologist believes that if a gene is known to be involved in gastrointestinal cancer, then it has 0.9 prior probability of also being involved in colorectal cancer.
This situation is rare. Experts do not have such a precise knowledge, specially when there are $q >1$ meta-covariates.
This is important: if the prior grossly failed to reflect the true effect of $\bZ$ (i.e., the expert provided a poor prior for $\bomega$), then one may obtain worse inference than if $\bZ$ were not used at all.

We explain how empirical Bayes can be useful for data integration. 
By modeling the relation between the meta-covariates $\bZ$ and the parameters $\btheta$ governing the data $\by$,
one may learn $\bomega$ and hence integrate $\bZ$ in a way that's less sensitive to misspecification.
Empirical Bayes can be seen as a frequentist method where the model-fitting criteria are based on Bayesian thinking, and as such it can help reconcile long-standing debates between frequentist and Bayesian inference, see for example \cite{efron:2024}. 

The paper is structured as follows.
Section \ref{sec:ebayes} reviews empirical Bayes and the frequentist validity of Bayesian methods, in the setting of data integration with full data.
In Section \ref{sec:framework} we outline an empirical Bayes framework for data integration with meta-covariates.
For concreteness, we focus on variable selection, but the setup extends directly to more general settings.
Section \ref{sec:theory} provides theory on high-dimensional linear regression with meta-covariates.
We discuss how a simple form of meta-covariate (partitioning covariates into blocks) leads to consistent model selection (identifying the truly non-zero parameters) in settings where it is otherwise mathematically impossible (milder sparsity or betamin conditions), and that it improves consistency rates.
Sections \ref{sec:ebayes}-\ref{sec:framework} are not methodologically new,
but they reflect the views of the authors on using empirical Bayes for data integration.
The main novelty in Section \ref{sec:theory} is adapting results from \cite{rognon:2025} to a Bayesian setting.
Section \ref{sec:computation} includes new computational results to estimate the hyper-parameters $\hat{\bomega}$  that apply to very general settings, and in particular beyond regression. 
Evaluating the log marginal likelihood $p(\by \mid \bomega)$ requires a sum over models (e.g. $2^p$ in regression), but its gradient can be expressed as a sum over $p$ terms, and we propose an Expectation-Maximization (EM, \cite{dempster:1977}) algorithm. 
In Section \ref{sec:results} we present simulations and an application to colon cancer data.
These illustrate that fully Bayesian inference enjoys excellent properties and that its results are often aligned with those from empirical Bayes. However in some settings, including our colon cancer example, there can be moderate but non-negligible differences between the two frameworks.
Section \ref{sec:discussion} concludes.

\section{Empirical Bayes} \label{sec:ebayes}

\begin{figure}
\begin{center}
\begin{tabular}{cc}
\includegraphics[width=0.5\textwidth]{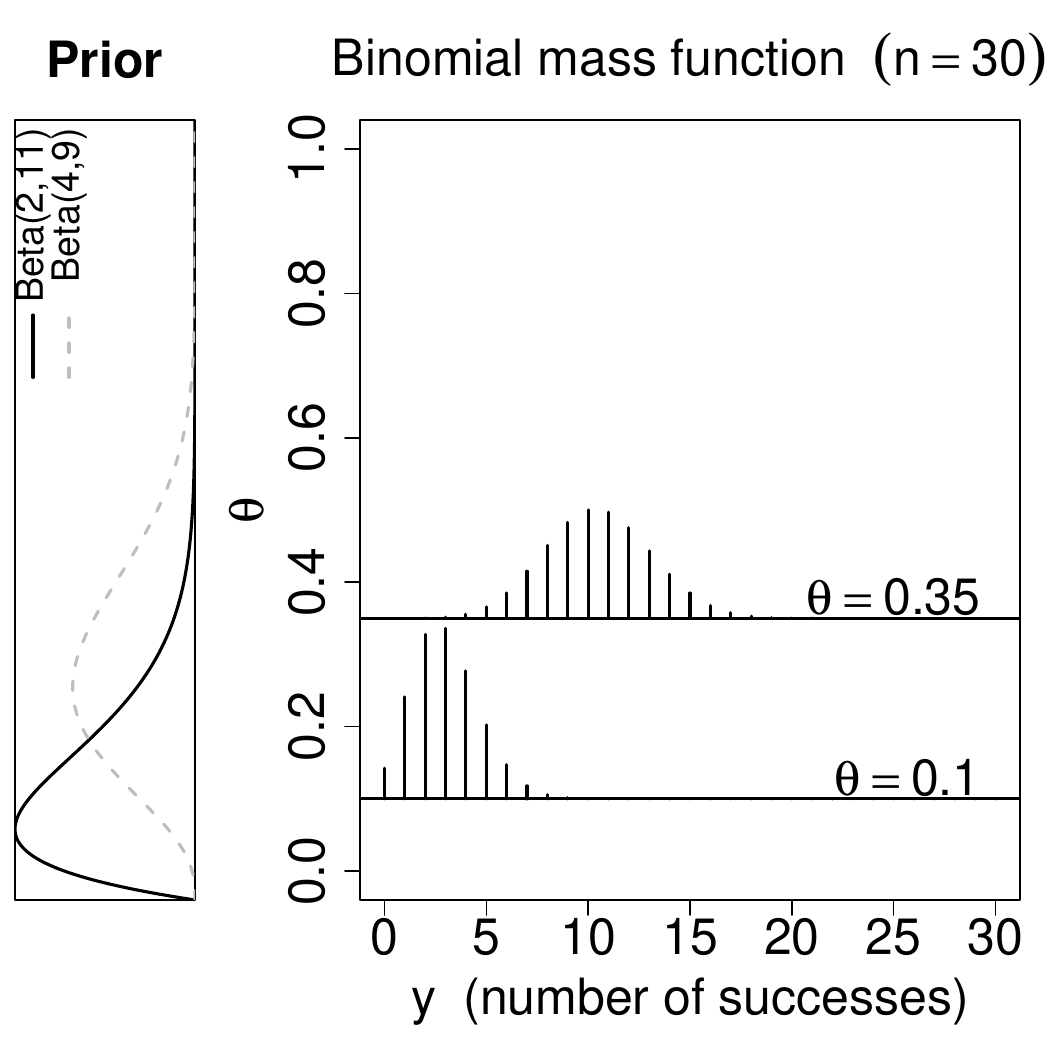} &
\includegraphics[width=0.5\textwidth]{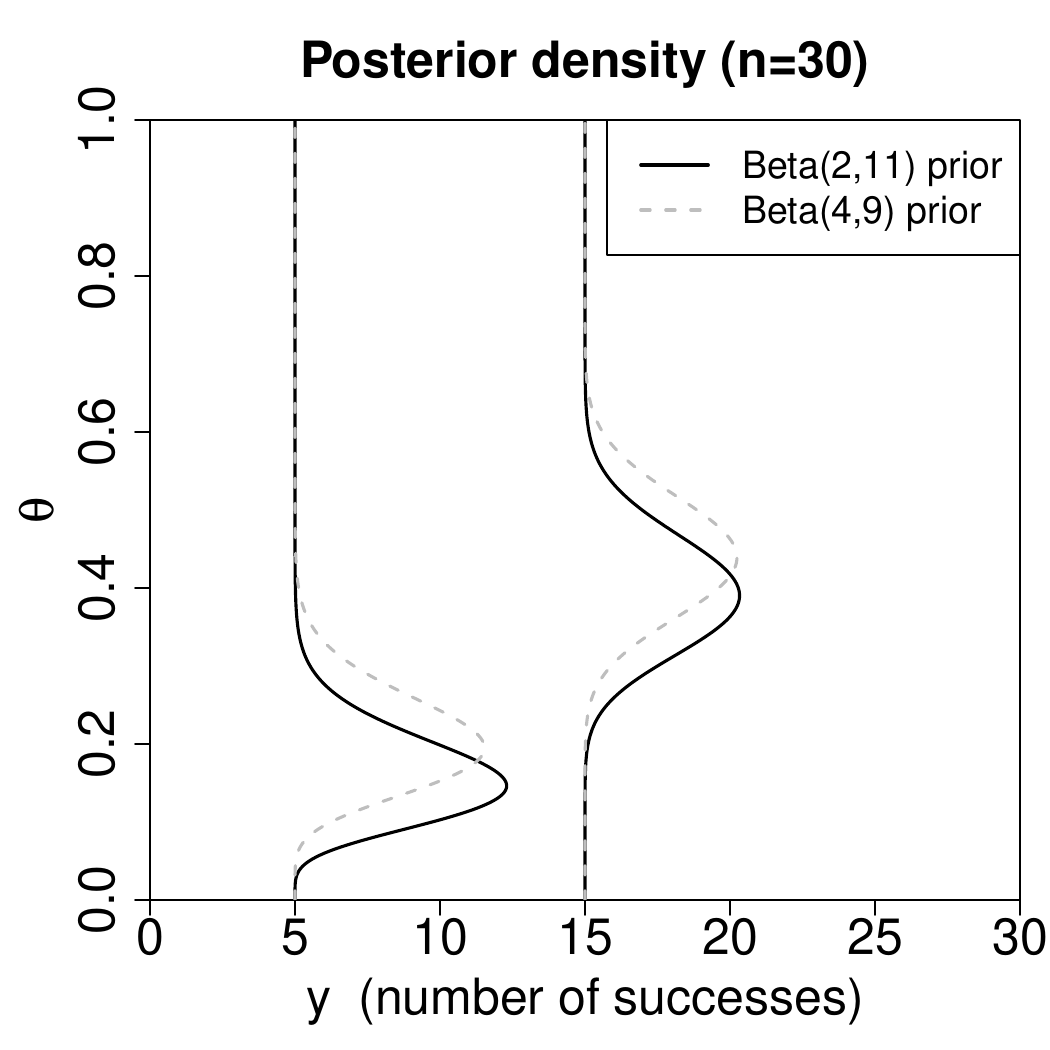} \\
\includegraphics[width=0.5\textwidth]{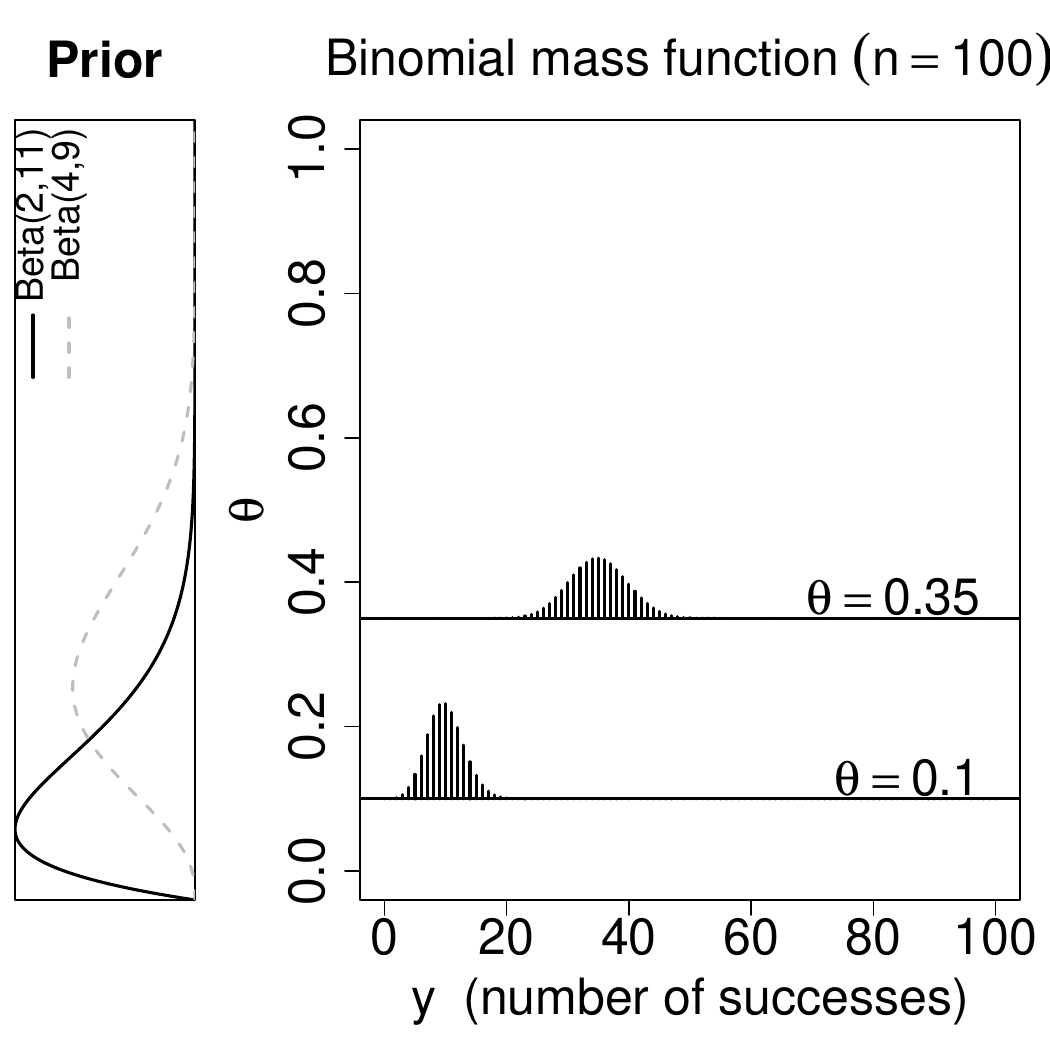} &
\includegraphics[width=0.5\textwidth]{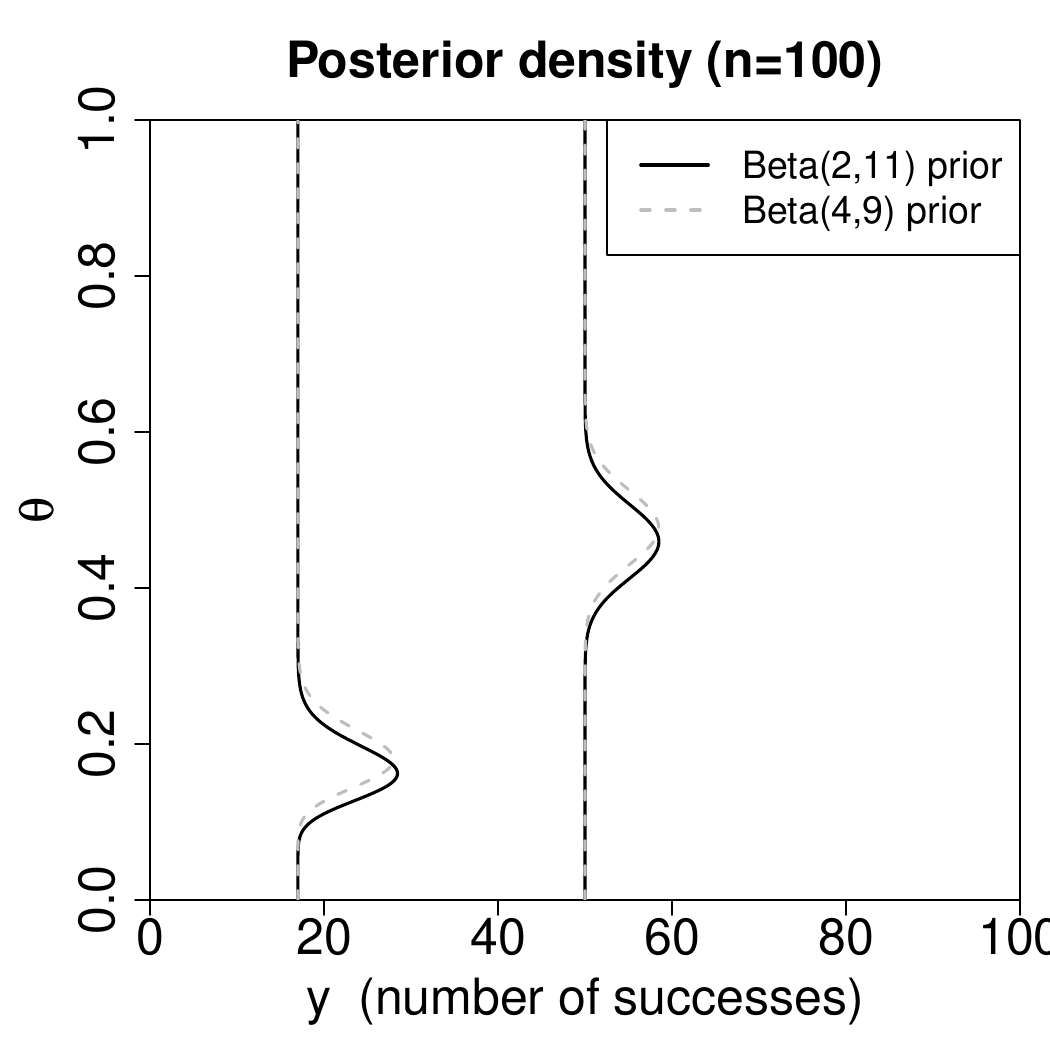}
\end{tabular}
\end{center}
\caption{Multiple Binomial experiments are conducted with $n=30$ (top) and $n=100$ (bottom). The left panels display two prior distributions and the Binomial distribution for success probabilities $\theta=0.1$ and $\theta=0.35$. The right panels display the posterior distribution from experiments with $y=5$ and $y=15$ (top), and $y=17$ and $y=50$ (bottom)}
\label{fig:frequentist_bayes}
\end{figure}

We start by discussing a purely frequentist interpretation of Bayesian inference.
The interpretation holds exactly when the prior matches a certain data-generating truth, and asymptotically under certain settings.
We discuss how this motivates using empirical Bayes to try to improve the frequentist calibration of fully Bayesian inference.
We end by mentioning pitfalls of empirical Bayes and how to ameliorate them.

Interpreting Bayesian inference from a frequentist standpoint may be easiest with an example.
Consider a setting where repeated experiments are conducted, with data from experiment $b$ generated as $\by^{(b)} \sim p(\by^{(b)} \mid \btheta^{(b)})$, and where the data-generating parameters $\btheta^{(b)}$ are drawn from some distribution $p(\btheta)$.
The usual frequentist setting envisions repeated studies with the same parameter $\btheta^{(1)}=\btheta^{(2)}=\ldots$, whereas we envision a frequentist setting where each experiment has its own data-generating $\btheta^{(b)} \sim p(\btheta)$.
In a Bayesian analysis, one also specifies a prior $\pi(\btheta)$, which in general will not match $p(\btheta)$.
For example, a pharmaceutical company conducts repeated experiments over time to assess the efficacy of drugs $b=1,2,\ldots$.
One observes the number of cured patients $y^{(b)} \sim \mbox{Beta}(n^{(b)}, \theta^{(b)})$ and the parameter of interest is the cure probability $\theta^{(b)}$.
Suppose that, unknown to the company, truly $\theta^{(b)} \sim \mbox{Beta}(2,11)$ independently across $b$, denoted by $p(\theta)$ earlier.
Figure \ref{fig:frequentist_bayes} shows an illustration for $n=30$ and $n=100$.
For simplicity, it assumes that $n^{(1)}= n^{(2)}= \ldots = n$, and that the experiments are independent given $\theta^{(b)}$.
The left panels show the $\mbox{Beta}(2,11)$ density, two realizations $\theta^{(b)}=0.1$ and $0.35$, and the corresponding Binomial probability mass functions.
Suppose that one sets independent  $\mbox{Beta}(a, l)$  priors to each $\theta^{(b)}$, denoted $\pi(\theta^{(b)})$ earlier, for some  $a,l>0$. 
The right panels show the posterior distributions $\pi(\theta^{(b)} \mid y^{(b)})$ for two putative $y^{(b)}$ if the data analyst were to use the $\mbox{Beta}(2,11)$ prior (black solid lines) matching the data-generating truth. These posteriors are the frequentist distribution of the $\theta^{(b)}$'s conditional on the observed $y^{(b)}$, that is $\pi(\theta^{(b)} \mid y^{(b)}) = p(\theta^{(b)} \mid y^{(b)})$ where $p(\theta^{(b)} \mid y^{(b)}) \propto p(y^{(b)} \mid \theta^{(b)}) p(\theta^{(b)})$. No Bayesian interpretation is needed: $(\theta^{(b)},y^{(b)})$ have a frequentist distribution, $p(\theta^{(b)} \mid y^{(b)})$ is a well-defined frequentist object,
and posterior inference is well-calibrated from a frequentist standpoint.

Figure \ref{fig:frequentist_bayes} also shows another posterior that only has a Bayesian interpretation, obtained by setting the prior $\pi(\theta^{(b)})$ to a $\mbox{Beta}(4,9)$ (dashed grey lines).
The new posterior no longer has a direct frequentist interpretation under the data-generating $p(\btheta)$,
but as $n$ grows (bottom panels) both posteriors converge to the same answer by virtue of the Bernstein-von Mises theorem (\cite{vandervaart:1998}, Theorem 10.1).
That is, even if one sets the wrong prior $\pi(\theta^{(b)})$ (in a frequentist sense, that is not matching $p(\theta^{(b)})$), when certain limiting results apply the posterior asymptotically has a valid frequentist interpretation.
However, such limiting theorems typically apply to finite-dimensional models, and high-dimensional settings must be handled with care.

The empirical Bayes framework \citep{robbins:1992}, replaces a fully Bayesian prior $\pi(\btheta \mid \bomega)$ 
by a data-based $\pi(\btheta \mid \hat{\bomega})$,
where $\hat{\bomega}$ is usually obtained via  maximizing the  marginal likelihood (see Sections \ref{sec:framework} and \ref{sec:computation}).
The hope is that $\pi(\btheta \mid \hat{\bomega})$ is close to the data-generating $p(\btheta)$ and hence has a better frequentist calibration than $\pi(\btheta \mid \bomega)$ for a possibly poorly chosen $\bomega$.
For a review of empirical Bayes in estimation see \cite{carlin:2000} (Chapter 3)
and for  theory see \cite{butucea:2009,efron:2016}.
As mentioned, here we are more interested in data integration and structural learning than in estimation. 

Under suitable assumptions, if $\pi(\btheta \mid \bomega)$ is a parametric prior, then, asymptotically, $\hat{\bomega}$ maximizes the prior density of the data-generating $\btheta^*$ 
\citep{petrone:2014}. See \cite{rousseau:2017} for non-parametric extensions.
A natural question is whether a fully Bayesian framework shares such a frequentist interpretation. The answer is that not always.
In some settings the empirical Bayes posterior $\pi(\btheta \mid \by, \hat{\bomega})$ asymptotically matches $\pi(\btheta \mid \by)$ from a fully Bayesian analysis \citep{petrone:2014}, extending the Bernstein-von Mises argument given above, whereas in others empirical and full Bayes cannot be reconciled.
In Section \ref{sec:framework} we discuss intuitively why empirical and full Bayes may give different answers in our data integration setting, in Section \ref{sec:theory} we give precise conditions when this occurs, and in Section \ref{sec:results} we show via examples that the differences are usually not large but that they can be non-negligible.

We remark that there are pitfalls to empirical Bayes.
A first concern for Bayesians is that, by setting a data-dependent $\pi(\btheta \mid \hat{\bomega})$, Bayesian updating is no longer coherent. Coherence refers to the standard desideratum that, if one observes the data $\by=(\by_1,\by_2)$ as two separate batches, the posterior $\pi(\btheta \mid \by) \propto p(\by \mid \btheta) \pi(\btheta)$ is equal to the posterior under the sequential updating
$\pi(\btheta \mid \by) \propto p(\by_2 \mid \btheta) \pi(\btheta \mid \by_1)$.
If $\pi(\btheta)$ is replaced by $\pi(\btheta \mid \hat{\bomega})$, where  $\hat{\bomega}$  depends on $\by$, such coherence no longer holds.
The second concern is more critical in that it also applies to non-Bayesians. The empirical Bayes estimate $\hat{\bomega}$ may collapse onto degenerate solutions, e.g. where $\pi(\btheta \mid \hat{\bomega})$ assigns prior probability 0 or 1 to a parameter being zero, resulting in poor inference.
We discuss practical solutions to such degeneracies in Section \ref{sec:framework}.

\section{Framework} \label{sec:framework}

We lay out an empirical Bayes framework for data integration via meta-covariates, akin to \cite{vandewiel:2019}.
Consider a setting where the main data records an outcome $\by=(y_1,\ldots,y_n)^T \in \mathbb{R}^n$ for $n$ individuals 
and $p$ covariates $\bx_i \in \mathbb{R}^p$ for $i=1,\ldots,n$. 
For each covariate $j=1,\ldots,p$, one also has a vector of $q$ meta-covariates $\bz_j \in \mathbb{R}^q$.
We denote by $\bX$ the $n \times p$ covariates matrix with $i^{th}$ row equal to $\bx_i$, and by $\bZ$ the $p \times q$ meta-covariates matrix with $\bz_j$ as its $j^{th}$ row.
The framework can be extended to other settings such as covariance models \citep{jewson:2023}, but for simplicity we focus on regression.

Before proceeding,  we discuss examples to make ideas concrete.
Suppose that one wants to study the association between colon cancer survival ($y_i$) and gene expression ($\bx_i$), and one has three lists of genes found to be associated with gastrointestinal, endometrial and ovarian cancer.
Here, $\bz_j \in \{0,1\}^3$ indicates whether gene $j$ is included in these 3 lists. 
The meta-covariates $\bz_j$ may contain further information:
biological pathways in which each gene features,
whether the covariate is a clinical or genomic variable, etc.
As another example, a school's director wants to identify covariates $\bX$ associated with academic performance $\by$. The director has reports $\bZ$ on estimated covariate effects in other schools, but does not have the full data from those schools.

\subsection{Empirical vs. full Bayes} \label{ssec:ebayes_vs_fullbayes}

The key idea is to set a prior
$\pi(\btheta \mid \bZ, \bomega)$ that depends on the meta-covariates $\bZ$ via hyper-parameters $\bomega$,
and where $\bomega$ is estimated via empirical Bayes.
For simplicity, we denote $\pi(\btheta \mid \bomega) := \pi(\btheta \mid \bZ, \bomega)$.
We are interested in structural learning where one considers many different models for $\btheta$.
We denote such models generically by $\bgamma \in \Gamma$, where $\Gamma$ is the set of models.
For example, in variable selection we have $\btheta=(\theta_1,\ldots,\theta_p)^T$ regression parameters
and the models are indexed by $\bgamma=(\gamma_1,\ldots,\gamma_p) \in \{0,1\}^p$, where $\gamma_j= \mbox{I}(\theta_j \neq 0)$ is an inclusion indicator for covariate $j=1,\ldots,p$.
One may also have parameters that are present in all models, such as the error variance or covariates not undergoing selection, but we omit these from the notation for simplicity.

Formally, each model $\bgamma$ has its own likelihood $p(\by \mid \btheta, \bgamma)$, and the joint Bayesian model is
$
p(\by, \btheta, \bgamma)= p(\by \mid \btheta, \bgamma) \pi(\btheta \mid \bgamma) \pi(\bgamma),
$
where $\pi(\bgamma)$ is a given model prior.
Using Bayes formula, posterior model probabilities are $p(\bgamma \mid \by) \propto p(\by \mid \bgamma) \pi(\bgamma)$, where
\begin{align}
p(\by \mid \bgamma)= \int p(\by \mid \btheta, \bgamma) d\Pi(\btheta \mid \bgamma) 
\nonumber
\end{align}
is the marginal likelihood of model $\bgamma$, and $\Pi(\btheta \mid \bgamma)$ the prior measure on $\btheta$ under $\bgamma$.

To achieve data integration within full Bayes, one would replace $\pi(\bgamma)$ by some $\pi(\bgamma \mid \bomega)$ that depends on the meta-covariates $\bZ$, and then set a hyper-prior $\pi(\bomega)$.
Although much of our discussion and some of our results are more general, to fix ideas consider $\pi(\bgamma \mid \bomega)= \prod_{j=1}^p \mbox{Bern}(\gamma_j; m_j(\bomega))$ for a given differentiable $m_j(\bomega)$, such as the inverse logit function $m_j(\bomega)= 1 / (1 + e^{-\bz_j^T \bomega})$. In what follows, we indicate where this prior structure is explicitly assumed.
The issue is that such a fully Bayesian framework gives posterior probabilities $\pi(\bgamma \mid \by) \propto p(\by \mid \bgamma) \pi(\bgamma)$, where
\begin{align}
 \pi(\bgamma)= \int \pi(\bgamma \mid \bomega) d\Pi(\bomega).
\nonumber
\end{align}
By definition $\pi(\bgamma)$ does not depend on $\by$, it is a quantity defined a priori,  and hence by definition it is not data-adaptive.   
That is, although $\pi(\bgamma)$ depends on meta-covariates, this dependence is set entirely a priori.
 This statement may sound surprising, given the popular practice of placing priors on hyper-parameters such as $\bomega$, with the hope of learning $\bomega$ and guiding posterior inference via $\pi(\bomega \mid \by)$. Specifically, since
\begin{align}
 \pi(\bgamma \mid \by) = \int \pi(\bgamma \mid \by, \bomega) \pi(\bomega \mid \by) d\bomega
\label{eq:pp_fullbayes}
\end{align}
where $\pi(\bomega \mid \by)$ depends on $\by$, it should follow that $\pi(\bgamma \mid \by)$ adapts to learning $\bomega$ from $\by$. 
This interpretation of \eqref{eq:pp_fullbayes} is misleading.
As noted 
$\pi(\bgamma \mid \by) \propto p(\by \mid \bgamma) \pi(\bgamma)$
where $\pi(\bgamma)$ does not depend on $\by$, hence $\pi(\bgamma \mid \by)$ can only adapt to $\by$ via $p(\by \mid \bgamma)$.
To reconcile this observation with \eqref{eq:pp_fullbayes}, using that
$\pi(\bgamma \mid \by, \bomega)= p(\by \mid \bgamma, \bomega) \pi(\bgamma \mid \bomega) / p(\by \mid \bomega)$
and that $\pi(\bomega \mid \by)= p(\by \mid \bomega) \pi(\bomega) / p(\by)$,
re-write \eqref{eq:pp_fullbayes} as $\pi(\bgamma \mid \by) =$
\begin{align}
\int \frac{p(\by \mid \bgamma, \bomega) \pi(\bgamma \mid \bomega)}{p(\by \mid \bomega)}
\frac{p(\by \mid \bomega) \pi(\bomega)}{p(\by)} d\bomega
=
\frac{\int p(\by \mid \bgamma, \bomega) \pi(\bgamma \mid \bomega) \pi(\bomega) d\bomega}{p(\by)}
= \frac{p(\by \mid \bgamma) \pi(\bgamma)}{p(\by)}
\label{eq:pp_fullbayes2}
\end{align}
since $p(\by \mid \bgamma, \bomega)= p(\by \mid \bgamma)$.
Critically, in the left-hand side of \eqref{eq:pp_fullbayes2} the term $p(\by \mid \bomega)$ coming from $\pi(\bomega \mid \by)$, which is how one learns hyper-parameters from $\by$, cancels with $p(\by \mid \bomega)$ coming from $\pi(\bgamma \mid \bomega, \by)$.
That is, is true that $\pi(\bomega \mid \by)$ learns about $\bomega$ from $\by$, but the effect of such learning is cancelled in $\pi(\bgamma \mid \by)$ because $p(\by \mid \bomega)$ also features in $\pi(\bgamma \mid \bomega, \by)$.
These arguments are not restricted to Bayesian model selection, they apply to any hierarchical model where one sets a prior on hyper-parameters.

Here we consider an empirical Bayes framework where one sets $\pi(\bgamma \mid \hat{\bomega})$ that depends on $\bZ$ and where $\hat{\bomega}$ is learned from $\by$.
The empirical and full Bayes solutions may not agree even as $n \to \infty$ when $\pi(\bgamma \mid \hat{\bomega})$ differs sufficiently from a fully Bayesian $\pi(\bgamma)$, e.g, in settings where the (marginal) likelihood $p(\by \mid \bgamma)$ does not fully dictate inference asymptotically (e.g. high-dimensional settings).
 We do not argue that fully Bayesian frameworks lead to poor inference (quite the contrary).
Further, hierarchical priors $\pi(\bgamma \mid \bomega) \pi(\bomega)$ are useful to induce prior dependence in $\bgamma$, which is often helpful in applications. Our point is that empirical Bayes' plug-in $\pi(\bgamma \mid \hat{\bomega})$ can behave like a data-adaptive penalty, and in high dimensions the difference between integrating out $\bomega$ and plugging in $\hat{\bomega}$ can matter.

We remark that in the framework above only the model prior $\pi(\bgamma \mid \bomega)$ depends on $\bZ$, but it is also possible to consider a prior on parameters $\pi(\btheta \mid \bgamma, \bomega)$ that depends on $\bomega$.
For example, \cite{jewson:2023} considered a spike-and-slab graphical model where the slab's mean and variance depend on $\bZ$.
The theory of Section \ref{sec:theory} can be extended to such settings, but the computational algorithms in Section \ref{sec:computation} get more involved, and in particular some of the simplified gradients given there no longer apply.

\subsection{Hyper-parameter estimation in empirical Bayes} \label{ssec:ebayes_hyperpar}

A possible strategy to obtain $\hat{\bomega}$ is to maximize the log-marginal likelihood
\begin{align}
\log p(\by \mid \bomega)=
\log \left( \sum_{\gamma \in \Gamma} p(\by \mid \bgamma, \bomega) \pi(\bgamma \mid \bomega) \right).
\nonumber
\end{align}
However, in our setting maximizing $p(\by \mid \bomega)$ may lead to degenerate priors,
e.g., $\hat{\bomega}$ such that prior inclusion probabilities $\pi(\gamma_j=1 \mid \by, \hat{\bomega})$ are either 1 or 0 for all covariates $j$, see Section \ref{ssec:comp_gradientbased}.
We discuss three practical fixes. First, one may constrain the support of $\bomega$ to exclude the set of values $\overline{\Omega}$ that result in a degenerate prior. For example, under $\pi(\bgamma \mid \bomega)= \prod_{j=1}^p \mbox{Bern}(\gamma_j; 1/(1+e^{-\bz_j^T \bomega})$, one could define $\overline{\Omega}$ such that $1/(1+e^{-\bz_j^T \bomega}) \in [0.001,0.999]$ for all $\bomega \not\in \overline{\Omega}$ and all $j$. Second, one can set a minimally informative prior $\pi(\bomega)$ that assigns zero prior probability to $\overline{\Omega}$, and define $\hat{\bomega}$ to be the posterior mode maximizing $\log p(\by \mid \bomega) + \log \pi(\bomega)$, akin to \cite{gu_mengyang:2018}.
Third, one can consider a two-step solution where one does not fully maximize $\log p(\by \mid \bomega)$,  in a manner such that $\hat{\bomega} \not\in \overline{\Omega}$. 
The theory in Section \ref{sec:theory} discusses such a two-step solution, which we use there because it facilitates theoretical analysis.
The results therein show that, for the particular case where one has a single meta-covariate and a particular initial choice of $\bomega$, the two-step solution asymptotically avoids degeneracies, that is $\hat{\bomega} \not\in \overline{\Omega}$. The two-step strategy is in principle applicable beyond single meta-covariate settings: one obtains an initial $\hat{\bomega}^{(0)} \not\in \overline{\Omega}$ and takes a single optimization step in a manner that ensures that the updated $\hat{\bomega} \not\in \overline{\Omega}$. In practice, we advocate for setting minimally informative $\pi(\bomega)$ that prevents $\hat{\bomega}$ from taking extreme values, also for finite $n$.
We next discuss such a prior, based on the idea that prior inclusion probabilities should not be too close to 0 or 1, and which we use for our examples in Section \ref{sec:results}.

Suppose that $\pi(\bgamma \mid \bomega)= \prod_{j=1}^p \pi(\gamma_j \mid \bomega)$ and that $\pi(\gamma_j = 1 \mid \bomega)= m_j(\bomega)$ for some invertible function $m_j$.
For example, the inverse logit function $m_j(\bomega)= 1 / (1 + e^{-\bz_j^T \bomega})$ is interpretable and leads to some computational simplifications (Section \ref{sec:computation}).
Consider the prior $\bomega \sim N( {\bf 0}, g_\omega \bV)$ for some known $\bV$ where,
taking inspiration from the popular Zellner prior, by default we set $\bV= (\bZ^T \bZ / p)^{-1}$.
The scaling by $p$ facilitates interpretability: if the columns in $\bZ$ (other than the intercept) have zero mean, then 
$\bV$ is the sample precision (inverse covariance) matrix of $\bZ$ multiplied by $g_\omega$.

We next discuss how to set a minimally informative default for $g_\omega$.
The idea is to set $g_\omega$ such that $\pi(\bomega)$ is unlikely to generate a value of $\bomega$ that leads to prior inclusion probabilities that are very close to 0 or 1.
In more detail, note that the prior inclusion probabilities $m_j(\bomega)$ are random, since they depend on $\bomega$.
We propose setting $g_\omega$ such that 
the prior probability $P \left( m_j(\bomega) \in [0.001, 0.999] \right) \geq 0.95$ for all $j$.
If $m_j$ is the inverse logit function, we seek $g_\omega$ such that, for all $j$,
\begin{align}
0.95 \leq P \left( m_j(\bomega) \in [0.001, 0.999] \right)
\Rightarrow
 g_\omega \leq \frac{1}{v_j} \left( \frac{\log(0.001/0.999)}{\Phi^{-1}(0.05 / 2)} \right)^2,
\nonumber
\end{align}
where $\Phi^{-1}$ is the standard normal quantile function, $v_j= \bz_j^T \bV \bz_j$, and we used that $\bz_j^T \bomega \sim N(0, g_\omega v_j)$.
We hence take $g_\omega$ by replacing $v_j$ above by $\max_{j=1}^p v_j$.

\section{Theory for variable selection consistency} \label{sec:theory}

We consider a simplified setting where one has a high-dimensional linear regression $\by \sim N(\bX \btheta, \phi I)$,
where $\phi > 0$ is the error variance,
and a single discrete meta-covariate $z_j \in \{1,\ldots,B\}$ that splits covariates $j=1,\ldots,p$ into $B$ blocks.
Letting $\gamma_j = \mbox{I}(\theta_j \neq 0)$ as before, we consider a model prior
$\pi(\bgamma \mid \bomega) = \prod_{j=1}^p \mbox{Bern}(\gamma_j; \omega_{z_j})$,
where $\bomega= (\omega_1, \ldots, \omega_B) \in [0,1]^B$ are prior inclusion probabilities for each block.
The number of covariates $p$ can grow with $n$.
To ease the exposition and proofs, we assume that $\phi$ is known and that one sets Zellner's prior
\begin{align}
 \pi(\btheta_{\bgamma} \mid \bgamma)= N(\btheta_{\bgamma}; {\bf 0}, g \phi (\bX_{\bgamma}^T \bX_{\bgamma}/n)^{-1}),
\nonumber
\end{align}
where $\btheta_{\bgamma}= \{ \theta_j : \gamma_j=1 \}$ are the non-zero parameters under model $\bgamma$, 
$\bX_{\bgamma}$ the corresponding columns of $\bX$ and $g > 0$ a given prior dispersion.
 By default, we recommend $g=1$, mimicking Zellner's unit information prior, a popular prior that is tightly connected to the Bayesian information criterion \citep{kass:1995}. In all our examples we used $g=1$, and results were insensitive to varying $g \in [0.1,10]$.  
Our results can be extended to unknown $\phi$ and to other Gaussian priors by placing eigenvalue conditions on $\bX_{\bgamma}^T \bX_{\bgamma}$, for example those listed in Sections 2.3-2.4 of \cite{rossell:2022}.
In particular, one could consider a block-dependent prior dispersion, but we do not do this here for brevity and because its main role would be to induce block-dependent sparsity, which is already achieved by $\bomega$.
 It would also be possible to assign a prior to $g$, leading to a new (thicker-tailed) marginal prior $\tilde{\pi}(\btheta_{\bgamma} \mid \bgamma)$. While our results should remain unaffected, subject to mild conditions, we focus on Zellner's prior to streamline the proof. Gaussian priors are also computationally convenient in that marginal likelihoods have a closed-form for models with Gaussian errors.  Throughout, we assume that one constrains attention to models such that $\bX_{\bgamma}$ has full column rank for any $\bgamma \in \Gamma$, and that the true model lies in the set of considered models.

We provide two main results. Theorem \ref{thm:suffcondlinearmodel} in Section \ref{ssec:consistency_dataintegration} 
states that, when $\bomega$ satisfies certain conditions, the posterior distribution on models $\pi(\bgamma \mid \by, \bomega)$ concentrates on $\bgamma^*=(\gamma_1^*, \ldots, \gamma_p^*)$, where $\gamma_j^* = \mbox{I}(\theta_j^* \neq 0)$ and $\btheta^*$ is the data-generating parameter.
Theorem \ref{thm:suffcondlinearmodel} holds under milder conditions than those required when setting equal prior inclusion probabilities in all blocks $\omega_1=\ldots = \omega_B$, 
or when using the Beta-Binomial model prior  of \cite{scott:2010}  (which treats all blocks equally)  or any standard (non block-dependent) $L_0$ criterion.
The interest of Theorem \ref{thm:suffcondlinearmodel} is that it gives the best possible consistency result attainable using the meta-covariates $\bZ$, if an oracle were to set $\bomega$ to an ideal value.
Theorem \ref{theo:empbayesselconsist} in Section \ref{ssec:consistency_ebayes} shows that, under slightly stronger conditions, one can obtain $\hat{\bomega}$ that also leads to consistent model selection under milder conditions than under $\omega_1=\ldots = \omega_B$, a Beta-Binomial model prior, and standard $L_0$ criteria.
 Our results do not assume that the blocks are informative: the blocks could be useless, in that the true sparsity and signal strength are identical across blocks.
Theorems \ref{thm:suffcondlinearmodel}-\ref{theo:empbayesselconsist} guarantee consistent model selection even in this setting, and Section \ref{ssec:simstudy} includes an empirical illustration (Scenario 3). The key is that if the $B>1$ blocks are informative then Conditions A2-A3 required by Theorem \ref{thm:suffcondlinearmodel} become milder than in the $B=1$ case, and similarly for Conditions A4-A6 in Theorem \ref{theo:empbayesselconsist}. If, on the contrary, blocks are uninformative then these conditions become equivalent to the $B=1$ case.

\subsection{Improved consistency with data integration}
\label{ssec:consistency_dataintegration}

We denote by $p_b= \sum_{j=1}^p \mbox{I}(z_j = b)$ the number of variables in block $b$,
by $p_{\bgamma,b}= \sum_{j=1}^p \gamma_j \mbox{I}(z_j=b)$ the number of those selected by $\bgamma$,
and by $s_b= \sum_{j=1}^p \gamma_j^* \mbox{I}(z_j=b)$ the number of truly active variables in block $b$.
We denote by $\bX_{\bgamma}$ the subset of columns of $\bX$ selected by $\bgamma$ (i.e., such that $\gamma_j=1$),
and $\bX_{\bgamma^* \setminus \bgamma}$ those selected by $\bgamma^*$ but not $\bgamma$.
Under the priors set above, simple algebra shows that
\begin{align}
\frac{\pi(\bgamma^* \mid \by,\bomega)}{\pi(\bgamma \mid \by, \bomega)}=
\exp \left\{ -\frac{g n W_{\bgamma \bgamma^*}}{2\phi (1+g n)} \right\}
\prod_{b=1}^B \left( \frac{\omega_b}{(1 + gn)^{1/2} (1 - \omega_b)} \right)^{p_{\bgamma^*,b} - p_{\bgamma,b}}
\label{eq:bf_zellner_known}
\end{align}
where $W_{\bgamma \bgamma^*}=\hat{\btheta}_{\bgamma}^T \bX_{\bgamma}^T \bX_{\bgamma}\hat{\btheta}_{\bgamma} - \hat{\btheta}_{\bgamma^*}^T \bX_{\bgamma^*}^T \bX_{\bgamma^*}\hat{\btheta}_{\bgamma^*}$
and $\hat{\btheta}_{\bgamma}= (\bX_{\bgamma}^T \bX_{\bgamma})^{-1} \bX_{\bgamma}^T \by$ is the least-squares estimate.
The first term on the right-hand side in \eqref{eq:bf_zellner_known} compares the sum of squared residuals between $\bgamma$ and $\bgamma^*$,
and hence rewards goodness-of-fit, whereas the second term is a complexity penalty.
Its logarithm is $\sum_{b=1}^B (p_{\bgamma,b} - p_{\bgamma^*,b}) \kappa_b$, where
\begin{align}
\kappa_b= \frac{1}{2} \log(1 + gn) + \log (1 / \omega_b - 1)
\label{eq:kappa}
\end{align}
plays a key role in our conditions, which we now state.  Recall that by default we set $g=1$, but our upcoming results are essentially unaffected as long as $|\log g|=o(\log n)$. 
\begin{enumerate}
\item [{\bf A1.}] The number of blocks $B$ is constant.

\item [{\bf A2.}] For each block $b$, there exists $f_b\to \infty$ (as $n\to \infty$) such that for sufficiently large $n$,
$\kappa_b \;=\;\log(p_b-s_b) + f_b$.

\item [{\bf A3.}] For each block $b$, there exists $g_b\to \infty$ such that for sufficiently large $n$,
    \begin{equation*}
        \sqrt{\frac{(1-\nu) n\phi^{-1}\rho(\bX)}{6}}{\theta_{\min,b}^*} \,- \,\sqrt{\kappa_b} \;=\; \sqrt{\log(s_b)} + g_b ,
    \end{equation*}
where $\theta_{\min,b}^*= \min_{z_j = b, \theta_j^* \neq 0} |\theta_j^*|$,
$\nu:=\tfrac12(1+\max_b\log(p_b-s_b)/\kappa_b)\in (\tfrac12,1)$,
 \begin{equation}\label{eq:rhoX}
        \rho(\bX)=  \min_{\bgamma: \bgamma \not \supseteq \bgamma^*} \lambda_{\min}\big(\tfrac{1}{n} \bX_{\bgamma^* \setminus \bgamma}^{\top}\left(I - \bP_{\bgamma} \right) \bX_{\bgamma^* \setminus \bgamma}\big),
    \end{equation}
$\lambda_{\min}$ denotes the smallest eigenvalue,
and $\bP_{\bgamma}= \bX_{\bgamma} (\bX_{\bgamma}^T \bX_{\bgamma})^{-1} \bX_{\bgamma}^T$ is the projection matrix onto the column span of $\bX_{\bgamma}$.
\end{enumerate}

Assumption A2 requires $\kappa_b$ to be large enough, that is that the prior penalizes sufficiently complex models.
A3 is a betamin condition on the smallest non-zero parameter $\theta_{\min,b}^*$ in block $b$.
The quantity $\rho(\bX)$ is non-negative and relates to how distinguishable non-overfitted models $\bgamma$ are from $\bgamma^*$.  More specifically, $\tfrac{1}{n} \bX_{\bgamma^* \setminus \bgamma}^T \left(I- P_{\bgamma} \right) \bX_{\bgamma^* \setminus \bgamma}$ is the sample covariance of the residuals when regressing $\bX_{\bgamma^* \setminus \bgamma}$ on $\bX_{\bgamma}$,
and $\rho(\bX)=1$ in a simpler orthonormal case where $\bX^T\bX=n I$.
Assumptions A2-A3 are fairly minimal.
For orthonormal designs A2-A3 are similar to Assumptions A4-A5 in \cite{rognon:2025}, which the authors showed to be near-necessary. 
Also, when applied to the $B=1$ blocks setting, these conditions are very similar (and in some settings, slightly weaker) than those required by $L_0$ criteria that know the true model size analysed in \cite{wainwright:2009information}. 
Further, \cite{rognon:2025} (Sections 4.3 and S7) showed that A2-A3 nearly match necessary conditions for consistent model recovery in a wide range of regimes, and that A2-A3 are milder than those required when one uses standard $L_0$ criteria or a Beta-Binomial model prior,  that is the $B=1$ case penalizing all blocks equally. In particular, the dependence of A3 on $\rho(\bX)$ (which measures covariate collinearity) becomes slightly milder for $B>1$. 

\begin{theorem}\label{thm:suffcondlinearmodel}
If Assumption A1 holds and $\kappa_b$ implied by $(\bomega,g)$ satisfies A2 and A3,
then $\lim_{n \to \infty} E \left[ \pi(\bgamma^* \mid \by, \bomega) \right] =1$.
\end{theorem}

Theorem \ref{thm:suffcondlinearmodel} shows that $\pi(\bgamma^* \mid \by, \bomega)$ converges to 1 in the $L_1$ sense.
By Proposition 1 in \cite{rossell:2022}, such $L_1$ convergence rates bound the frequentist probability of a wrong model selection, as well as type I and II error probabilities.
 The proof of Theorem \ref{thm:suffcondlinearmodel} bounds the convergence rate for $E[\pi(\bgamma^* \mid \by, \bomega)]$ as a function of the number of truly non-zero and zero parameters $(s_b, p_b-s_b)$ and the magnitude of covariate effects via $\theta_{\min,b}^*$ in each block. Using these, an oracle could set penalties $\kappa_b^*$ that depend on $(s_b,p_b, \theta_{\min,b}^*)$ and approximately optimize the convergence rate (Section \ref{ssec:oracle_conv_rate}).
Briefly, neglecting constants, $\kappa_b^*$ is of order
\begin{align}
 n\rho(\bX) \theta^*_{\min,b}/\phi \left( 1 + \frac{\ln({p_b/s_b-1})}{n \rho(\bX) \theta^*_{\min,b}/\phi} \right)^2
\nonumber
\end{align}
Our two-step procedure in Section \ref{ssec:consistency_ebayes} could estimate $(s_b, \theta_{\min,b}^*)$ in Step 1 and plug them into $\kappa_b^*$ to define $\kappa_b^{(1)}$ in Step 2, but we do not pursue this here for simplicity.


\subsection{Consistency for empirical Bayes}
\label{ssec:consistency_ebayes}

Theorem \ref{thm:suffcondlinearmodel} gives a range of $(\bomega,g)$ that lead to consistent model selection.
We now show that it is possible to obtain a data-based $\hat{\bomega}$ that also leads to consistent model selection.
We analyze the following two-step procedure.

\begin{enumerate}
\item Set $\hat{\omega}_b^{(0)} = 1 / (p + 1)$ for $b=1,\ldots,B$, so that 
$\kappa_b^{(0)}= \log(p) + \frac{1}{2} \log(1 + g n)$.

\item Set $\hat{\omega}_b^{(1)}= \frac{1}{p_b} \sum_{z_j=b} \pi(\gamma_j=1 \mid \by, \hat{\bomega}^{(0)})$,
and $\kappa_b^{(1)}= \log(1/\hat{\omega}_b^{(0)}  - 1) + \frac{1}{2} \log(1 + g n)$.
\end{enumerate}

This procedure can be seen as first taking $\hat{\omega}_b^{(0)}=1/(p+1)$ in all blocks, which is a sparse choice assuming a constant prior expected number of active covariates as $p$ grows.
If one sets $\hat{\omega}_b^{(0)}=c/(p+1)$ for constant $c$, our theory continues to hold with minor modifications.
In Step 2, one obtains $\hat{\omega}_b^{(1)}$ by taking the average posterior inclusion probability in block $b$, given the data and $\hat{\bomega}^{(0)}$. Formally, Step 2 is a fixed-point iteration aimed at setting $\nabla_{\bomega} \log p(\by \mid \bomega) = {\bf 0}$, see the discussion after Theorem \ref{thm:ebayes_zerograd} in Section \ref{ssec:comp_gradientbased}.

The consistency of the two-step procedure can be established under slightly stronger conditions than A2-A3, called A4-A6 below, but still milder than those attainable by a model prior that does not use the meta-covariates.
To discuss these assumptions, we introduce a notion of small, intermediate and large truly active signals $|\theta_j^*|/\sqrt{\phi}$ in each block,
denoted $S_b^S$, $S_b^I$, and $S_b^L$ respectively. 
Let $\bar{\lambda}$ be the largest eigenvalue of $\bX_{\bgamma^*}^T \bX_{\bgamma^*}$
and $\kappa_b$ as defined in \eqref{eq:kappa}, then
\begin{align}
    &S_b^{S}(\kappa_b) := \left\{j : z_j=b, \theta_j^* \neq 0, \sqrt{n \bar{\lambda}}|\theta_j^*| / \sqrt{\phi}=o\big(\sqrt{\kappa_b}\big) \right\}
\nonumber\\
    &S_b^{L}(\kappa_b) := \left\{j: z_j=b, \theta_j^* \neq 0, \sqrt{\frac{(1-\nu) n \rho(\bX)}{6 \phi}}|\theta_j^*| \,- \,\sqrt{\kappa_b} \;=\; \sqrt{\ln(s_b)} + g_b\right\}.
\nonumber\\
    &S_b^{I}(\kappa_b) := \left\{ j: z_j=b, \theta_j^* \neq 0  \right\} \setminus \big(S_b^{L}(\kappa_b) \cup S_b^{S}(\kappa_b)\big).
    \nonumber 
\end{align}



We are now ready to state the conditions.
Condition A4 imposes a mild bound on number of truly active covariates.
A5 is a slightly stronger version of A3, requiring that the number of large signals is not too small,
whereas A6 requires that the number of small signals is not too large.
Both conditions are stated in terms of $\kappa_b^{(0)}$ defined in Step 1.
Intuitively, if there were too many small signals or too few large signals relative to $\kappa_b^{(0)}$, then $\hat{\omega}_b^{(1)}$ would be a poor estimator of the proportion of truly active covariates in block $b$, and the posterior would not be guaranteed to concentrate on $\bgamma^*$.

\begin{enumerate}
\item [{\bf A4.}] The total number of truly active covariates $\sum_{b=1}^B s_b = o(\sqrt{1+gn})$.

\item [{\bf A5.}] For each block $b$, there exists $a_b\to \infty$ such that for sufficiently large $n$, 
\begin{equation*}\label{cond:suffIIreg2}
        \sqrt{\frac{ (1-\psi)n\phi^{-1}\rho(\bX)}{6}}{\theta_{\min,b}^*} - \sqrt{\log \bigg(\frac{p_b}{|S^{L}_b(\kappa_b^{(0)})|} - 1\bigg) + \frac12 \log(1+g n)} \;=\; \sqrt{\log(s_b)} + a_b,
    \end{equation*}
where $\psi=\tfrac12\Big(1+\max_b \log(p_b-s_b)/\big(\log(p_b/s_b-1)+\log(1+gn)/2\big)\Big)$.

\item [{\bf A6.}] For each block $b$, it holds that  $|S^{I}_b(\kappa_b^{(0)})|=O(1)$ and $|S^{S}_b(\kappa_b^{(0)})| = O(p_b - s_b)$.
\end{enumerate}

\begin{theorem}\label{theo:empbayesselconsist}
If Assumptions A1, A4, A5 and A6 hold,
then $\lim_{n \to \infty} E \left[ \pi(\bgamma^* \mid \by, \hat{\bomega}^{(1)}) \right] =1$.
\end{theorem}


\section{Computational framework} \label{sec:computation}

We discuss computational algorithms to obtain the empirical Bayes estimate 
$\hat{\bomega}= \arg\max_{\bomega} \log \pi(\bomega \mid \by)=\arg\max_{\bomega} \log p(\by \mid \bomega) + \log \pi(\bomega)=$
\begin{align}
\arg\max_{\bomega} \log \left( \sum_{\bgamma} p(\by \mid \bgamma) \pi(\bgamma \mid \bomega) \right) -\frac{1}{2 g_\omega} \bomega^T \bV^{-1} \bomega,
\label{eq:ebayes_objective_regression}
\end{align}
where $(g_\omega,\bV)$ are as in Section \ref{ssec:ebayes_hyperpar}. 
Directly maximizing \eqref{eq:ebayes_objective_regression} poses a computational challenge, 
because it involves a prohibitive sum over all models $\bgamma$, e.g. $2^p$ in regression.
We next discuss some strategies that render this problem tractable.

In Section \ref{ssec:comp_gradientbased} we provide results
showing that maximizing \eqref{eq:ebayes_objective_regression} is amenable to stochastic gradient methods, as long as one can quickly evaluate $p(\by \mid \bgamma)$ for any given model $\bgamma$.
We remark that $p(\by \mid \bgamma)$ is only available in closed-form for specific models such as linear or additive regression with Gaussian outcomes. 
However, in practice one can use Laplace approximations and recent extensions like the approximate Laplace approximation \citep{kass:1990,rossell:2021,rossell:2021b}, which allow to quickly approximate $p(\by \mid \bgamma)$ in popular models such as generalized linear and generalized additive models.
Hence, the results in Section \ref{ssec:comp_gradientbased} can be useful for a relatively wide model class.
Section \ref{ssec:em_algorithm} extends the results of Section \ref{ssec:comp_gradientbased} to an EM algorithm
that eases somewhat the computation of stochastic gradients in Section \ref{ssec:comp_gradientbased}.
In Section \ref{ssec:other_comput_methods} we briefly discuss how 
to obtain a full posterior $\pi(\bomega \mid \by)$, beyond just a point estimate $\hat{\bomega}$,
and strategies for the case where $p(\by \mid \bgamma)$ cannot be easily computed.

We remark that $\log \pi(\bomega \mid \by)$ may be multi-modal, hence gradient-based methods find a local mode,
whereas sampling $\pi(\bomega \mid \by)$ may be more effective at exploring multiple modes.

In our software we initialize $\hat{\bomega}^{(0)}$ from a least-squares regression of $\log(\hat{\bpi}^{BB}/[1-\hat{\bpi}^{BB}])$ on $\bZ$,
where $\hat{\pi}_j^{BB}= \pi(\gamma_j=1 \mid \by)$ are marginal posterior inclusion probabilities under a Beta-Binomial model prior, estimated using a relatively small number of iterations $M$ (by default, $M=1000$).
In our examples of Section~\ref{sec:results}, such $\hat{\bomega}^{(0)}$ lead to virtually identical results to $\hat{\bomega}^{(0)}= {\bf 0}$, and to using $M=100$ (Table \ref{tab:dif_omega} in Appendix~\ref{sec:runtimes}).
Also, our approach had a similar computational cost than running an MCMC under standard model priors such as the Beta-Binomial. 
In both cases the cost scales well with $p$ under sparsity conditions, otherwise it quickly becomes unmanageable as $p$ grows. Under sparsity conditions, \cite{rossell:2017} and \cite{zhou_quan:2022} run an MCMC for $p > 10^4$ and $p> 10^5$ in a few minutes, respectively.
See Appendix \ref{sec:runtimes} for a comparison of run times and further discussion.

\subsection{Gradient-based methods}
\label{ssec:comp_gradientbased}

Our main results are as follows.
First, Theorem \ref{thm:ebayes_zerograd} shows that under fairly general conditions the gradient of \eqref{eq:ebayes_objective_regression} can be evaluated as a sum over only $p$ terms.
We then derive an EM algorithm to obtain a local mode for \eqref{eq:ebayes_objective_regression}
that also only requires sums over $p$ terms (Corollary \ref{cor:em_zerograd}), and brings some further simplifications.
Finally, Proposition \ref{prop:mstep_fullrank} shows that in an important particular case where the meta-covariates $\bZ$ partition the parameters into blocks, the M-step is available essentially in closed-form.
For example, the partition could be given by a list of covariates found to be relevant in a prior study, or by clustering the covariates into blocks using $\bZ$.
More specifically, the main implication of Theorem \ref{thm:ebayes_zerograd} is that, if $\pi(\bgamma \mid \bomega)$ factors,
then $\nabla_{\bomega} \log p(\by \mid \bomega)$ can be evaluated as a sum over $p$ terms.
Theorem \ref{thm:ebayes_zerograd} is not limited to $\pi(\gamma_j=1 \mid \bomega)$ given by the logit function:
one may parameterize $\pi(\gamma_j=1 \mid \bomega)$ using an arbitrary function $m_j()$, including non-parametric choices.
Theorem \ref{thm:ebayes_zerograd} is a novel result that is related to Proposition 3.1 in \cite{papaspiliopoulos:2025} and Lemma S0.1 in \cite{rognon:2025}, but is more general as these authors considered different settings to ours and only the logit function.

\begin{theorem}
Let $\pi(\bgamma \mid \by, \bomega) \propto p(\by \mid \bgamma) \pi(\bgamma \mid \bomega)$ for
an arbitrary model prior $\pi(\bgamma \mid \bomega)$. Then
\begin{align}
\nabla_{\bomega}\log p(\by \mid \bomega)= E_{\bgamma} \left[ \nabla_{\bomega} \log \pi(\bgamma \mid \bomega) \mid \by, \bomega \right]=
\sum_{\bgamma} \pi(\bgamma \mid \by, \bomega) \nabla_{\bomega} \log \pi(\bgamma \mid \bomega).
\nonumber
\end{align}

If $\pi(\bgamma \mid \bomega)= \prod_{j=1}^p \mbox{Bern}(\gamma_j; m_j(\bomega))$ for a given differentiable $m_j(\bomega)$, then
\begin{align}
\nabla_{\bomega}\log p(\by \mid \bomega)= \sum_{j=1}^p \frac{\nabla_{\bomega} m_j(\bomega)}{m_j(\bomega) [1 - m_j(\bomega)]} 
\left[ \pi(\gamma_j=1 \mid \by, \bomega) - \pi(\gamma_j=1 \mid \bomega) \right].
\nonumber
\end{align}
Further, $\nabla_{\bomega}\log p(\by \mid \bomega)= \sum_{j=1}^p \bz_j \left[ \pi(\gamma_j=1 \mid \by, \bomega) - \pi(\gamma_j=1 \mid \bomega) \right]$
for the inverse logit function $m_j(\bomega)= (1 + e^{-\bz_j^T \bomega})^{-1}$.
\label{thm:ebayes_zerograd}
\end{theorem}

Beyond computation, Theorem \ref{thm:ebayes_zerograd} provides intuition for the empirical Bayes solution, and for the degeneracies associated with maximizing $p(\by \mid \bomega)$ mentioned in Section \ref{ssec:ebayes_hyperpar}.
Said interpretation is simplest when $\pi(\bgamma \mid \bomega)= \prod_{j=1}^p \mbox{Bern}(\gamma_j; m_j(\bomega))$ and $m_j(\bomega)$ is the inverse logit function. Then, setting $\nabla_{\bomega} \log p(\by \mid \bomega)= {\bf 0}$ is equivalent to
\begin{align}
\sum_{j=1}^p \bz_j  \pi(\gamma_j=1 \mid \bomega) = \sum_{j=1}^p \bz_j \pi(\gamma_j=1 \mid \by, \bomega).
\nonumber
\end{align}
That is, the inner-products between the columns of $\bZ$ and their prior inclusion probabilities
$(\pi(\gamma_1=1 \mid \bomega), \ldots, \pi(\gamma_p=1 \mid \bomega))$ must be equal to the inner products with the posterior inclusion probabilities.
For example, if $\bZ$ divides covariates into $B$ blocks, this implies that the average prior and posterior inclusion probabilities within each block must be equal.
This is on the one hand intuitive, one sets higher prior inclusion probabilities in the blocks where the corresponding posterior probabilities are higher.
On the other hand the expression for the gradient highlights the degeneracy problem: if $\bomega$ is such that all $\pi(\gamma_j=1 \mid \bomega)$ is equal to either 1 or 0 for all $j$, then $\pi(\gamma_j=1 \mid \by, \bomega)= \pi(\gamma_j=1 \mid \bomega)$ and $\nabla_{\bomega} \log p(\by \mid \bomega)= {\bf 0}$.
This issue can be addressed by setting a minimally informative $\pi(\bomega)$ (Section \ref{ssec:ebayes_hyperpar}).

Theorem \ref{thm:ebayes_zerograd} enables the use of stochastic gradient methods where one uses samples
from $\pi(\bgamma \mid \by, \bomega)$ to estimate $\nabla_{\bomega} \log p(\by \mid \bomega)$.
 In fact, only marginal posterior inclusion probabilities $\pi(\gamma_j=1 \mid \by, \bomega)$ are needed, which are usually easier to estimate than joint probabilities $\pi(\bgamma \mid \by, \bomega)$. 
In principle, one must estimate $\pi(\gamma_j=1 \mid \by, \bomega)$ for each newly considered $\bomega$.
This is less cumbersome than it might first appear.
First, $\pi(\bgamma \mid \by, \bomega) \propto p(\by \mid \bgamma) \pi(\bgamma \mid \bomega)$ where $p(\by \mid \bgamma)$ does not depend on  $\bomega$. Hence, once $p(\by \mid \bgamma)$ is first computed, it can be stored and re-used for all $\pi(\bgamma \mid \bomega)$.
Second, if one has samples from $\pi(\bgamma \mid \by, \bomega)$ one may use sequential Monte Carlo to sample from 
$\pi(\bgamma \mid \by, \bomega')$, and the associated weights are simply $\pi(\bgamma \mid \bomega') / \pi(\bgamma \mid \bomega)$.
Third, we next derive an EM that allows one to do multiple updates of $\bomega$ using a single set of samples from $\pi(\bgamma \mid \bomega)$ at each EM iteration.

\subsection{Expectation-Maximization algorithm}
\label{ssec:em_algorithm}

Consider an EM algorithm setting where the likelihood for the full data $(\by, \bgamma)$ is given by
$p(\by, \bgamma \mid \bomega) \propto p(\by \mid \bgamma) \pi(\bgamma \mid \bomega)$,
and where $\bgamma$ are the latent variables in the usual definition of the EM algorithm.
The algorithm requires an E-step where one obtains the expected full data log-likelihood with respect to $\bgamma$,
given the data $\by$ and the current parameter estimate  $\hat{\bomega}^{(k)}$.
This requires obtaining $\pi(\bgamma \mid \by, \hat{\bomega}^{(k)})$ and, if an exact computation is too costly (one cannot enumerate all models), one may use a stochastic EM algorithm where $\pi(\bgamma \mid \by, \hat{\bomega}^{(k)})$ is replaced by the frequency of posterior samples visiting model $\bgamma$.
Given such $\pi(\bgamma \mid \by, \hat{\bomega}^{(k)})$, the M-step seeks to maximize
\begin{align}
 E_{\bgamma} \left[ \log p(\by \mid \bgamma) + \log \pi(\bgamma \mid \bomega) \mid \by, \bomega= \hat{\bomega}^{(k)} \right]
+ \log \pi(\bomega)=
c + f(\bomega) + \log \pi(\bomega),
\nonumber
\end{align}
where $c$ is a constant that does not depend on $\bomega$ and
\begin{align}
f(\bomega)= 
  \sum_{\bgamma} \pi(\bgamma \mid \by, \hat{\bomega}^{(k)}) \log \pi(\bgamma \mid \bomega).
\nonumber
\end{align}
This expression is analogous to that in Theorem \ref{thm:ebayes_zerograd}, replacing $\pi(\bgamma \mid \by ,\bomega)$ by $\pi(\bgamma \mid \by, \hat{\bomega}^{(k)})$.
Also similar to Theorem \ref{thm:ebayes_zerograd}, 
Corollary \ref{cor:em_zerograd} gives an expression to evaluate $\nabla_{\bomega} f(\bomega)$ using only $p$ operations
when $\pi(\bgamma \mid \bomega)$ factors.

\begin{corollary}
If $\pi(\bgamma \mid \bomega)= \prod_{j=1}^p \mbox{Bern}(\gamma_j \mid m_j(\bomega))$, then
\begin{align}
\nabla_{\bomega} f(\bomega)= \sum_{j=1}^p \frac{\nabla_{\bomega} m_j(\bomega)}{m_j(\bomega) [1 - m_j(\bomega)]} \left[ \pi(\gamma_j=1 \mid \by, \hat{\bomega}^{(k)}) - \pi(\gamma_j=1 \mid \bomega) \right].
\nonumber
\end{align}

In the particular case $m_j(\bomega)= 1/(1 + e^{\bz_j^T \bomega})$, then
$\nabla_{\bomega} f(\bomega)= \bZ^T [\hat{\bpi} - {\bf m}(\bomega)]$
 and 
$\nabla_{\bomega}^2 f(\bomega)= - \bZ^T \bD(\bomega) \bZ$ where 
$\hat{\bpi}= (\pi(\gamma_1 \mid \by, \hat{\bomega}^{(k)}), \ldots, \pi(\gamma_p \mid \by, \hat{\bomega}^{(k)}))^T$,
${\bf m}(\bomega)= (m_1(\bomega), \ldots, m_p(\bomega))^T$ and
$\bD(\bomega)= \mbox{diag}(m_1(\bomega), \ldots, m_p(\bomega))$.
\label{cor:em_zerograd}
\end{corollary}

Briefly, each M-step requires solving $\nabla_{\bomega} f(\bomega)= \bA [\hat{\bpi} - {\bf m}(\bomega)]= {\bf 0}$,
where $\bA$ is the $q \times p$ matrix with $j^{th}$ column equal to $[\nabla_{\bomega} m_j(\bomega)]/(m_j(\bomega) [1 - m_j(\bomega)])$, 
${\bf m}(\bomega) \in \mathbb{R}^p$ the vector with $j^{th}$ entry equal to $m_j(\bomega)$, 
and $\hat{\bpi} \in \mathbb{R}^p$ the vector with $j^{th}$ entry $\pi(\gamma_j=1 \mid \by,\hat{\bomega}^{(k)})$.
One may use any standard optimization algorithm, and in our examples we used Newton-Raphson.
When $m_j(\bomega)$ is the inverse logit function, it is easy to see that the objective function of the M-step is strictly concave and therefore the maximum, if it exists, is unique.

Finally, we consider in Proposition \ref{prop:mstep_fullrank} an important particular case where the M-step has a particularly simple solution.
The result holds when $\bZ$ has $q$ unique and linearly independent rows.
This occurs when $\bZ$ divides the covariates into $q$ blocks,
and when the rows of $\bZ$ are clustered into $q$ blocks, and one replaces $\bZ$ by the mean of the assigned cluster.
The solution requires a univariate function $h()$ that can be easily evaluated using Newton's method.

\begin{proposition}
Let $\bU$ be the sub-matrix containing the unique rows of the $p \times q$ matrix $\bZ$.
Suppose that $\bU$ has $q$ rows and rank $q$, $m_j(\bomega)= 1/(1+e^{-\bz_j^T \bomega})$,
and $\bomega \sim N({\bf 0}, g (\bZ^T\bZ/p)^{-1})$.
Then the posterior mode $\hat{\bomega}$ satisfies $\nabla_{\bomega} f(\hat{\bomega}) + \nabla_{\bomega} \log \pi(\hat{\bomega})= {\bf 0}$,
where $\hat{\bomega}= \bU^{-1} \tilde{\bomega}$ and
$\tilde{\omega}_j= h(a_j, g_\omega p)$, where
$h(a,c)$ is the solution to $1/(1+e^{-w}) + w/c= a$,
$a_j= \sum_{i=1}^p \tilde{z}_{ij}^T \hat{\pi}_i / \sum_{i=1}^p z_{ij}$,
$\tilde{\bZ}= \bZ \bU^{-1}$ and $\hat{\bpi}$ is as in Corollary \ref{cor:em_zerograd}.
\label{prop:mstep_fullrank}
\end{proposition}

\subsection{Other computational methods}
\label{ssec:other_comput_methods}

An alternative to obtaining a point estimate $\hat{\bomega}$ maximizing the posterior $\pi(\bomega \mid \by)$ is to approximate this posterior with MCMC. That is, one performs fully Bayesian inference using $\pi(\bomega \mid \by)$, obtains $\hat{\bomega}$ as some sensible posterior summary such as the posterior mean,
and finally uses a second MCMC run to approximate $\pi(\btheta \mid \by, \hat{\bomega})$.
Here we do not discuss the many options available in the MCMC literature,
but rather outline strategies that are relatively simple to implement. 

A first option is to use the gradients derived in Section \ref{ssec:comp_gradientbased} within stochastic gradient MCMC  (see \cite{nemeth:2021} for a review).
Alternatively, it is often simple to design an MCMC algorithm to sample either from $\pi(\bgamma,\bomega \mid \by)$ or from $\pi(\btheta, \bgamma, \bomega \mid \by)$.
For example, this strategy was used by \cite{jewson:2023} in Gaussian graphical model setting. 
Briefly, given some initialization $(\bgamma^{(0)}, \bomega^{(0)})$,
one may use a block Gibbs sampler where its iteration $l=1,2,\ldots$ proceeds as follows:
\begin{enumerate}
\item Sample $\bgamma^{(l)} \sim \pi(\bgamma \mid \by, \bomega^{(l-1)}) \propto p(\by \mid \bgamma) \pi(\bgamma \mid \bomega^{(l-1)})$.

\item Sample $\bomega^{(l)} \sim \pi(\bomega \mid \by, \bgamma^{(l)}) \propto \pi(\bgamma^{(l)} \mid \bomega) \pi(\bomega)$.
\end{enumerate}

Upon convergence, $\bomega^{(l)}$ is by definition a sample from $\pi(\bomega \mid \by)$.
One may set $\hat{\bomega}$ to the posterior mean $L^{-1} \sum_{l=1}^L \bomega^{(l)}$ (possibly after discarding a burn-in period),
or to some Monte Carlo estimate of the posterior mode.
Step 1 requires sampling from the posterior of the model $\bgamma$ given hyper-parameters $\bomega= \bomega^{(l-1)}$.
This step has the same complexity as sampling $\bgamma$ in a fully Bayesian setting where no data integration is performed.
Step 1 is easy to implement as long as one can compute or approximate $p(\by \mid \bgamma)$, as is the case for generalized linear and generalized additive models, for example (see the opening of Section \ref{sec:computation} for some discussion). 
Step 2 is typically also easy. 
First, note that $\bomega$ is often low-dimensional and that
$\nabla_{\bomega} \log \pi(\bgamma \mid \bomega)$ is usually easy to obtain (else, one can use automatic differentiation), enabling the use of gradient-based methods such as Hamiltonian Monte Carlo.
Second, if $\pi(\bgamma \mid \bomega)= \prod_j \mbox{Bern}(\gamma_j; m_j(\bomega))$ where $m_j$ is the inverse logit (or some other suitable) function, Step 2 requires sampling from the posterior of a logistic regression model where the observed outcomes are $\gamma_1^{(l)},\ldots,\gamma_p^{(l)}$, for which a  plethora  of algorithms are available.

If marginal likelihoods $p(\by \mid \bgamma)$ in Step 1 are hard to approximate, one could alternatively sample
$\btheta^{(l)} \sim \pi(\btheta \mid \by)$ using PDMPs such as the sticky Zig-Zag sampler \citep{bierkens:2023},
and subsequently obtain $\bgamma^{(l)}$ as a deterministic function of $\btheta^{(l)}$.
Briefly, PDMPs only require being able to evaluate the likelihood $p(\by \mid \btheta)$, 
and can handle situations where $\pi(\btheta \mid \bomega)$ is such that some singletons have non-zero prior probability.
For example, in variable selection we have non-zero $\pi(\theta_j = 0 \mid \bomega)$ and,
after sampling $\btheta^{(l)} \sim \pi(\btheta \mid \by, \bomega)$, one simply computes $\gamma_j^{(l)}= \mbox{I}(\theta_j^{(l)} \neq 0)$.

If one only seeks a point estimate $\hat{\bomega}$, there's yet another strategy.
One can replace the model selection prior $\pi(\btheta \mid \bomega)$, which places positive probability to zero parameter values, by a continuous prior that sets high probability close to zero \citep{george:1993}.
Typically, this is achieved by replacing Dirac measures at 0 by Gaussian or Laplace distributions that have zero mean and small variance.
In this setting it is possible to devise an EM algorithm to find a local mode of
$\pi(\btheta, \bomega \mid \by)$ \citep{rockova:2014,avalospacheco:2022}.
Briefly, the algorithm views $\bgamma$ as a latent variable and
defines a complete-data likelihood $p(\by,\bgamma \mid \btheta, \bomega)$,
similarly to the EM algorithm outlined in Section \ref{ssec:em_algorithm},
except that the parameters are now $(\btheta,\bomega)$ rather than only $\bomega$.
This strategy can be very fast computationally, because the EM updates do not involve integrals over $\bgamma$
and they may be even available in closed-form.
However, the value of $\bomega$ featuring in the mode of $\pi(\btheta, \bomega \mid \by)$ is different from the mode of $\pi(\bomega \mid \by)$. Whereas the latter has been well-studied theoretically (see Section \ref{sec:ebayes}), the former lacks theoretical guarantees (to our knowledge).
Unless such guarantees are provided, we do not recommend this strategy.

\section{Examples}\label{sec:results}

To assess the improvements in inference brought by data integration,
we compared the performance of our empirical Bayes framework relative to an analogous framework where $\bZ$ only contains the intercept.
We also considered the Beta-Binomial prior \citep{scott:2010}
 as a fully Bayesian counterpart to our framework where 
$\bZ$ only contains the intercept,  
as in our experience this method is an excellent default,
and three penalized likelihood methods.
The latter are the LASSO \citep{tibshirani:1996}, adaptive LASSO (ALASSO, \cite{zou:2006}) and SCAD \citep{fan:2001},
and we set their penalization parameters with 10-fold cross-validation
using functions mylars and ncvreg in R packages parcor 0.2.6 and ncvreg 3.15.0 (respectively).

For our empirical Bayes methodology we used the EM algorithm from Section \ref{sec:computation}, which we implemented in function modelSelection\_eBayes in R package modelSelection, using $M=1,000$ full Gibbs iterations to estimate log-gradients and hessians at each hyper-parameter value $\bomega$.
For fully Bayesian inference under the Beta-Binomial prior we used function modelSelection, also at modelSelection.
In both cases, we used $L=5,000$ full Gibbs iterations to explore the model space.
 Appendix \ref{sec:runtimes} shows that very similar results are obtained with $M=100$ and $L=1,000$. It also shows that run times are comparable to those under the Beta-Binomial prior. 
Regarding the prior on the coefficients, we set a product MOM prior with default prior dispersion ($g=1/3$, after the columns in $\bX$ are standardized to unit sample variance) \citep{johnson:2012,rossell:2017}.
As a sensitivity analysis, we also obtained results using Zellner's prior with the default $g=1$,  and also for $g \in \{0.01,10\}$.  For clarity of the figures and space constraints, Zellner's prior results are in the supplement.
The prior on the error variance $\phi$ was set to a (minimally informative) inverse gamma(0.01, 0.01).
The R code to reproduce our analyses is at 
\url{https://github.com/davidrusi/paper_examples/tree/main/2025_dataintegration_eBayes}.

\subsection{Simulation study}
\label{ssec:simstudy}

\begin{figure}
\begin{center}
\begin{tabular}{cc}
\multicolumn{2}{c}{Scenario 1 $(\omega_1=2, \omega_2=0)$} \\
\includegraphics[width=0.49\textwidth]{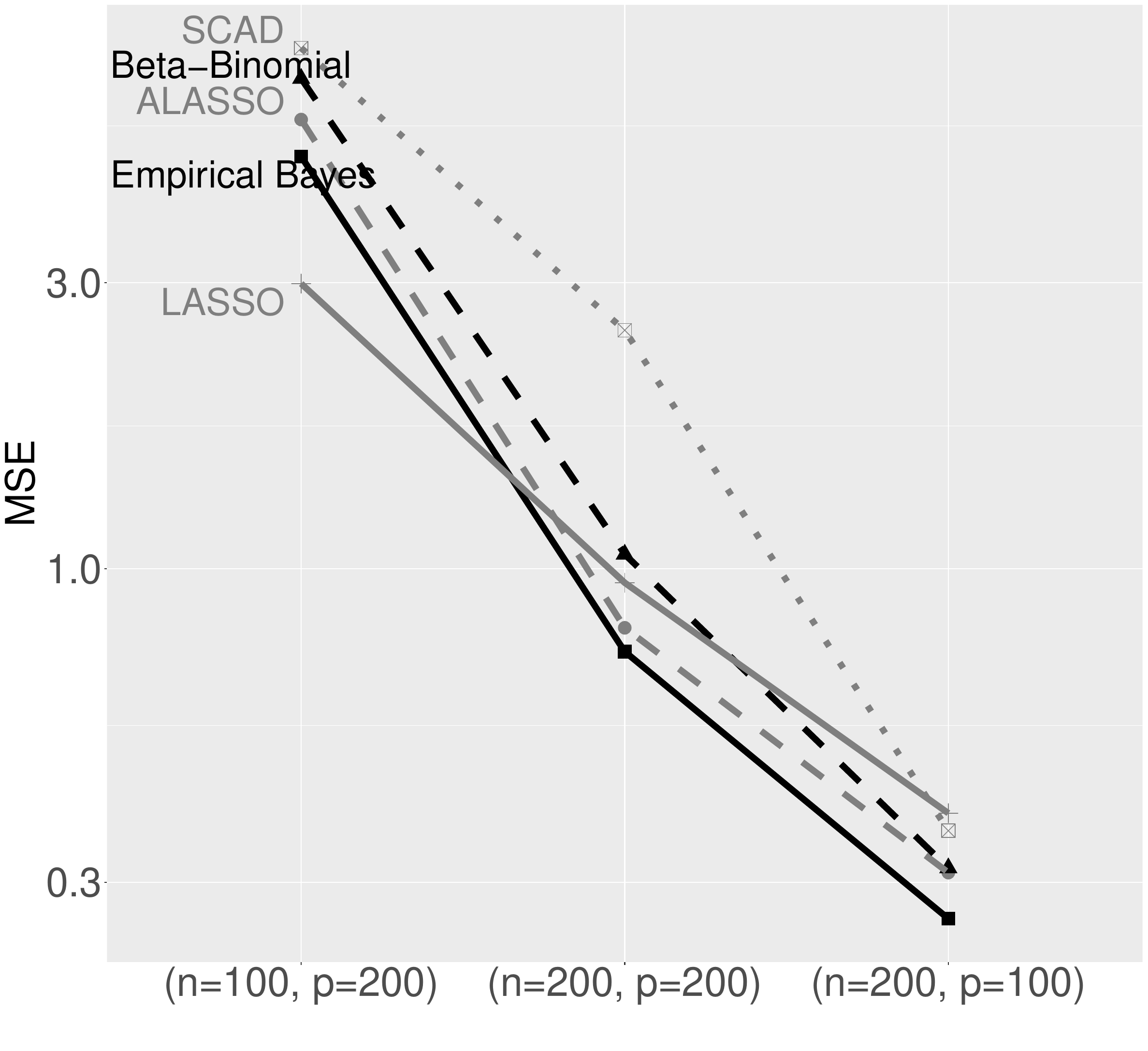} &
\includegraphics[width=0.49\textwidth]{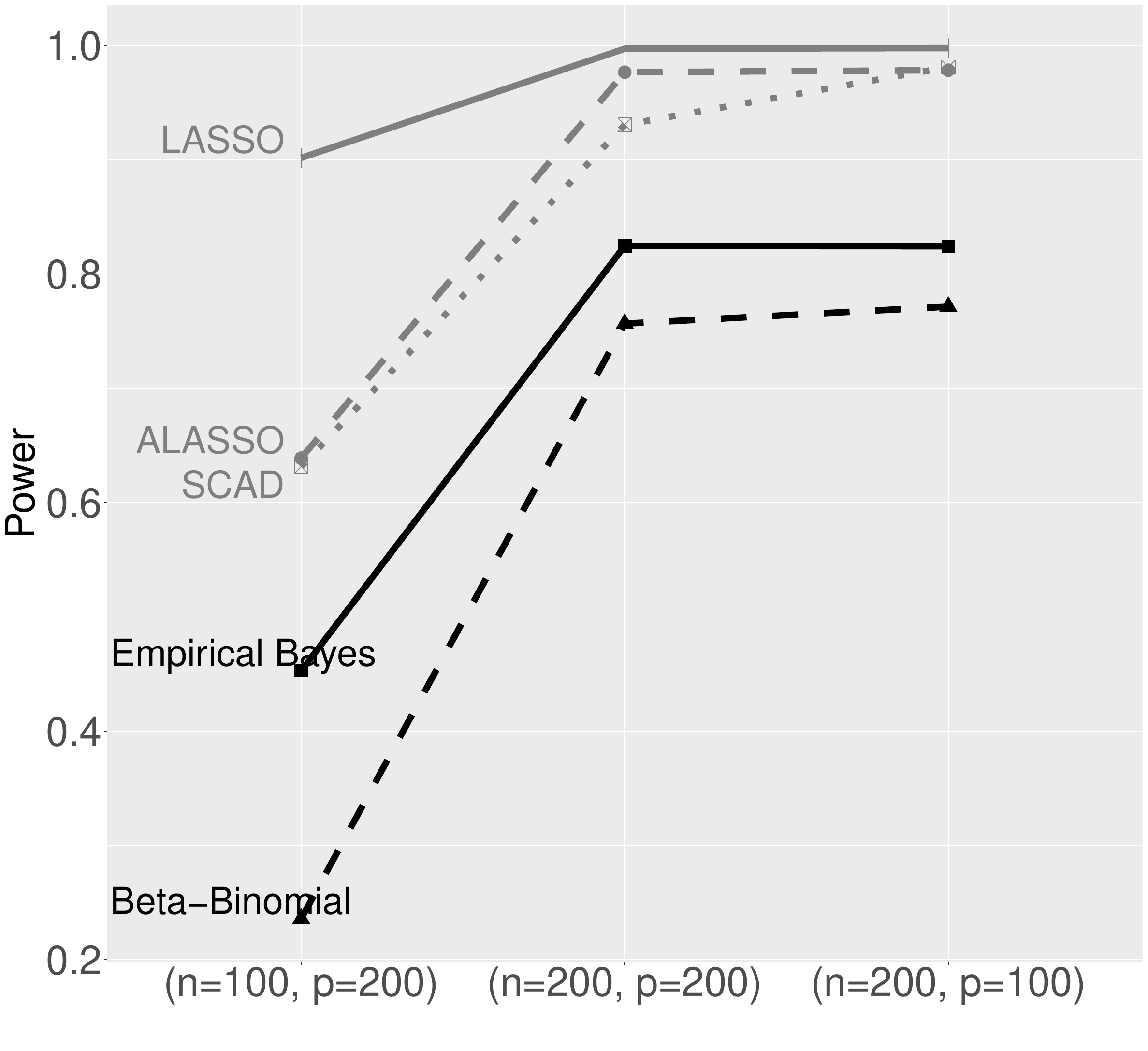} \\
\multicolumn{2}{c}{Scenario 2 $(\omega_1=1, \omega_2=0)$} \\
\includegraphics[width=0.49\textwidth]{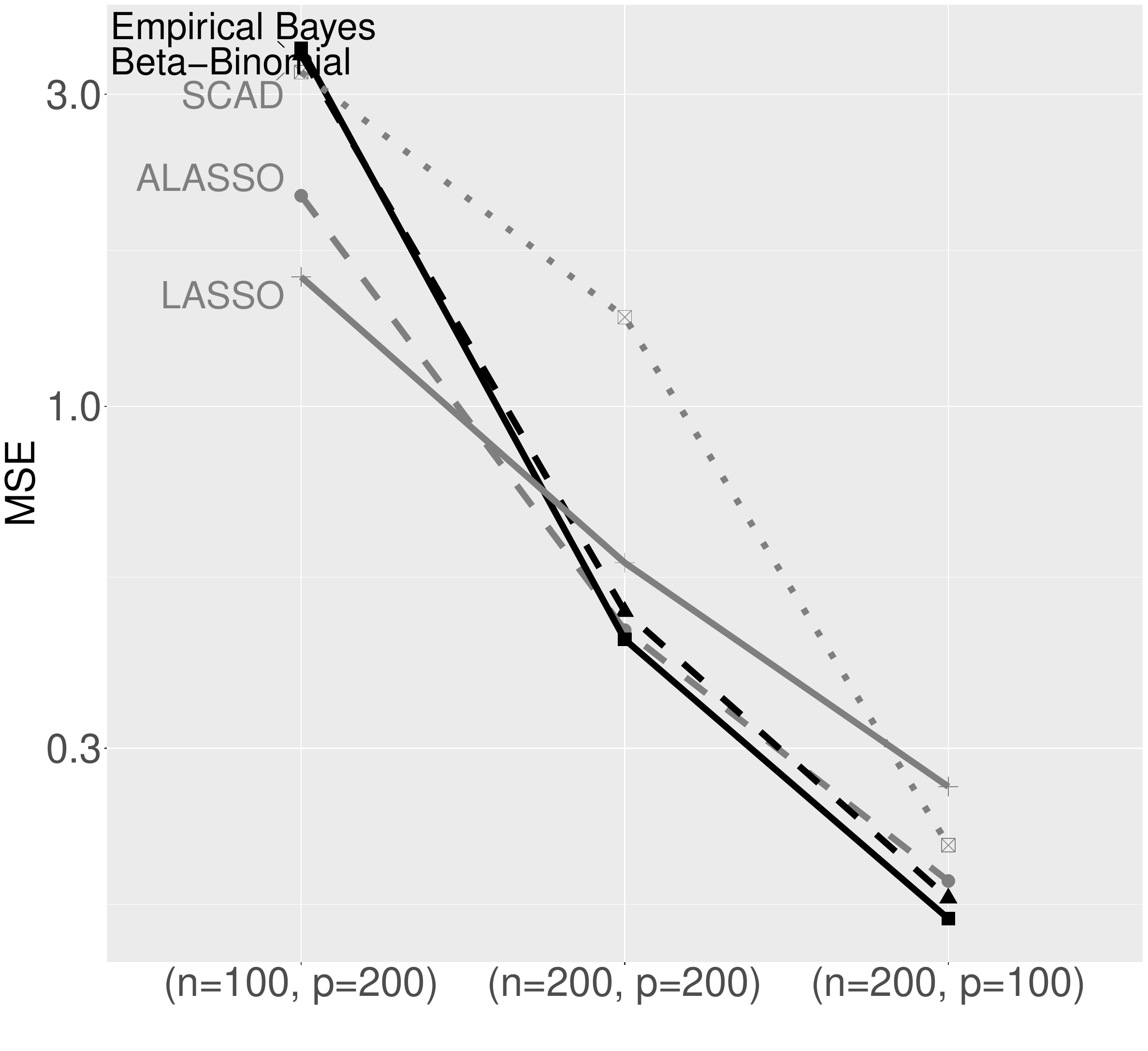} &
\includegraphics[width=0.49\textwidth]{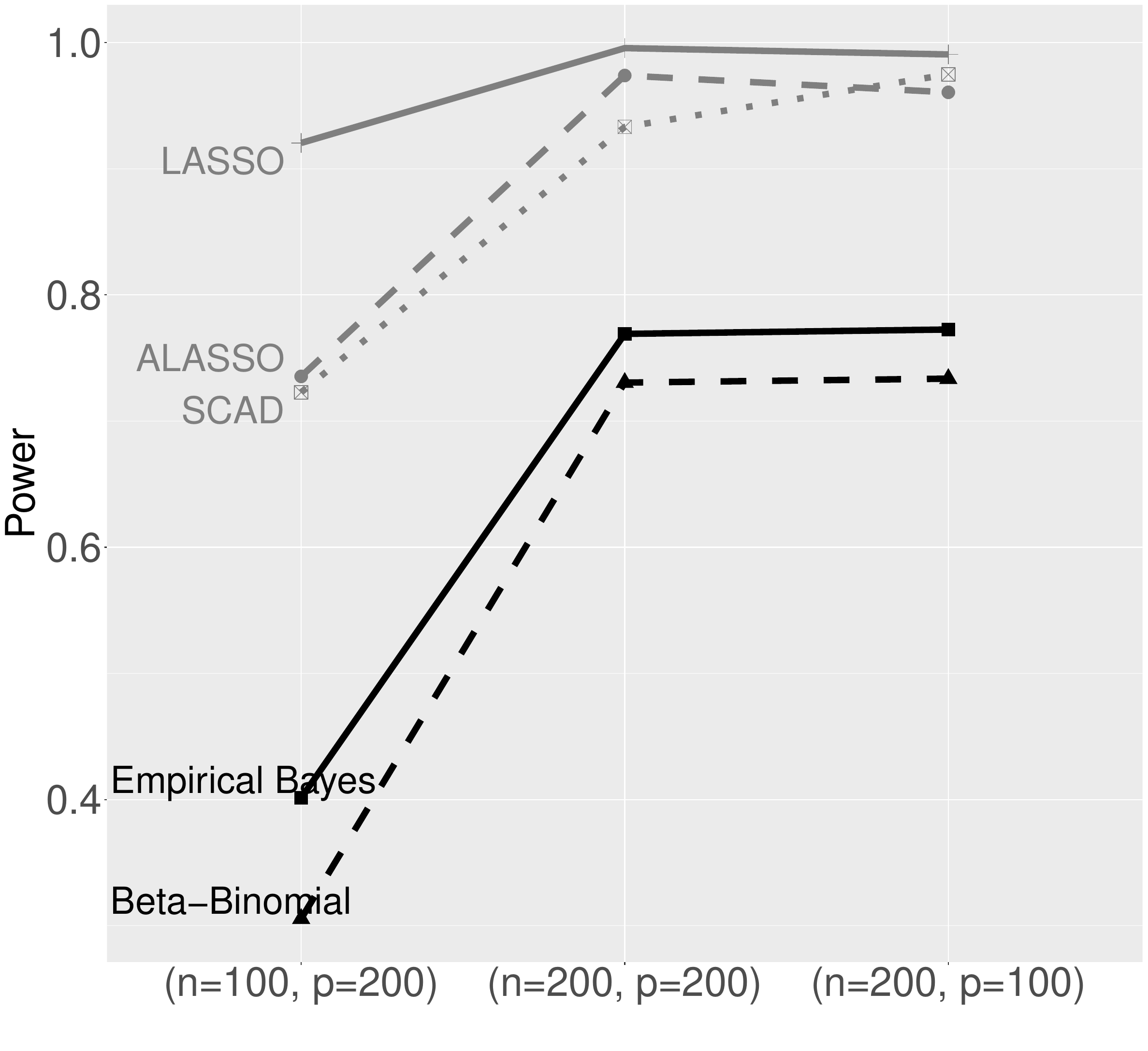} \\
\multicolumn{2}{c}{Scenario 3 $(\omega_1=0, \omega_2=0)$} \\
\includegraphics[width=0.49\textwidth]{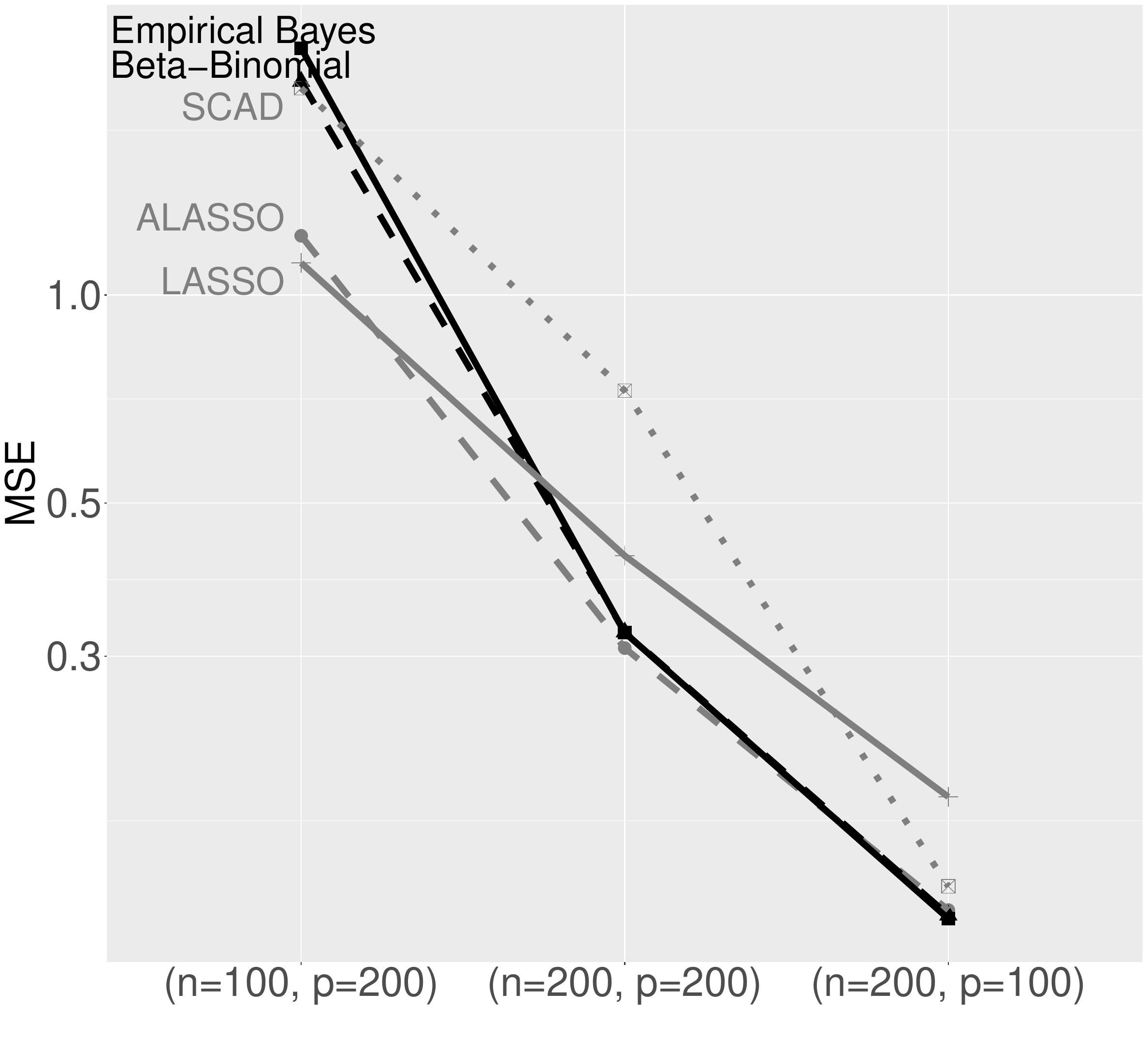} &
\includegraphics[width=0.49\textwidth]{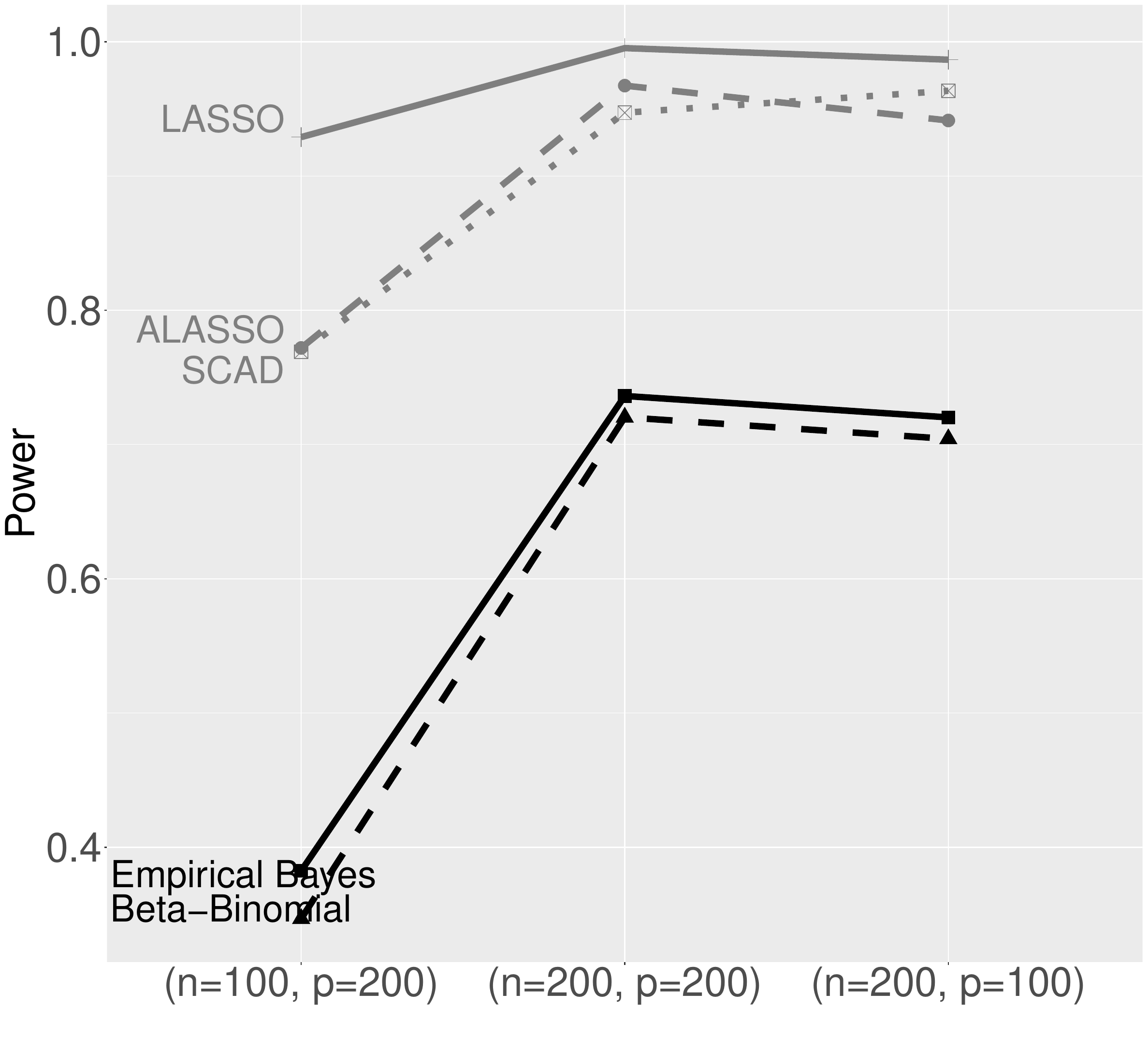} \\
\end{tabular}
\end{center}
\caption{Mean squared estimation error (left) and power (right) in simulation study}
\label{fig:simstudy}
\end{figure}


\begin{figure}
\begin{center}
\begin{tabular}{ccc}
Scenario 1 & Scenario 2 & Scenario 3 \\
\includegraphics[width=0.33\textwidth]{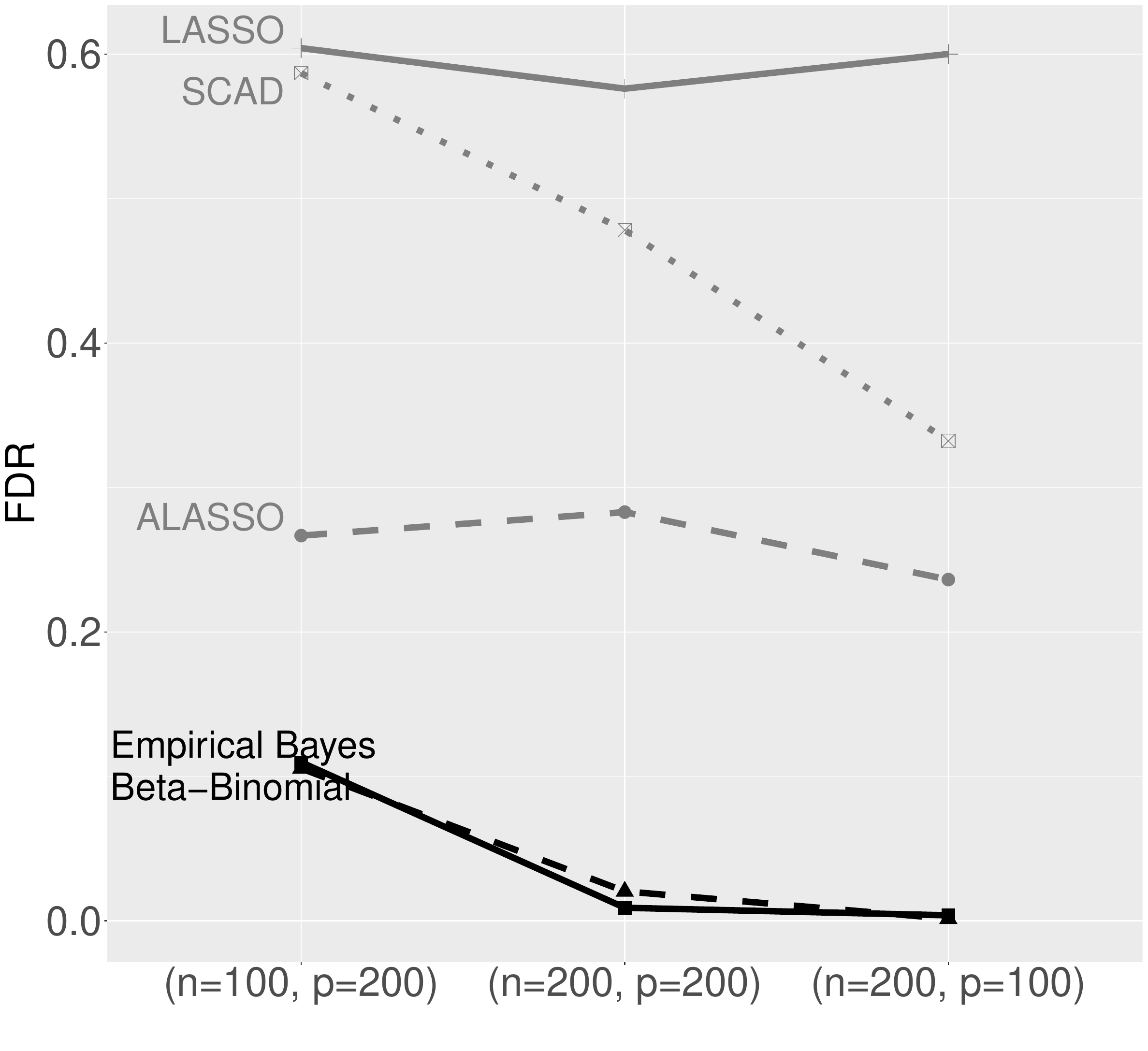} &
\includegraphics[width=0.33\textwidth]{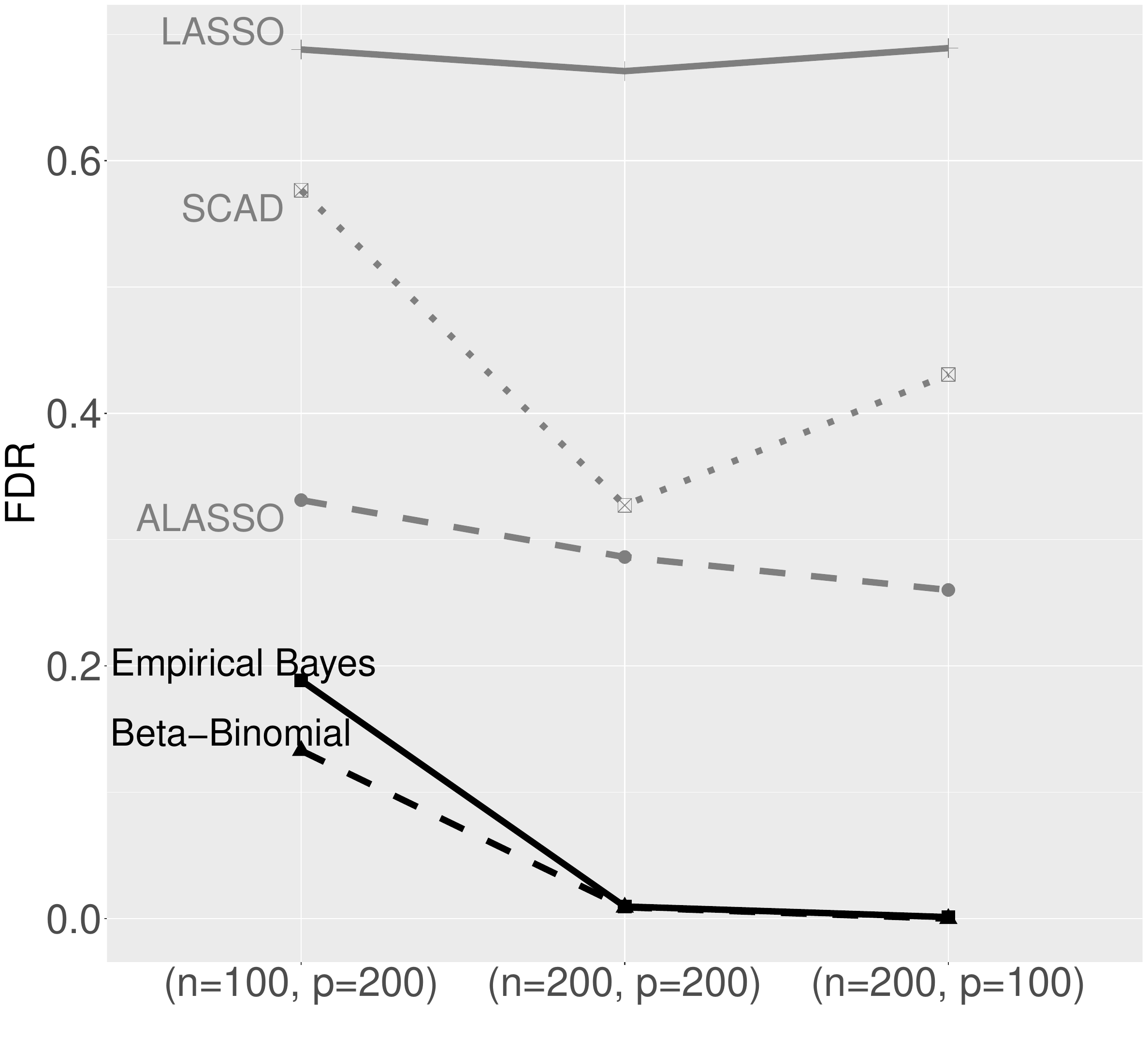} &
\includegraphics[width=0.33\textwidth]{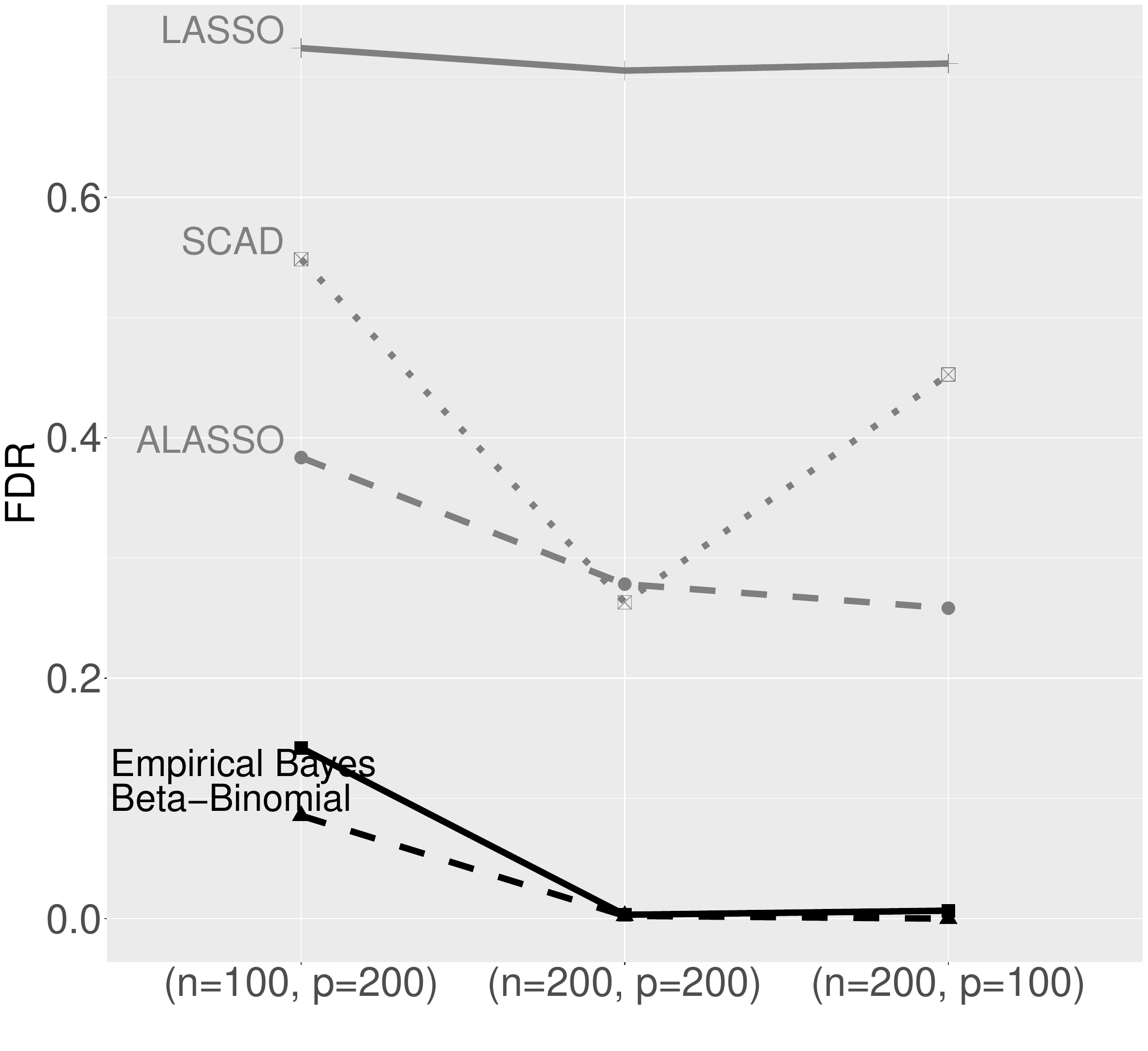} \\
\end{tabular}
\end{center}
\caption{False discovery rate in simulation study.
Scenario 1 corresponds to $(\omega_1=2, \omega_2=0)$, Scenario 2 to $(\omega_1=1, \omega_2=0)$ and Scenario 3 to $(\omega_1=0, \omega_2=0)$}
\label{fig:simstudy_fdr}
\end{figure}


We considered five simulation scenarios that differ in the degree of informativeness of the external data $\bZ$.
In all scenarios, $\bZ$ contains an intercept and two meta-covariates generated from standard Gaussians with 0.5 correlation.
Variable inclusion indicators were simulated using $\pi(\gamma_j =1 \mid \bz_j)= 1 / (1 + e^{-\omega_0 - \omega_1 z_{j1} - \omega_2 z_{j2}})$,
where $\omega_0= \log(0.05/0.95)$ and $\omega_2=0$ in all five scenarios.
That is, our empirical Bayes model prior was estimated using the two meta-covariates, but the second meta-covariate was truly non-informative.
The scenarios differ only in $\omega_1$:
Scenario 1 is a strongly informative case where $\omega_1=2$,
Scenario 2 a moderately informative case where $\omega_1=1$,
and Scenario 3 a non-informative case where $\omega_1=0$.
Scenarios 4 and 5 were added for further exploration, by setting $\omega_1=1.5$ and $0.75$ respectively. The results are similar to those from Scenarios 1 and 3, hence we refer the reader to the supplementary results.
The regression parameters $\theta_j$ for the active covariates ($\gamma_j=1$) were set uniformly spaced in [1/3, 2/3]. 
We chose this range because $\theta_j \not\in [1/3,1/2]$ were not very informative: all methods had a high power to detect $\theta_j > 2/3$, and low power to detect $\theta_j < 1/3$.
In all settings, covariate values were drawn from a multivariate Gaussian with zero means, unit variances, and pairwise correlations $\mbox{cor}(x_{ij}, x_{il})= 0.5$ for $j \neq l$.

In each simulation scenario we considered three combinations of $n$ and $p$.
The most challenging one features $(n=100, p=200)$, then we also considered $(n=200, p=200)$ and $(n=200, p=100)$.
For each of these combinations, we generated 100 independent datasets, and we report averaged results.
For Bayesian methods, the point estimate was the Bayesian model averaging estimator.
To obtain the frequentist power (proportion of truly active covariates that were selected) and false discovery rate (FDR, average proportion of false positives among the selected covariates), variable selection was based on their posterior inclusion probability being $\geq 0.95$.

Figure \ref{fig:simstudy} displays the mean squared parameter estimation error (MSE) of the point estimates $\hat{\btheta}$, and the power. Figure \ref{fig:simstudy_fdr} displays the FDR. 
To assess the potential benefits of data integration, the main comparisons of interest are our empirical Bayes framework vs. its counterpart where $\bZ$ only features the intercept, and vs. the Beta-Binomial prior.
To avoid cluttering the figures, we do not include the results for the intercept-only empirical Bayes. These are in the supplement, and show that the power and MSE are extremely similar to the Beta-Binomial and the FDR to empirical Bayes using all the meta-covariates.
The main findings are that in Scenario 1 the empirical Bayes model prior attained significantly higher power and lower MSE than the Beta-Binomial (and the intercept-only empirical Bayes), and to a lesser extent also in Scenario 2. 
The improvements in power were marked in some cases, e.g. from 0.25 to 0.45 in Scenario 1 when $(n=100, p=200)$.
The FDR was similar in most cases, except in Scenarios 2 and 3 where for the $(n=100, p=200)$ case the FDR was a higher for empirical Bayes.
Overall, these results show that data integration can significantly improve inference when some of the meta-covariates are truly informative.
In Scenario 3, where the meta-covariates were completely uninformative, the MSE and power of the empirical Bayes prior was similar to the Beta-Binomial, and the FDR was similar to that in Scenario 2.
In Scenario 3 we do not expect empirical Bayes to help, but reassuringly it does not hurt much either.
As mentioned, the results when using Zellner's prior on the coefficients were qualitatively similar, see the supplementary results.

Penalized likelihood methods attained higher power but incurred a much higher FDR, roughly in the (0.3,0.7) range, which would be deemed too high in most applications.
This is not an issue per se, given that these methods are more geared towards estimation than false discovery control.
The predictive accuracy of LASSO and adaptive LASSO (ALASSO) ranged from slightly better to slightly worse than for the Bayesian methods, depending on the setting.

\subsection{Colon cancer}
\label{ssec:colon_cancer}

\begin{figure}
\begin{center}
\begin{tabular}{cc}
\includegraphics[width=0.5\textwidth]{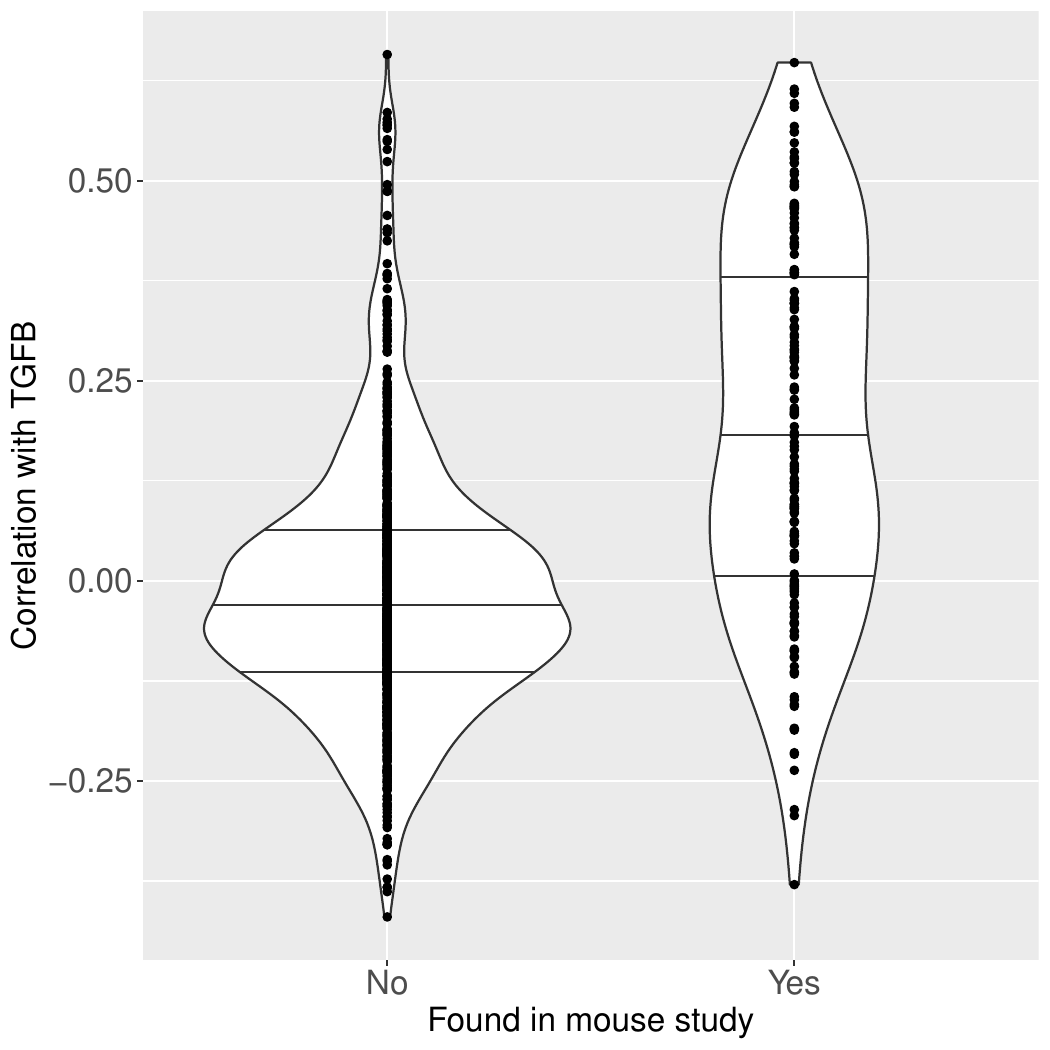} &
\includegraphics[width=0.5\textwidth]{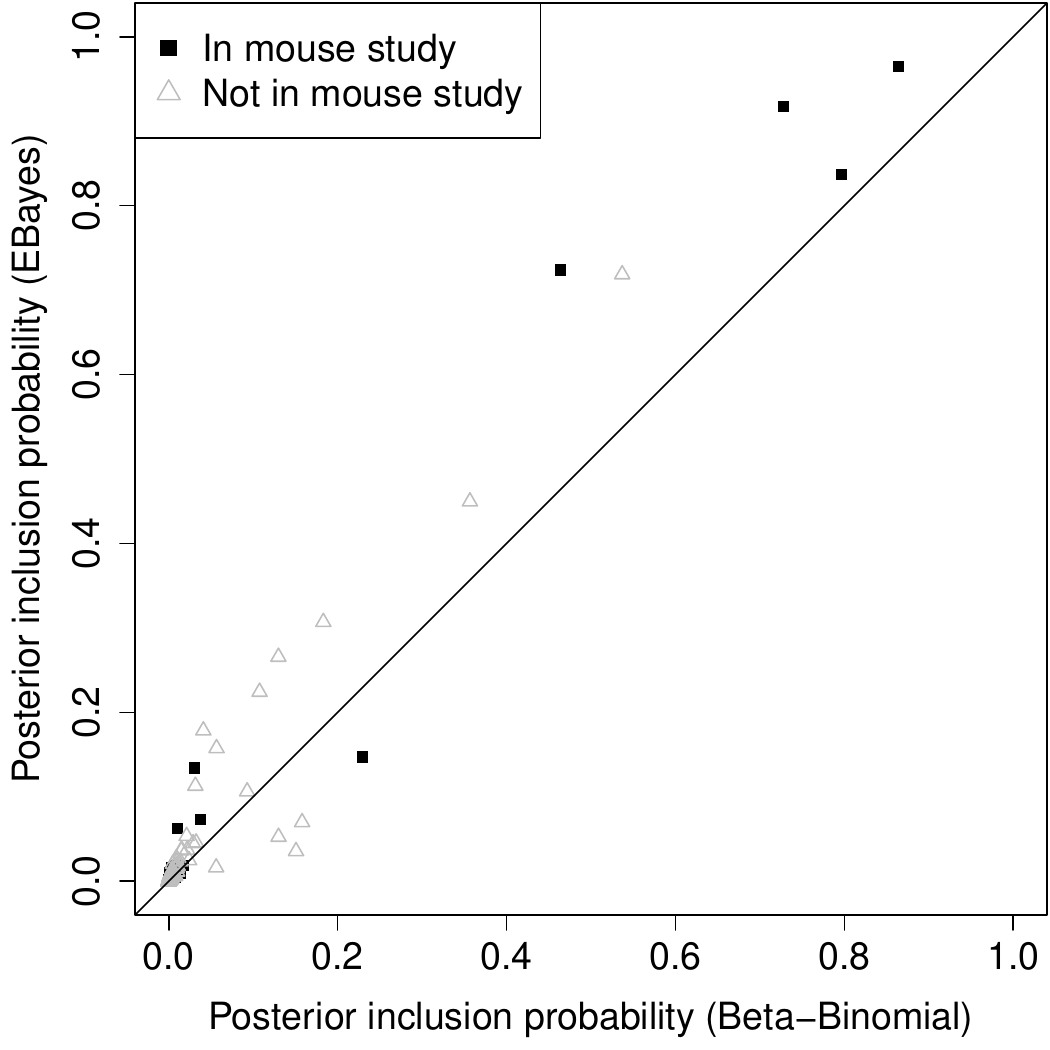} \\
\end{tabular}
\end{center}
\caption{TGFB data. Left: marginal correlations with the outcome vs. gene being in/out the mouse list
Right: posterior inclusion probabilities using empirical Bayes vs. using Beta-Binomial prior}
\label{fig:tgfb_exploratory}
\end{figure}

\begin{table}[h]
\begin{tabular}{@{}l|ccc|ccc|@{}}
\hline
     & \multicolumn{3}{c|}{EBayes ($R^2=0.527$)} & \multicolumn{3}{c|}{Beta-Binomial ($R^2=0.488$)} \\
Gene  & $E(\theta_j \mid \by)$ & 95\% interval & PIP & $E(\theta_j \mid \by)$ & 95\% interval & PIP \\
\hline
CILP      & 0.23 & (0.14, 0.32) & 0.96 & 0.23 & (0.00, 0.34) & 0.86 \\ 
GAS1      & 0.33 & (0.00, 0.45) & 0.92 & 0.25 & (0.00, 0.47) & 0.73 \\ 
HIC1      & 0.24 & (0.00, 0.38) & 0.84 & 0.22 & (0.00, 0.39) & 0.80 \\ 
ESM1      & 0.15 & (0.00, 0.27) & 0.72 & 0.09 & (0.00, 0.27) & 0.46 \\ 
KCNJ5-AS1 & 0.15 & (0.00, 0.26) & 0.72 & 0.10 & (0.00, 0.26) & 0.54 \\ 
\hline
\end{tabular}
\caption{TGFB data. BMA estimates and 95\% intervals for genes with $\pi(\gamma_j=1 \mid \by) > 0.5$ (PIP) in either the empirical Bayes or Beta-Binomial based analysis. $R^2$ is the squared correlation between the outcome and its leave-one-out prediction}
\label{tab:tgfb_topgenes}
\end{table}

In molecular biology it is common to use animal models (e.g. mice) to study genomic mechanisms such as gene expression.
Findings from such studies often serve as a basis to study human diseases, for example to assess whether genes that played key roles in the animal model are also important in humans.
We assess our empirical Bayes methodology in such a transfer learning task (from mice to humans) in colon cancer.

Specifically, \cite{calon:2012} studied a list of 172 genes that were associated to TGFB in mice, and showed that TGFB plays a crucial role in colon cancer progression.
We consider a colon cancer dataset analysed in \cite{rossell:2017}, which measures the expression of TGFB plus $p=1,000$ genes (including the 172 genes from the mouse list) for $n=262$ human patients.
Both TGFB and the other genes were standardized to zero mean and unit variance.
Given the importance of TGFB in colon cancer, it is of interest to understand how its expression is associated with other genes.
To this end, we used our methodology to regress TGFB on the 1,000 genes, using as a meta-covariate the binary indicator of whether the gene was in the mouse shortlist or not.
Hence, $\bZ$ is a $1000 \times 2$ matrix containing the intercept and said indicator (to be specific, coded with a sum-to-zero constraint, so that it is orthogonal to the intercept).

As a purely exploratory analysis, TGFB showed stronger correlations with the 172 genes inside the mouse shortlist than with other genes, see Figure \ref{fig:tgfb_exploratory} (left).
While this suggests that the mouse findings may have value for humans, these correlations only measure marginal association, and are not concluding evidence that the mouse study helped identify genes that are predictive of TGFB (i.e., conditionally associated with TGFB, given all other genes).
To assess this rigorously, we performed Bayesian variable selection with our empirical Bayes framework.
For comparison, we repeated the exercise under a Beta-Binomial prior, which does not use the mouse data.
 The run time under the Beta-Binomial prior was 14 seconds and for our empirical Bayes framework it was 49 seconds (Ubuntu laptop, 13$^{th}$ gen Intel core i7, 32Gb RAM). 
We report results when using the pMOM prior on parameters, the results for Zellner's prior are almost identical and are in the supplement.
The estimated empirical Bayes prior inclusion probabilities were 0.032 for genes in the mouse list and 0.015 for those outside, supporting that genes identified in mice are more likely to be relevant in humans.
Figure \ref{fig:tgfb_exploratory} (right) shows the posterior inclusion probabilities (PIPs) for all genes using empirical Bayes and using the Beta-Binomial.
While the PIP for most genes was close to 0 under both, there are non-negligible differences for genes with large PIPs.
Table \ref{tab:tgfb_topgenes} shows genes with PIP$>$0.5 under either analysis.
Empirical Bayes assigned higher PIPs to all these genes, e.g. 2 had PIP$>$0.9 and 5 had PIP$>$0.5, whereas for the Beta-Binomial there were 0 and 4 genes respectively.

These differences in PIP are not huge, but they can have practical implications.
For example, if one sets a stringent PIP threshold such as 0.95, only empirical Bayes would find a gene associated with TGFB.
This gene is CILP, a cartilage intermediate layer protein that has been independently found to be associated with colon cancer progression 
\citep{wang_xueli:2023}.
At a 0.9 threshold a second gene emerges, GAS1. This is a growth-arrest specific gene that is associated with colon cancer  metastasis  \citep{li_qingquo:2016}.
Note that both CILP and GAS1 were in the mouse shortlist, as well as HIC1 and ESM1, but not KCNJ5-AS1.

While this discussion already suggests that the empirical Bayes findings are supported by the biological literature, we conducted two more analyses to further assess that using the mouse-based meta-covariates gave better inference.
First, we used leave-one-out cross-validation to assess the accuracy of the Bayesian model averaging predictions.
Specifically, we computed the $R^2$ coefficient between out-of-sample predictions and observations (i.e. their squared correlation),
obtaining $R^2=$0.527 for empirical Bayes and 0.488 for the Beta-Binomial.
That is, adding the mouse data gave slightly higher out-of-sample predictive accuracy.
Second, there is the question of whether the estimated $\hat{\omega}_2= 0.74$ is significantly different from 0.
To assess this in an independent fashion from our Bayesian analysis, we used the hyper-parameter estimates from the leave-one-out exercise to obtain a jackknife confidence interval. Said interval was (0.68, 0.82), and all the estimated $\omega_2$ were $>$0.64, providing strong evidence that truly $\omega_2 \neq 0$ and that the mouse findings are informative for humans.

\section{Discussion}\label{sec:discussion}

Bayesian inference is naturally suited to incorporate previous data or knowledge. 
Fully Bayesian inference is particularly appealing when one is willing to set a moderately informative prior linking the current and previous data or parameters.
In most situations full Bayes enjoys excellent properties. In particular, as $n$ grows its solution often matches that of empirical Bayes and enjoys a frequentist interpretation as conditional probabilities on parameters given the data, under repeated sampling of both data and parameters.
There are situations when full and empirical Bayes do not match asymptotically, these are somewhat extreme settings where the prior plays a critical role. 
In our data integration context, this occurs for example when parameters are divided into two blocks and one is much sparser than the other.
In most of our examples the benefits of using empirical Bayes for data integration were moderate but non-negligible.
From a theoretical viewpoint, it is interesting that the mathematical conditions for model selection consistency are milder when  one  integrates data. 
 Such benefits can be obtained with empirical Bayes, but not will full Bayes. Our theory applies to models with Gaussian errors, such as linear and non-linear additive regression, when there is a single meta-covariate that splits covariates into blocks. Our algorithms apply to fully general settings, they only require estimating posterior inclusion probabilities,  and that hyper-parameters only feature in prior inclusion probabilities.  In fact, our software includes generalized linear and additive models.
Interesting future work includes developing theory beyond Gaussian errors  and for cases where $\bZ$ contains many meta-covariates, and comparing our results to those obtained when doing data integration with full data.  Our current theory still applies to such settings, provided that one discretizes or clusters the covariates into blocks, but because of the potential information loss our current results would not be tight in such settings.

Performing data integration via empirical Bayes has a computational cost, due to adding more layers to the model.
We discussed simple strategies that rendered computations feasible in our examples, but more work in this area is needed.
 In particular, we proposed generic algorithms to regress prior inclusion probabilities on meta-covariates that apply to any probability model, as long as one can evaluate or approximate marginal likelihoods $p(\by \mid \bgamma)$ for candidate models $\bgamma$. An interesting extension is to regress the prior dispersion ($g$ in Zellner's prior) on meta-covariates $\bZ$, to learn how the magnitude of non-zero coefficients varies with $\bZ$. Such an extension would however complicate the algorithms, and they would likely need to be derived on a case-by-case basis, as opposed to the generic algorithms outlined here.

\section*{Acknowledgments}

PRV acknowledges the support of Departament de Recerca i Universitats de la Generalitat de Catalunya and the European Social Fund. DR was funded by grant Consolidaci\'on investigadora CNS2022-135963 by the Agencia Espa\~{n}ola de Investigaci\'on (AE), and PID2022-138268NB-I00 by Ministerio de Ciencia e Innovaci\'on (MICIN)/Agencia Espa\~{n}ola de Investigaci\'on(AEI)/10.13039/501100011033 /Fondo Europeo de Desarrollo Regional (FEDER), and an ICREA Academia fellowship from AGAUR (Generalitat de Catalunya).

\section*{Supplementary material}

The sections below contain proofs and additional data analysis results. 

Additionally, folder 2025\_dataintegration\_eBayes contains data and R scripts to reproduce our results, available at https://github.com/davidrusi/paper\_examples.




\section{Supplementary code}

R code to reproduce our examples is available at
{\small \texttt{https://github.com/davidrusi/paper\_examples/tree/main/2025\_dataintegration\_eBayes}}

\section{Proof of Theorem \ref{thm:ebayes_zerograd}}\label{proof:ebayes_zerograd}

Since $p(\by \mid \bomega)= \sum_{\bgamma} p(\by \mid \bgamma) \pi(\bgamma \mid \bomega)$
and $\nabla_{\bomega} \log \pi(\bgamma \mid \bomega)= [\nabla_{\bomega} \pi(\bgamma \mid \bomega)] / \pi(\bgamma \mid \bomega)$, we have that
\begin{align}
&\nabla_{\bomega}\log p(\by \mid \bomega)= \frac{\nabla_{\bomega} p(\by \mid \bomega)}{p(\by \mid \bomega)}= 
\frac{\sum_{\bgamma} p(\by \mid \bgamma) \nabla_{\bomega} \pi(\bgamma \mid \bomega)}{p(\by \mid \bomega)}= 
\nonumber \\
&= \frac{\sum_{\bgamma} p(\by \mid \bgamma) \pi(\bgamma \mid \bomega) \nabla_{\bomega} \log \pi(\bgamma \mid \bomega)}{p(\by \mid \bomega)}
= \sum_{\bgamma} \pi(\bgamma \mid \by, \bomega) \nabla_{\bomega} \log \pi(\bgamma \mid \bomega),
\nonumber
\end{align}
as we wished to prove.

In the particular case where $\pi(\bgamma \mid \bomega)= \prod_{j=1}^p \pi(\gamma_j \mid \bomega)$, we get
\begin{align}
 \nabla_{\bomega} \log \pi(\bgamma \mid \bomega)=
\sum_{j=1}^p \nabla_{\bomega} \log \pi(\gamma_j \mid \bomega)=
\sum_{j=1}^p \frac{\nabla_{\bomega} \pi(\gamma_j \mid \bomega)}{\pi(\gamma_j \mid \bomega)}.
\nonumber
\end{align}

Let $\pi(\gamma_j \mid \bomega)= m_j(\bomega)^{\gamma_j} [1 - m_j(\bomega)]^{1- \gamma_j}$ for an arbitrary differentiable $m_j()$.
Then
\begin{align}
& \nabla_{\bomega} \pi(\gamma_j \mid \omega)= \begin{cases}
\nabla_{\bomega} m_j(\bomega)= 
\frac{\nabla_{\bomega} m_j(\bomega)}{m_j(\bomega)} \pi(\gamma_j \mid \bomega)
\mbox{, if } \gamma_j=1
\\
- \nabla_{\bomega} m_j(\bomega)=
- \frac{\nabla_{\bomega} m_j(\bomega)}{1 - m_j(\bomega)} \pi(\gamma_j \mid \bomega)
\mbox{, if } \gamma_j=0
\end{cases}
\nonumber \\
&= \pi(\gamma_j \mid \bomega) \nabla_{\bomega} m_j(\bomega) \left[ \frac{\gamma_j}{m_j(\bomega)} - \frac{1 - \gamma_j}{1 - m_j(\bomega)} \right]
= \pi(\gamma_j \mid \bomega) \frac{\nabla_{\bomega} m_j(\bomega)}{m_j(\bomega) [1 - m_j(\bomega)]} \left[ \gamma_j - m_j(\bomega) \right].
\label{eq:grad_priorprob_indep}
\end{align}
This gives that
\begin{align}
&\nabla_{\bomega}\log p(\by \mid \bomega)= 
E \left[  \sum_{j=1}^p 
\frac{\nabla_{\bomega} m_j(\bomega)}{m_j(\bomega) [1 - m_j(\bomega)]} \left[ \gamma_j - m_j(\bomega) \right]
\mid \by, \bomega \right]
\nonumber \\
& \sum_{j=1}^p \frac{\nabla_{\bomega} m_j(\bomega)}{m_j(\bomega) [1 - m_j(\bomega)]} 
\left[ \pi(\gamma_j=1 \mid \by, \bomega) - \pi(\gamma_j=1 \mid \bomega) \right],
\nonumber
\end{align}
as we wished to prove.
In the particular case where $m_j(\bomega)= (1 + e^{\bz_j^T \bomega})^{-1}$, we have
$\nabla_{\bomega} m_j(\bomega)= \bz_j  e^{\bz_j^T \bomega} / (1 + e^{\bz_j^T \bomega})^2= \bz_j m_j(\bomega) [1 - m_j(\bomega)]$,
and then 
$\nabla_{\bomega}\log p(\by \mid \bomega)= \sum_{j=1}^p \bz_j \left[ \pi(\gamma_j=1 \mid \by, \bomega) - \pi(\gamma_j=1 \mid \bomega) \right]$.

\section{Proof of Corollary \ref{cor:em_zerograd}}
\label{sec:proof_em_zerograd}

The gradient of 
$f(\bomega)= \sum_{\bgamma} \pi(\bgamma \mid \by, \hat{\bomega}^{(k)}) \log \pi(\bgamma \mid \bomega)$
is 
$\nabla_{\bomega} f(\bomega)= \sum_{\bgamma} \pi(\bgamma \mid \by, \hat{\bomega}^{(k)}) \nabla_{\bomega} \log \pi(\bgamma \mid \bomega)$.
Using \eqref{eq:grad_priorprob_indep}, we obtain that
\begin{align}
 \nabla_{\bomega} f(\bomega)= \sum_{\bgamma} \pi(\bgamma \mid \by, \hat{\bomega}^{(k)})
\sum_{j=1}^p \frac{\nabla_{\bomega} m_j(\bomega)}{m_j(\bomega) [1 - m_j(\bomega)]} \left[ \gamma_j - \pi(\gamma_j=1 \mid \bomega) \right]
\nonumber \\
= \sum_{j=1}^p \frac{\nabla_{\bomega} m_j(\bomega)}{m_j(\bomega) [1 - m_j(\bomega)]} \left[ \pi(\gamma_j=1 \mid \by, \hat{\bomega}) - \pi(\gamma_j=1 \mid \bomega) \right].
\nonumber
\end{align}

In the particular case where $m_j(\bomega)= (1 + e^{\bz_j^T \bomega})^{-1}$, from Theorem \ref{thm:ebayes_zerograd} we have that
$[\nabla_{\bomega} m_j(\bomega)]/(m_j(\bomega) [1 - m_j(\bomega)])= \bz_j$
and then 
$\nabla_{\bomega} f(\bomega)= \sum_{j=1}^p \bz_j \left[ \pi(\gamma_j=1 \mid \by, \bomega) - \pi(\gamma_j=1 \mid \bomega) \right]= \bZ^T [\hat{\bpi}(\bomega) - {\bf m}(\bomega)]$.
The expression for the hessian $\nabla_{\bomega}^2 f(\bomega)$ follows from straightforward calculus.

\section{Proof of Proposition \ref{prop:mstep_fullrank}}
\label{sec:proof_mstep_fullrank}

The goal is to find $\bomega$ such that $\nabla_{\bomega} f(\bomega) + \nabla_{\bomega} \log \pi(\bomega)= {\bf 0}$.
Plugging in the expression for $\nabla_{\bomega} f(\bomega)$ and $\nabla_{\bomega} \log \pi(\bomega)= -\frac{1}{g} \bV^{-1} \bomega$,
where $\bV^{-1}= (\bZ^T \bZ)/p$, we seek $\bomega$ such that
\begin{align}
\bZ^T [ \hat{\bpi} - {\bf m}(\bomega) ] - \frac{1}{gp} \bZ^T \bZ \bomega=0
\Leftrightarrow
\bZ^T \left[ {\bf m}(\bomega) + \frac{1}{gp} \bZ \bomega \right]= \bZ^T \hat{\bpi}, 
\nonumber
\end{align}
where $m_j(\bomega)= 1/(1 + e^{- \bz_j^T \bomega})$.
The $q \times q$ matrix $\bU$ contains the unique rows in $\bZ$ and is full-rank, hence $\tilde{\bZ}= \bZ \bU^{-1}$ is a matrix whose unique rows equal the identity matrix. 
Define $\tilde{\bomega}= \bU \bomega$ and consider the problem of finding $\tilde{\bomega}$ such that $\tilde{\bZ}^T [\tilde{{\bf m}}(\tilde{\bomega}) - (gp)^{-1} \tilde{\bZ} \tilde{\bomega}]= \tilde{\bZ}^T \hat{\bpi}$, where $\tilde{m}_j(\tilde{\bomega})= 1/(1+ e^{- \tilde{z}_j^T \tilde{\bomega}})$. 
We first show how to obtain $\tilde{\bomega}$, and subsequently that one can find $\bomega$ from $\tilde{\bomega}$.

To obtain $\tilde{\bomega}$, intuitively this is simple because the subset of unique rows in $\tilde{\bZ}$ is the identity matrix.
Specifically, the $j^{th}$ equation in $\tilde{\bZ}^T [\tilde{{\bf m}}(\tilde{\bomega})  - (gp)^{-1} \tilde{\bZ} \tilde{\bomega} ]= \tilde{\bZ}^T \hat{\bpi}$ is 
\begin{align}
 \sum_{i=1}^p  \mbox{I}(\tilde{z}_{ij}=1) \left[ \frac{1}{1 + e^{- \tilde{\bz}_i^T \tilde{\bomega}}} + \frac{\tilde{\bz}_i^T \tilde{\bomega}}{gp} \right]
= \tilde{\bz}_{.j}^T \hat{\bpi}
\Rightarrow
\nonumber \\
 \frac{1}{1 + e^{- \tilde{\omega}_j}} + \frac{\tilde{w}_j}{gp} = \frac{\tilde{\bz}_j^T \hat{\bpi}}{\sum_{i=1}^p \mbox{I}(\tilde{z}_{ij}=1)}
= \frac{\tilde{\bz}_{.j}^T \hat{\bpi}}{\tilde{\bz}_{.j}^T \mathbf{1}},
\nonumber
\end{align}
where $\tilde{\bz}_{.j}$ is the $j^{th}$ column in $\bZ$,
and we used that $\tilde{\bz}_i^T \tilde{\bomega}= \tilde{\omega}_j$, since $\tilde{\bz}_{il}=1$ for $l=j$ and $=0$ otherwise.
We hence have that $\tilde{\omega}_j= h(\tilde{\bz}_{.j}^T \hat{\bpi} / \tilde{\bz}_{.j}^T \mathbf{1}, gp)$, 
where $h(a,c)$ denotes the solution (in terms of $w$) to $1/(1+e^{-w}) + w/c= a$.
For the case $c= \infty$, $h(a,\infty)= \log(a/(1-a))$.
More generally, the solution to this univariate equation can be easily found via Newton-Raphson.

To complete the proof, define $\hat{\bomega}= \bU^{-1} \tilde{\bomega}$,
so that $\tilde{\bZ} \tilde{\bomega}= \bZ \bU^{-1} \bU \hat{\bomega}= \bZ \hat{\bomega}$,
and therefore ${\bf m}(\hat{\bomega})= \tilde{{\bf m}}(\tilde{\bomega})$.
By definition of $\tilde{\bomega}$ we have that 
\begin{align}
\tilde{\bZ}^T [ \tilde{{\bf m}}(\tilde{\bomega}) - \frac{\tilde{\bZ} \tilde{\bomega}}{gp}]= \tilde{\bZ}^T \hat{\bpi}
\Rightarrow
(\bU^{-1})^T [ \bZ^T {\bf m}(\hat{\bomega})  - \frac{\bZ \hat{\bomega}}{gp} ]= (\bU^{-1})^T \bZ^T \hat{\bpi}
\nonumber
\end{align}
which implies that $\bZ^T [{\bf m}(\hat{\bomega}) - (gp)^{-1} \bZ \hat{\bomega}] = \bZ^T \hat{\bpi}$, as we wished to prove.

\section{Run times and EM algorithm stability}
\label{sec:runtimes}

\begin{table}[ht]
\centering
\begin{tabular}{lrrrr}
  \hline
 & $p=100$ & $p=500$ & $p=1,000$ & $p=2,000$ \\ 
 & ($s=2.5$) & ($s=5.0$) & ($s=11.5$) & ($s=19.5$) \\ \hline
Beta-Binomial   ($L$=5,000) & 0.4 & 5.4 & 23.1 & 81.8 \\ 
Empirical Bayes ($L=$ 5,000; $M=$1,000) & 0.5 & 4.9 & 28.0 & 115.2 \\ 
Empirical Bayes ($L=$ 5,000; $M=$100) & 0.3 & 3.2 & 18.4 & 74.0 \\ 
Empirical Bayes ($L=$ 1,000; $M=$100) & 0.1 & 1.1 & 5.8 & 22.0 \\ 
   \hline
\end{tabular}
\caption{Run times (seconds) in an Ubuntu laptop (13th gen Intel core i7, 32Gb RAM) under simulation settings as in Section \ref{ssec:simstudy} with $n=500$ and $\bomega=(-4.6, 0.5, 0)$. $p$ is the total number of covariates, $s$ the average number that truly have an effect,
$L$ the number of MCMC iterations, and $M$ the number of MCMC iterations in the M-step}
\label{tab:runtimes}
\end{table}

\begin{table}
\centering
\begin{tabular}{|l|ccc|ccc|}\hline
$p$ & \multicolumn{3}{c|}{$\hat{\bomega}^{BB, M=1000} - \hat{\bomega}^{0, M=1000}$} & \multicolumn{3}{c|}{$\hat{\bomega}^{BB, M=1000} - \hat{\bomega}^{0, M=100}$} \\ \hline
100 & 0.00006 &  0.0001  & 0.0001 & 0.0001 &  0.00005 & 0.0001 \\
500 & 0.00337 &  0.0007  & 0.0004 & 0.0059 &  0.00205 & 0.0013  \\ 
\hline
\end{tabular}
\caption{Mean absolute difference in $\hat{\bomega}$ for Beta-Binomial initialization $\hat{\bomega}^{BB}$
vs. $\hat{\bomega}^{(0)}= {\bf 0}$, and $M=1000$ vs. $100$ M-step iterations }
\label{tab:dif_omega}
\end{table}

Table \ref{tab:runtimes} shows the run times for simulations
under the same settings as in Section \ref{ssec:simstudy} with $p \in \{100, 500, 1000, 2000 \}$, $n=500$ and data-generating $\bomega=(-4.6, 0.5, 0)$.
As a benchmark, the run time for $L=5,000$ Gibbs iterations under a standard Beta-Binomial prior was similar to that of our framework 
using $M=1,000$ iterations to estimate $\hat{\pi}(\gamma_j=1 \mid \by, \bomega)$ within the M-step and also $L=5,000$ Gibbs iterations.
The EM algorithm converged in $\leq 10$ iterations in all cases (often, $\leq 5$ iterations), hence $\hat{\bomega}$ was obtained relatively quickly and the main cost was running the $L=5,000$ iterations afterwards.
Table \ref{tab:runtimes} also reports run times for reducing $M=1,000$ to $M=100$, and then further reducing $L=5,000$ to  $L=1,000$. In the latter case the algorithm required 22 seconds to complete.
These reductions had a fairly mild effect on the estimates: the mean difference in $\hat{\pi}(\gamma_j=1 \mid \hat{\bomega})$ between using $L=5,000$ and $L=1,000$ Gibbs iterations was $<0.001$, and similarly for reducing $M=1,000$ to $M=100$.

Table \ref{tab:dif_omega} assesses the effect of $M$ and of the initial hyper-parameter estimate $\hat{\bomega}^{(0)}$ in the final $\hat{\bomega}$ returned by the EM algorithm.

\section{Proofs of Theorems \ref{thm:suffcondlinearmodel} and \ref{theo:empbayesselconsist}}

 \subsection{Auxiliary results}
 In this subsection we give a series of results instrumental to the proofs of Theorem \ref{thm:suffcondlinearmodel} and Theorem \ref{thm:ebayes_zerograd}. We start by recalling some notation and introducing new one. Recall, in \eqref{eq:bf_zellner_known} and \eqref{eq:kappa} we defined for any $\omega_1,\ldots,\omega_B$,
 \begin{equation}\label{eq:kappa2}
 \begin{aligned}
     & \hat{\btheta}_{\bgamma}= (\bX_{\bgamma}^T \bX_{\bgamma})^{-1} \bX_{\bgamma}^T \by \\
     & \kappa_b= \frac{1}{2} \log(1 + gn) + \log (1 / \omega_b - 1).
 \end{aligned}
 \end{equation}
 In addition, for any set of variable indices $\bgamma$, $\bgamma_b$ denotes the set of indices in $\bgamma$ that are in block $b$. For any two models $\bgamma$ and $\tilde{\bgamma}$, let 
 \begin{equation} \label{eq:muQM}
\begin{aligned}
    & \Delta(\bgamma, \tilde{\bgamma}):= \sum_{b=1}^B (p_{\bgamma,b} - p_{\tilde{\bgamma},b}) \kappa_b, \\
    & \bGamma({\bgamma},\tilde{\bgamma}) := \bgamma\cup \bgamma^*,\\
    & L(\bgamma, \tilde{\bgamma}) \,:=\, \frac{1}{\phi}\hat{\btheta}_{\bgamma}^T \bX_{\bgamma}^T \bX_{\bgamma}\hat{\btheta}_{\bgamma}
- \frac{1}{\phi}\hat{\btheta}_{\tilde{\bgamma}}^T \bX_{\tilde{\bgamma}}^T \bX_{\tilde{\bgamma}}\hat{\btheta}_{\tilde{\bgamma}}, \\
    & \mu(\bgamma,\tilde{\bgamma})\,:=\,\frac{1}{\phi}\|\big(I_n-P_{\tilde{\bgamma}}\big) \bX_{\bgamma \setminus \tilde{\bgamma}} \btheta^*_{\bgamma \setminus \tilde{\bgamma}}\|^2,
\end{aligned}
\end{equation}
where $P_{\tilde{\bgamma}}=\bX_{\tilde{\bgamma}}\left(\bX_{\tilde{\bgamma}}^T \bX_{\tilde{\bgamma}}\right)^{-1} \bX_{\tilde{\bgamma}}^T$.
Then $\Delta(\bgamma, \tilde{\bgamma})$ is the difference in penalty between in $\bgamma$ and $\tilde{\bgamma}$, $\bGamma({\bgamma},\tilde{\bgamma})$ is the union model between $\bgamma$ and $\tilde{\bgamma}$, $L(\bgamma \tilde{\bgamma})$ is the difference in the squared $\ell_2$ norms of the fitted values under $\bgamma$ and $\tilde{\bgamma}$. In addition, $\mu(\bgamma,\tilde{\bgamma})$ is the projection of the true signals in $\bgamma$ missing in $\tilde{\bgamma}$ in the space of residuals of $\tilde{\bgamma}$, it appears as a noncentrality parameter of the chi-squared distribution discussed in Lemma \ref{lemma:s7} below. 

For any set of $\omega_1,\ldots,\omega_B$ and corresponding $\kappa_1,\ldots,\kappa_B$, let
\begin{equation}\label{eq:defSes}
\begin{aligned}
    &S_b^{S}(\kappa) := \left\{j : z_j=b, \theta_j^* \neq 0, \sqrt{n\bar{\lambda}}|\theta_j^*|=o\big(\sqrt{\kappa_b}\big) \right\},\\
    &S_b^{L}(\kappa) := \left\{j: z_j=b, \theta_j^* \neq 0, \sqrt{\frac{(1-\nu) n\rho(\bX)}{6}}|\theta_j^*| \,- \,\sqrt{\kappa_b} \;=\; \sqrt{\log(s_b)} + g_b\right\},\\
    &S_b^{I}(\kappa) := \left\{ j: z_j=b, \theta_j^* \neq 0\right\} \setminus \big(S_b^{L}(\kappa) \cup S_b^{S}(\kappa)\big),
\end{aligned}
\end{equation}
where $\bar{\lambda}$ is the largest eigenvalue of $\bX_{\bgamma^*}^T \bX_{\bgamma^*}$. Finally, let $\T(\kappa)$ be the set of models that contain all large signals in $S^{L}(\kappa)=\cup_{b=1}^b S_b^{L}(\kappa)$, and neither truly inactive parameters nor small signals in $S^S(\kappa)=\cup_{b=1}^b S^{S}_b(\kappa)$. That is,
\begin{equation} \label{eq:Tkappa}
   \T(\kappa) \;:=\;\Big\{\bgamma \,|\, S^L(\kappa)\subseteq\gamma \subseteq S^L(\kappa)\cup S^{I}(\kappa)\Big\}.
\end{equation}

\subsubsection{Distributional properties and tail bounds}

\begin{lemma}\label{lemma:s7}
    Let $\bgamma,\tilde{\bgamma}$ be any two nested models such that $\bgamma \subseteq \tilde{\bgamma}$. Then
$$
L(\tilde{\bgamma},\bgamma)\;\sim\; \chi_{|\tilde{\bgamma}\setminus \bgamma|}^2\left(\mu(\tilde{\bgamma}, \bgamma)\right), 
$$
where $\chi_{k}^2(\mu)$ denotes, when $\mu>0$, the noncentral chi-squared distribution with $k$ degrees of freedom and noncentrality parameter $\mu$ and, when $\mu=0$, the chi-squared distribution with $k$ degrees of freedom $\chi_{k}^2$.
\end{lemma}

\begin{proof} Proved as Lemma S7 in \cite{rossell:2022}.
\end{proof}

\begin{lemma}\label{lem:boundrhogen}
 For any $\tilde{\bgamma}\subseteq \bgamma^*$ such that $\tilde{\bgamma}\not\subseteq \bgamma$, let $\bGamma(\bgamma,\tilde{\bgamma})=\bgamma \cup \tilde{\bgamma}$. The non-centrality parameter defined in \eqref{eq:muQM} satisfies: 
\begin{equation}
    \mu_{\bGamma(\bgamma,\tilde{\bgamma})\bgamma}\; \geq \;n\,\phi^{-1}\,\rho(\bX)\,\sum_{j=1}^b|\tilde{\bgamma}_b\setminus \bgamma_b|\,\min_{i\in \tilde{\bgamma}_b\setminus \bgamma_b}{\theta_{i}^*}^2.
\end{equation}
\end{lemma}


\begin{proof} 
The non-centrality parameter $\mu(\bGamma(\bgamma,\tilde{\bgamma}),\bgamma)$, as defined in \eqref{eq:muQM}, satisfies
\begin{eqnarray*}
    \mu(\bGamma(\bgamma,\tilde{\bgamma}),\bgamma)&=&\phi^{-1}\|\big(I_n-P_{\bgamma}\big) \bX_{\bGamma(\bgamma,\tilde{\bgamma})\setminus \bgamma} \btheta^*_{\bGamma(\bgamma,\tilde{\bgamma})\setminus \bgamma}\|^2\\
    &=&\phi^{-1}\|\big(I_n-P_{\bgamma}\big) \bX_{\tilde{\bgamma}\setminus \bgamma} \btheta^*_{\tilde{\bgamma}\setminus \bgamma}\|^2\\
    &= & n\phi^{-1} {\theta^*_{\tilde{\bgamma}\setminus \bgamma}}^T \left(\tfrac1n \bX_{\tilde{\bgamma}\setminus \bgamma}^T(I_n-P_{\bgamma})\bX_{\tilde{\bgamma}\setminus \bgamma}\right)\theta^*_{\tilde{\bgamma}\setminus \bgamma}\\
    &\geq & n\phi^{-1}\lambda_{\min}\left(\tfrac1n \bX_{\tilde{\bgamma}\setminus \bgamma}^T(I_n-P_{\bgamma})\bX_{\tilde{\bgamma}\setminus \bgamma}\right)\|\theta^*_{\tilde{\bgamma}\setminus \bgamma}\|^2,
\end{eqnarray*}
where the second equality follows from observing that $\bGamma(\bgamma,\tilde{\bgamma})\setminus \bgamma=\tilde{\bgamma}\setminus \bgamma$. Since $\tilde{\bgamma}$ is a subset of $\bgamma^*$, by reordering columns, $X_{\bgamma^*\setminus \bgamma}=[X_{\tilde{\bgamma}\setminus \bgamma},X_{\bgamma^*\setminus (\bgamma\cup \tilde{\bgamma})}]$ and therefore $\tfrac1n \bX_{\tilde{\bgamma}\setminus \bgamma}^T(I_n-P_{\bgamma})\bX_{\tilde{\bgamma}\setminus \bgamma}$ is a principal submatrix of $\tfrac1n \bX_{\bgamma^*\setminus \bgamma}^T(I_n-P_{\bgamma})\bX_{\bgamma^*\setminus \bgamma}$. Hence, Cauchy's interlacing theorem gives that $\lambda_{\min}\left(\tfrac1n \bX_{\tilde{\bgamma}\setminus \bgamma}^T(I_n-P_{\bgamma})\bX_{\tilde{\bgamma}\setminus \bgamma}\right) \geq \lambda_{\min}\left(\tfrac1n \bX_{\bgamma^*\setminus \bgamma}^T(I_n-P_{\bgamma})\bX_{\bgamma^*\setminus \bgamma}\right)$. Finally, by definition of $\rho(\bX)$ in \eqref{eq:rhoX} we have that
$\lambda_{\min}\left(\tfrac1n \bX_{\bgamma^*\setminus \bgamma}^T(I_n-P_{\bgamma})\bX_{\bgamma^*\setminus \bgamma}\right) \geq \rho(\bX)$, and further noting that $\|\theta^*_{\tilde{\bgamma}\setminus \bgamma}\|^2\geq \sum_{j=1}^b|\tilde{\bgamma}_b\setminus \bgamma_b|\,\min_{i\in \tilde{\bgamma}_b\setminus \bgamma_b}{\theta_{i}^*}^2$ gives the desired result.
\end{proof}

\begin{lemma}\label{lem:upperboundnoncentralgen}
    For any $\bgamma$ and $\tilde{\bgamma}\subseteq \bgamma^*$, 
    \begin{equation}\label{eq:upperbounnnnnd}
      \mu( \bGamma({\bgamma},\tilde{\bgamma}), \tilde{\bgamma}) \leq n\,\phi^{-1}\,\bar{\lambda} \,\sum_{j=1}^b|(\bgamma^*_b\cap \bgamma_b )\setminus \tilde{\bgamma}_b|\,\max_{i\in (\bgamma^*_b\cap \bgamma_b )\setminus \tilde{\bgamma}_b}{\theta_{i}^*}^2 
    \end{equation}
    where $\bar{\lambda}\,:=\,\lambda_{\max}\Big(\frac1n \bX_{\bgamma^*}^\top\bX_{\bgamma^*}\Big)$, and
    \begin{equation}\label{eq:lowerbounnnnnd}
    \mu( \bGamma({\bgamma},\tilde{\bgamma}),\tilde{\bgamma}) \geq n\,\phi^{-1}\,\lambda_{\min}\Big(\frac1n \bX_{\bgamma^*}^\top\bX_{\bgamma^*}\Big)\,\sum_{j=1}^b|(\bgamma^*_b\cap \bgamma_b )\setminus \tilde{\bgamma}_b|\,\min_{i\in (\bgamma^*_b\cap \bgamma_b )\setminus \tilde{\bgamma}_b}{\theta_{i}^*}^2
    \end{equation}
\end{lemma}


\begin{proof}
Using the definition of $\mu(\bGamma({\bgamma},\tilde{\bgamma}), \tilde{\bgamma})$ in \eqref{eq:muQM}, we have that 
\begin{eqnarray}
    \mu( \bGamma({\bgamma},\tilde{\bgamma})&=&\phi^{-1}{\btheta^*_{ \bGamma({\bgamma},\tilde{\bgamma})\setminus \tilde{\bgamma}}}^T\bX_{ \bGamma({\bgamma},\tilde{\bgamma})\setminus \tilde{\bgamma}}^T  \big(I_n-P_{\tilde{\bgamma}}\big) \bX_{ \bGamma({\bgamma},\tilde{\bgamma})\setminus \tilde{\bgamma}} \btheta^*_{ \bGamma({\bgamma},\tilde{\bgamma})\setminus \tilde{\bgamma}}\nonumber\\
    &=&\phi^{-1}{\btheta^*_{\bgamma\setminus \tilde{\bgamma}}}^T\bX_{\bgamma\setminus \tilde{\bgamma}}^T  \big(I_n-P_{\tilde{\bgamma}}\big) \bX_{\bgamma\setminus \tilde{\bgamma}} \btheta^*_{\bgamma\setminus \tilde{\bgamma}}\nonumber\\
    &=&\phi^{-1}{\btheta^*_{(\bgamma^*\cap \bgamma)\setminus \tilde{\bgamma}}}^T\bX_{(\bgamma^*\cap \bgamma)\setminus \tilde{\bgamma}}^T  \big(I_n-P_{\tilde{\bgamma}}\big) \bX_{(\bgamma^*\cap \bgamma)\setminus \tilde{\bgamma}} \btheta^*_{(\bgamma^*\cap \bgamma)\setminus \tilde{\bgamma}}\label{eq:poihe}
\end{eqnarray}
where the second equality follows from $ \bGamma({\bgamma},\tilde{\bgamma})\setminus \tilde{\bgamma} =\bgamma\setminus \tilde{\bgamma} $ and the third equality from $\btheta^*_{\bgamma\setminus \bgamma^*}=0$.
We start by showing the upper bound in \eqref{eq:upperbounnnnnd}. Denote for any square matrix $A$, its largest eigenvalue 
$\lambda_{\max}(A)$. By \eqref{eq:poihe}, we have that
$$\mu( \bGamma({\bgamma},\tilde{\bgamma})\leq n\phi^{-1}\lambda_{\max}\Big(\frac{1}{n}\bX_{(\bgamma^*\cap \bgamma)\setminus \tilde{\bgamma}}^T  \big(I_n-P_{\tilde{\bgamma}}\big) \bX_{(\bgamma^*\cap \bgamma)\setminus \tilde{\bgamma}}\Big) \|\btheta^*_{(\bgamma^*\cap \bgamma)\setminus \tilde{\bgamma}}\|^2.$$
Let $B:=\tfrac{1}{n}\bX_{(\bgamma^*\cap \bgamma)\setminus \tilde{\bgamma}}^T  \big(I_n-P_{\tilde{\bgamma}}\big) \bX_{(\bgamma^*\cap \bgamma)\setminus \tilde{\bgamma}}$, $C:=\tfrac{1}{n}\bX_{(\bgamma^*\cap \bgamma)\cup \tilde{\bgamma}}^T\bX_{(\bgamma^*\cap \bgamma)\cup \tilde{\bgamma}}$ and $D:=\tfrac{1}{n}\bX_{\tilde{\bgamma}}^T \bX_{\tilde{\bgamma}}$. $D$ is a principal submatrix of $C$ and $B$ is the Schur complement of $D$ of $C$. The inverse $B^{-1}$ is then a principal submatrix of $C^{-1}$, and by Cauchy's interlacing theorem we have that   $\lambda_{\min}(B^{-1})\geq \lambda_{\min}(C^{-1})$ and then $\lambda_{\max}(B)\leq \lambda_{\max}(C)$. Since $\tilde{\bgamma} \subseteq \bgamma^*$ by assumption, we also have that $(\bgamma^*\cap \bgamma)\cup \tilde{\bgamma} \subseteq \bgamma^*$, then by interlacing again $\lambda_{\max}(C)\leq \lambda_{\max}\Big(\tfrac{1}{n}\bX_{\bgamma^*}^T\bX_{\bgamma^*}\Big)=\bar{\lambda}$. The upper bound in \eqref{eq:upperbounnnnnd} follows from the latter inequality and also observing that $\|\btheta^*_{(\bgamma^*\cap \bgamma)\setminus \tilde{\bgamma}}\|_2^2\leq \sum_{j=1}^b |(\bgamma^*_b\cap \bgamma_b)\setminus \tilde{\bgamma}_b|\max_{i\in (\bgamma^*_b\cap \bgamma_b)\setminus \tilde{\bgamma}_b}{\theta_i^*}^2$.

We now derive the lower bound in \eqref{eq:lowerbounnnnnd}. By \eqref{eq:poihe}, we have that 
$$\mu( \bGamma({\bgamma},\tilde{\bgamma})\geq n\phi^{-1}\lambda_{\min}\Big(\frac{1}{n}\bX_{(\bgamma^*\cap \bgamma)\setminus \tilde{\bgamma}}^T  \big(I_n-P_{\tilde{\bgamma}}\big) \bX_{(\bgamma^*\cap \bgamma)\setminus \tilde{\bgamma}}\Big) \|\btheta^*_{(\bgamma^*\cap \bgamma)\setminus \tilde{\bgamma}}\|_2^2.$$

Recall that $B^{-1}$ is a principal submatrix of $C^{-1}$, hence by interlacing $\lambda_{\max}(B^{-1})\leq \lambda_{\max}(C^{-1})$ and $\lambda_{\min}(B)\geq \lambda_{\min}(C)$. Since $\tilde{\bgamma} \subseteq \bgamma^*$ by assumption, we have that $(\bgamma^*\cap \bgamma)\cup \tilde{\bgamma} \subseteq \bgamma^*$, and hence $\lambda_{\min}(C)\geq\lambda_{\min}(\frac{1}{n}\bX_{\bgamma^*}^T \bX_{\bgamma^*})$.
The bound in \eqref{eq:lowerbounnnnnd} is obtained by using the latter inequality and noting that also $\|\btheta^*_{(\bgamma^*\cap \bgamma)\setminus \tilde{\bgamma}}\|_2^2\geq \sum_{j=1}^b |(\bgamma^*_b\cap \bgamma_b)\setminus \tilde{\bgamma}_b|\min_{i\in (\bgamma^*_b\cap \bgamma_b)\setminus \tilde{\bgamma}_b}{\theta_i^*}^2.$
\end{proof}

\begin{lemma}\label{lemma:s1s3}
Let $W\sim\chi_{\nu}^2(\mu)$ with $\mu \geq 0$, then for any $w>\mu+\nu$
$$P(W > w) \leq e^{-\big(\frac{w+\mu}{2}  - \sqrt{2w( 2 \mu+ \nu)-2 \mu \nu - \nu^2 } \big)}. $$
Moreover, assume $w$, $\nu$ and $\mu$ are functions of $n$ such that $w$ is increasing, $\nu=o(w)$, and $\mu=o(w)$. Then, for any $\phi\in(0,1)$ and $n$ large enough
$$P(W > w)  \leq  e^{-\phi\frac{w}{2}}.$$
\end{lemma}

\begin{proof} Proved as Lemma S1.4. in \cite{rognon:2025}.
\end{proof}

\begin{lemma}\label{lemma:s2} Let $W \sim \chi_\nu^2(\mu)$ with $\mu>0$. For any $w<\mu$, 
$$
P(W<w) \leq \frac{e^{-\frac{1}{2}(\sqrt{\mu}-\sqrt{w})^2}}{(\mu / w)^{\nu / 4}}.$$
\end{lemma}

\begin{proof} Proved as Lemma S1.5. in \cite{rognon:2025}.
\end{proof}

\begin{lemma}\label{lemma:s18} 
    Let $W\sim \chi_\eta^2(\mu)$ with $\mu \geq 0$. Assume that $h$, $c$, $\eta$ and $\mu$ are functions of $n$ such that $h$ is positive and increasing, $c$ is positive, $\eta=o(c\log(h))$, and $\mu=o(c\log(h))$. Let $\bar{u},\underline{u}$ in $(0,1)$ such that $1>\bar{u}>\underline{u}\geq\left(1+h^{\psi}e^{-(\eta+\mu)/c}\right)^{-1}$ where $\psi\in(0,1)$, then for every $n$ large enough, we have 
$$
\int_{\underline{u}}^{\bar{u}} P\left(W>c\log \left(\frac{h}{1 / u-1}\right)\right) d u \leq  \frac{1}{h^{\psi}}\Big(\bar{u}-\underline{u}+\log\Big(\frac{\bar{u}}{\underline{u}}\Big)\Big).
$$
\end{lemma}


\begin{proof}
For any $u\in [\underline{u},\bar{u}]$, we have $c\log \left(\frac{h}{1 / u-1}\right) \geq c\log \left(\frac{h}{1 / \underline{u}-1}\right)$. Since $\underline{u}\geq (1+h^{\psi}e^{-(\eta+\mu)/c})^{-1}$ by assumption  we also have that, for any $u\in (\underline{u},\bar{u})$,
$$c\log \left(\frac{h}{1 / u-1}\right) \geq c(1-\psi)\log(h)+\eta+\mu.$$
It follows,
$$\frac{\eta}{c\log \left(\frac{h}{1 / u-1}\right)}\leq \frac{\eta}{c(1-\psi)\log(h)+\eta+\mu}
\quad\text{and}\quad
\frac{\mu}{c\log \left(\frac{h}{1 / u-1}\right)}\leq\frac{\mu}{ c(1-\psi)\log(h)+\eta+\mu}.$$
By assumption $\eta=o\left(c\log (h)\right)$ and $\mu=o\left(c\log (h)\right)$, then for any $u\in (\underline{u},\bar{u})$, $\eta=o\left(c\log \left(\frac{h}{1 / u-1}\right)\right)$ and $\mu=o\left(c\log \left(\frac{h}{1 / u-1}\right)\right)$. By Lemma~\ref{lemma:s1s3}, for any $\psi\in(0,1)$ and every $n$ large enough,
\begin{equation}\label{eq:randomm}
    \int_{\underline{u}}^{\bar{u}} P\left(W>c\log \left(\frac{h}{1 / u-1}\right)\right) d u< \frac{1}{h^{\psi}} \int_{\underline{u}}^{\bar{u}}(1 / u-1)^{\psi} d u .
\end{equation}
Applying the change of variables $v=1 / u-1$ to the integral on the right-hand side above gives
\begin{equation}\label{eq:rpoig}
    \int_{\underline{u}}^{\bar{u}}(1 / u-1)^{\psi} d u=\int_{1 / \bar{u}-1}^{1 / \underline{u}-1} \frac{v^{\psi}}{(v+1)^2} d v.
\end{equation}
Rewrite $v^{\psi}=(v-1 + 1)^{\psi}$. Since $\bar{u}<1$, we have that for any $v>1 / \bar{u}-1$, $v-1>-1$. Note that for any $x \geq-1$ and $r\in[0,1]$ $(1+x)^r \leq 1+r x$. Then, for any $v>1 / \bar{u}-1$, $v^{\psi}=(v-1 + 1)^{\psi}\leq 1 + \psi (v-1)\leq 1 + \psi(v+1)$. Applying this last inequality to the right-hand side in \eqref{eq:rpoig} gives
\begin{equation}\label{eq:random2}
\int_{\underline{u}}^{\bar{u}}(1 / u-1)^{\psi} d u<\int_{1 / \bar{u}-1}^{1 / \underline{u}-1} \frac{1}{(v+1)^2} +\frac{\psi}{v+1} d v = \bar{u}-\underline{u}+\psi\log\Big(\frac{\bar{u}}{\underline{u}}\Big).
\end{equation}
The result follows inputting the bound from \eqref{eq:random2} in \eqref{eq:randomm} and using that $\psi<1$ and $\log(\bar{u}/\underline{u}) \geq 0$ ($\underline{u} \leq \bar{u}$).
\end{proof}

\begin{lemma}\label{lemma:s20}
 Let $W\sim \chi_\eta^2(\mu)$ with $\mu \geq 0$. Assume that $h$, $c$, $\eta$ and $\mu$ are functions of $n$ such that $h$ is positive and increasing, $c$ is positive, $\eta=o(c\log(h))$, and $\mu=o(c\log(h))$. Then for any $\alpha\in(0,1)$ and every large enough $n$, we have
$$
\int_0^1 P\left(W>c \log \left(\frac{h}{1 / u-1}\right)\right) d u = o\big(  h^{-\alpha}\big).$$
\end{lemma}


\begin{proof}
Since a probability is bounded by $1$, for any $a\in(0,1)$,
\begin{equation}\label{eq:basebound}
    \int_0^1 P\left(W>c \log \left(\frac{h}{1 / u-1}\right)\right) d u \;\leq\; 2 a+\int_{a}^{1-a} P\left(W>c\log\left(\frac{h}{1 / u-1}\right)\right) d u.
\end{equation}

Take $a=\left(1+h^{\psi}e^{-(\eta+\mu)/ c}\right)^{-1}$ for some $\psi\in(\alpha,1)$. By Lemma~\ref{lemma:s18} with $\underline{u}=a$ and $\bar{u}=1-a$, we have that
$$
\int_0^1 P\left(W>c \log \left(\frac{h}{1 / u-1}\right)\right) d u \;\leq\; \frac{2}{1+h^{\psi}e^{-(\eta+\mu)/c}}+ 
\frac{1-2a+\log(h^{\psi}e^{-(\eta+\mu)/c})}{h^{\psi}}. $$
Since $1+h^{\psi}e^{-(\eta+\mu)/ c}>h^{\psi}e^{-(\eta+\mu)/ c}$ and $1-2a-\frac{\eta+\mu}{c}\leq \log(h^\psi)$ for every $n$ large enough, we have
$$
\int_0^1 P\left(W>2 \log \left(\frac{h}{1 / u-1}\right)\right) d u \;\leq\; h^{-\psi}\big(2e^{(\eta+\mu)/ c}+2\log(h^{\psi})\big) $$
We have that
\begin{equation}\label{eq:random24}
    \frac{h^{-\psi}2e^{(\eta+\mu)/ c}}{h^{-\alpha}} 
= e^{-(\psi-\alpha) \log(h) + \log(2) + (\eta + \mu)/c}
    \;=\;e^{-(\psi-\alpha)\log(h)\big(1-\frac{\eta+\mu}{c(\psi-\alpha)\log(h)}- \frac{\log(2)}{(\psi-\alpha)\log(h)}\big)}
\end{equation}
and similarly that
\begin{equation}\label{eq:random25}
    \frac{h^{-\psi}2\log(h^{\psi})}{h^{-\alpha}} \;=\;e^{-(\psi-\alpha)\log(h)\big(1-\frac{\log(\psi\log(h))}{(\psi-\alpha)\log(h)}- \frac{\log(2)}{(\psi-\alpha)\log(h)}\big)}.
\end{equation}
Since $\alpha<\psi$ as stated above, and by assumption $h$ is increasing, $\eta=o(c\log(h))$ and $\mu=o(c\log(h))$, both expressions in \eqref{eq:random24} and \eqref{eq:random25} vanish as $n$ grows. Hence, 
$$\int_0^1 P\left(W>2 \log \left(\frac{h}{1 / u-1}\right)\right) d u \;=\; o(h^{-\alpha}).$$
\end{proof}

\subsubsection{Bound on the expected posterior probability of an individual model}

\begin{lemma}\label{lem:boundexpnc}
For any $\tilde{\bgamma}\subseteq \bgamma^*$ and $\bgamma \neq \tilde{\bgamma}$, denote $A_{\tilde{\bgamma}}:=\nu\Delta(\bgamma,\tilde{\bgamma})+\tfrac{1-\nu}{6}\mu(\bGamma(\bgamma,\tilde{\bgamma}),\bgamma)$ (\emph{cf} \eqref{eq:muQM}). Suppose that, for some $\nu \in (1/2,1]$, it holds that $A_{\tilde{\bgamma}}>0$, and $|\bgamma \setminus \tilde{\bgamma}|=o(A_{\bgamma})$. For any $\psi\in(0,1)$ and every $n$ large enough,
$$E\left(\pi(\bgamma\mid \by, \bomega)\right)  \;\;\leq\;\; e^{-\psi A_{\bgamma}}.$$

\end{lemma}

\begin{proof}
For any $\bgamma\neq \tilde{\bgamma}$, since $\sum_{\bgamma}\pi(\bgamma \mid \by, \bomega)=1$, we have the following bound on the posterior probability
\begin{equation}\label{eq:saknv}
    \pi(\bgamma \mid \by, \bomega)= \frac{\pi(\bgamma \mid \by, \bomega)} { \sum_{\delta}\pi(\delta\mid \by, \bomega)} = \bigg[1+\sum_{\delta:\delta\neq \bgamma}\frac{\pi(\delta\mid \by, \bomega)}{ \pi(\bgamma \mid \by, \bomega)}\bigg]^{-1}
\leq \bigg[1+\frac{\pi(\tilde{\bgamma}\mid \by, \bomega)}{ \pi(\bgamma \mid \by, \bomega)}\bigg]^{-1}
\end{equation}
where the last inequality follows from $\sum_{\delta:\delta\notin \{\bgamma,\tilde{\bgamma}\}}\pi(\delta\mid \by, \bomega)/ \pi(\bgamma \mid \by, \bomega) \geq 0$.
By \eqref{eq:saknv} and using that for any random variable $Z \geq 0$, $E(Z)=\int_0^\infty P(Z>u) du$, we have that
\begin{equation*}
    E \left[ \pi(\bgamma\mid \by, \bomega) \right] \leq \int_0^1 P\left( \bigg[1+\frac{\pi(\tilde{\bgamma}\mid \by, \bomega)}{ \pi(\bgamma \mid \by, \bomega)}\bigg]^{-1}\geq u\right){\rm d} u,
\end{equation*}
and inverting and taking logarithm on both sides of the inequality within the probability in the right-hand side, we get
\begin{equation}\label{eq:fsan}
    E \left[ \pi(\bgamma\mid \by, \bomega) \right] \leq \int_0^1 P\left( \log \frac{\pi(\tilde{\bgamma}\mid \by, \bomega)}{ \pi(\bgamma \mid \by, \bomega)}\leq \log\Big(\frac 1u-1\Big)\right){\rm d} u.
\end{equation}
Simple algebra shows that
\begin{equation}\label{eq:pdSVN}
\frac{\pi(\tilde{\bgamma} \mid \by,\bomega)}{\pi(\bgamma \mid \by, \bomega)}=
\exp \left\{ -\frac{g n (\hat{\btheta}_{\bgamma}^T \bX_{\bgamma}^T \bX_{\bgamma}\hat{\btheta}_{\bgamma} - \hat{\btheta}_{\tilde{\bgamma}}^T \bX_{\tilde{\bgamma}}^T \bX_{\tilde{\bgamma}}\hat{\btheta}_{\tilde{\bgamma}})}{2\phi (1+g n)} \right\}
\prod_{b=1}^B \left( \frac{\omega_b}{(1 + gn)^{1/2} (1 - \omega_b)} \right)^{p_{\tilde{\bgamma},b} - p_{\bgamma,b}}
\end{equation}
Inputting \eqref{eq:pdSVN} in \eqref{eq:fsan}, and using the notation in \eqref{eq:kappa2} and \eqref{eq:muQM}, we have
\begin{eqnarray*}
 E \left[ \pi(\bgamma\mid \by, \bomega) \right] &\leq&\int_0^1 P\left(\frac{g n L(\tilde{\bgamma},\bgamma)}{2 (1+g n)} +\Delta(\bgamma,\tilde{\bgamma})\leq \log{(\tfrac1u-1)}\right){\rm d}u\\
  &=&\int_0^1 P\left(-\frac{g n L(\tilde{\bgamma},\bgamma)}{2 (1+g n)} -  \Big(\Delta(\bgamma,\tilde{\bgamma})-\log{(\tfrac1u-1)}\Big)\geq 0  \right) {\rm d}u.
\end{eqnarray*}

The second step of the proof is to use the union bound to upper bound the probability in the integrand above.
By adding and resting $ \phi^{-1}\hat{\btheta}_{\bGamma(\bgamma,\tilde{\bgamma})}^T \bX_{\bGamma(\bgamma,\tilde{\bgamma})}^T \bX_{\bGamma(\bgamma,\tilde{\bgamma})}\hat{\btheta}_{\bGamma(\bgamma,\tilde{\bgamma})}$ to $L(\tilde{\bgamma},\bgamma)$, we get   $$L(\tilde{\bgamma},\bgamma)=L(\bGamma(\bgamma,\tilde{\bgamma}),\bgamma)-L(\bGamma(\bgamma,\tilde{\bgamma}),\tilde{\bgamma}).$$ 
For any $\nu \in \mathrm{R}$,
\begin{equation*}
   \begin{aligned}
     -\frac{g n L(\tilde{\bgamma},\bgamma)}{2 (1+g n)} -  \Big(\Delta(\bgamma,\tilde{\bgamma})-\log{(\tfrac1u-1)}\Big)\;=
     \;&\left(\frac{gnL(\bGamma(\bgamma,\tilde{\bgamma}),\bgamma)}{2(1+gn)}-\log{\Bigl(\frac{e^{\nu\Delta(\bgamma,\tilde{\bgamma})}}{\tfrac1u-1}\Bigl)}\right)\\
     &-\Bigl(\frac{gnL(\bGamma(\bgamma,\tilde{\bgamma}),\tilde{\bgamma})}{2(1+gn)}+(1-\nu)\Delta(\bgamma,\tilde{\bgamma})\Bigl).
\end{aligned} 
\end{equation*}
Observe that for any random variables $U$, $V$, and any $\epsilon,\nu' \geq 0$, the event $
\{U-V \geq 0\}$ implies $\{U \geq \nu' \epsilon\}\cup
\{V  < \nu' \epsilon\}$. Let $U=\frac{gnL(\bGamma(\bgamma,\tilde{\bgamma}),\bgamma)}{2(1+gn)}-\log{\Bigl(\frac{e^{\nu\Delta(\bgamma,\tilde{\bgamma})}}{1/u-1}\Bigl)}$ and $V=\frac{gnL(\bGamma(\bgamma,\tilde{\bgamma}),\tilde{\bgamma})}{2(1+gn)}+(1-\nu)\Delta(\bgamma,\tilde{\bgamma})$. Take $\epsilon=\mu(\bGamma(\bgamma,\tilde{\bgamma}),\bgamma)$ and $\nu'=\tfrac16({1-\nu})$, and observe that
$A_{\tilde{\bgamma}}:=\nu\Delta(\bgamma,\tilde{\bgamma})+\tfrac{1-\nu}{6}\mu(\bGamma(\bgamma,\tilde{\bgamma}),\bgamma)= \nu\Delta(\bgamma,\tilde{\bgamma})+\nu'\mu(\bGamma(\bgamma,\tilde{\bgamma}),\bgamma)$. We then have
\begin{align*}
\{U\geq \nu'\epsilon\}\;&=\;\left\{\frac{gnL(\bGamma(\bgamma,\tilde{\bgamma}),\bgamma)}{2(1+gn)}\geq \log{\Bigl(\frac{e^{\nu\Delta(\bgamma,\tilde{\bgamma})}}{\tfrac1u-1}\Bigl)}+\nu'\mu(\bGamma(\bgamma,\tilde{\bgamma}),\bgamma)\right\}\\\;&
=\;\left\{\frac{gnL(\bGamma(\bgamma,\tilde{\bgamma}),\bgamma)}{2(1+gn)}\geq \log{\Bigl(\frac{e^{A_{\tilde{\bgamma}}}}{\tfrac1u-1}\Bigl)}\right\}
\end{align*}
and,
\begin{align*}
\{V<\nu'\epsilon\}\;&=\;\left\{\frac{gnL(\bGamma(\bgamma,\tilde{\bgamma}),\tilde{\bgamma})}{2(1+gn)}<-(1-\nu)\Delta(\bgamma,\tilde{\bgamma})+\nu'\mu(\bGamma(\bgamma,\tilde{\bgamma}),\bgamma)\right\}\;\\
&=\;\left\{\frac{gnL(\bGamma(\bgamma,\tilde{\bgamma}),\tilde{\bgamma})}{2(1+gn)}<\nu'(\mu(\bGamma(\bgamma,\tilde{\bgamma}),\bgamma)-6\Delta(\bgamma,\tilde{\bgamma}))\right\}.
\end{align*}
By the union bound we have that
\begin{equation}\label{eq:prooflemmas20nonspu}
  \begin{aligned}
    E \left[ \pi(\bgamma\mid \by, \bomega) \right]   \leq 
    & \int_0^1P\Bigl(\frac{gnL(\bGamma(\bgamma,\tilde{\bgamma}),\bgamma)}{2(1+gn)}\geq \log{\Bigl(\frac{e^{A_{\tilde{\bgamma}}}}{\tfrac1u-1}\Bigl)}\Bigl){\rm d}u \\
    &+ P\Bigl(\frac{gnL(\bGamma(\bgamma,\tilde{\bgamma}),\tilde{\bgamma})}{2(1+gn)}<\nu'(\mu(\bGamma(\bgamma,\tilde{\bgamma}),\bgamma)-6\Delta(\bgamma,\tilde{\bgamma}))\Bigl).
\end{aligned}  
\end{equation}

The third and final step of the proof is to upper bound each of the terms in the right-hand side of \eqref{eq:prooflemmas20nonspu}. The intuition is that both $\tilde{\bgamma}$ and $\bgamma$ are nested within $\bGamma(\bgamma,\tilde{\bgamma})$, and therefore $L(\bGamma(\bgamma,\tilde{\bgamma}),\tilde{\bgamma})$ and $L(\bGamma(\bgamma,\tilde{\bgamma}),\bgamma)$ follow chi-squared distributions.
We first bound the first term. 
If $\bgamma \subset \tilde{\bgamma}$ then $\bGamma(\bgamma,\tilde{\bgamma})=\tilde{\bgamma}$, $L(\bGamma(\bgamma,\tilde{\bgamma}),\tilde{\bgamma})=0$, and this term is zero. Suppose now that $\bgamma \not\subset \tilde{\bgamma}$. 
Then, by Lemma~\ref{lemma:s7}, 
$L(\bGamma(\bgamma,\tilde{\bgamma}),\tilde{\bgamma}) \sim \mathcal{\chi}_{|\bGamma(\bgamma,\tilde{\bgamma})\setminus \tilde{\bgamma}|}^2(\mu(\bGamma(\bgamma,\tilde{\bgamma}),\tilde{\bgamma}))$ with $|\bGamma(\bgamma,\tilde{\bgamma})\setminus \tilde{\bgamma}|=|\bgamma\setminus \tilde{\bgamma}|$. By assumption, $A_{\tilde{\bgamma}}>0$, $|\bgamma\setminus \tilde{\bgamma}|=o(\log(e^{A_{\tilde{\bgamma}}}))$ and $\mu(\bGamma(\bgamma,\tilde{\bgamma}),\tilde{\bgamma})=o(\log(e^{A_{\tilde{\bgamma}}}))$, then  by Lemma~\ref{lemma:s20}, for $\alpha\in(\psi,1)$ and $c=2(1+gn)/gn>0$, and every $n$ large enough,
\begin{equation}\label{eq:boundLQS}
\int_0^1P\Bigl(\frac{gnL(\bGamma(\bgamma,\tilde{\bgamma}),\bgamma)}{2(1+gn)}>\log{\Bigl(\frac{e^{A_{\tilde{\bgamma}}}}{1 / u-1}\Bigl)}\Bigl){\rm d}u \;<\; 
e^{-\alpha A_{\tilde{\bgamma}}}.
\end{equation}

We now bound the second term in \eqref{eq:prooflemmas20nonspu}. If $\bgamma \supset \tilde{\bgamma}$, then $\bGamma(\bgamma,\tilde{\bgamma})=\bgamma$, $L(\bGamma(\bgamma,\tilde{\bgamma}),\bgamma)=0$, $\mu(\bGamma(\bgamma,\tilde{\bgamma}),\bgamma)=0$, and this term is zero. Alternatively, if $\bgamma \not\supset \tilde{\bgamma}$ 
then, by Lemma~\ref{lemma:s7}, $L(\bGamma(\bgamma,\tilde{\bgamma}),\bgamma) \sim \chi_{|\bGamma(\bgamma,\tilde{\bgamma})\setminus \bgamma|}^2(\mu(\bGamma(\bgamma,\tilde{\bgamma}),\bgamma))$ with $|\bGamma(\bgamma,\tilde{\bgamma})\setminus \bgamma|=|\tilde{\bgamma}\setminus \bgamma|$.
Clearly, when $\mu(\bGamma(\bgamma,\tilde{\bgamma}),\bgamma)\leq 6\Delta(\bgamma,\tilde{\bgamma})$ this term is also zero, so suppose that $\mu(\bGamma(\bgamma,\tilde{\bgamma}),\bgamma)> 6\Delta(\bgamma,\tilde{\bgamma})$.  
We have  
$$P\left(\frac{gnL(\bGamma(\bgamma,\tilde{\bgamma}),\tilde{\bgamma})}{2(1+gn)}<\nu'(\mu(\bGamma(\bgamma,\tilde{\bgamma}),\bgamma)-6\Delta(\bgamma,\tilde{\bgamma}))\right)\; \leq \; P\left(\frac{gnL(\bGamma(\bgamma,\tilde{\bgamma}),\tilde{\bgamma})}{2(1+gn)}<\nu'\mu(\bGamma(\bgamma,\tilde{\bgamma}),\bgamma)\right)$$ 
and, by Lemma~\ref{lemma:s2} and using that $\nu'\in \Big(0,\tfrac{1}{12}\Big)$, we obtain that
\begin{align*}
    P\left(\frac{gnL(\bGamma(\bgamma,\tilde{\bgamma}),\tilde{\bgamma})}{2(1+gn)}<\nu'\mu(\bGamma(\bgamma,\tilde{\bgamma}),\bgamma)\right)&\;\;\leq\;\; \left(2\nu'\Big(1+\frac{1}{gn}\Big)\right)^{\frac{|\tilde{\bgamma}\setminus \bgamma|}{4}} e^{-\frac{1}{2}\Big(1-\sqrt{2\nu'\big(1+\frac{1}{gn}\big)}\Big)^2\mu(\bGamma(\bgamma,\tilde{\bgamma}),\bgamma)}.
\end{align*}
Since $\nu'<\tfrac{1}{12}$ and $\frac{1}{gn} \to 0$, for $n$ large enough $2\nu'\big(1+\frac{1}{gn}\big)<\frac{1}{6}$ and $\frac{1}{2}\Big(1-\sqrt{2\nu'\big(1+\frac{1}{gn}\big)}\Big)^2 >\frac{1}{2}\Big(1-\sqrt{1/6}\Big)^2>1/6$, we have
\begin{align*}
    P\left(\frac{gnL(\bGamma(\bgamma,\tilde{\bgamma}),\tilde{\bgamma})}{2(1+gn)}<\nu'\mu(\bGamma(\bgamma,\tilde{\bgamma}),\bgamma)\right)&\;\;\;\;\leq\;\;  \left(\frac{1}{6}\right)^{\frac{|\tilde{\bgamma}\setminus \bgamma|}{4}}e^{-\frac{1}{6}\mu(\bGamma(\bgamma,\tilde{\bgamma}),\bgamma)}.
\end{align*}
Since $\mu(\bGamma(\bgamma,\tilde{\bgamma}),\bgamma)>6\Delta(\bgamma,\tilde{\bgamma})$, we also have
$$
A_{\tilde{\bgamma}}\;=\;\nu'\mu(\bGamma(\bgamma,\tilde{\bgamma}),\bgamma)+\nu\Delta(\bgamma,\tilde{\bgamma})\;<\;\frac{1-\nu}{6}\mu(\bGamma(\bgamma,\tilde{\bgamma}),\bgamma)+\frac{\nu}{6}\mu(\bGamma(\bgamma,\tilde{\bgamma}),\bgamma)\;=\;\frac16\mu(\bGamma(\bgamma,\tilde{\bgamma}),\bgamma).
$$ It follows
\begin{equation}\label{eq:boundLQM}
     P\left(\frac{gnL(\bGamma(\bgamma,\tilde{\bgamma}),\tilde{\bgamma})}{2(1+gn)}<\nu'\mu(\bGamma(\bgamma,\tilde{\bgamma}),\bgamma)\right)\;\leq\; \left(\frac14\right)^{\frac{|\tilde{\bgamma}\setminus \bgamma|}{6}}e^{-A_{\tilde{\bgamma}}} \leq e^{-A_{\tilde{\bgamma}}}\leq e^{-\alpha A_{\tilde{\bgamma}}}
\end{equation}
Summing the bounds in \eqref{eq:boundLQS} and \eqref{eq:boundLQM} gives that for every $n$ large enough $E \left[ \pi(\bgamma\mid \by, \bomega) \right] < 2e^{-\alpha A_{\tilde{\bgamma}}}$. Since $\psi<\alpha$, for every $n$ large enough, we have that $E \left[ \pi(\bgamma\mid \by, \bomega) \right]  < e^{-\psi A_{\tilde{\bgamma}}}$ as we wished to prove.
\end{proof}

\subsubsection{Convergence to a reduced set of models}

\begin{lemma}\label{thm:convtoT}
    Assume that A1--A2 hold for $\omega_1,\ldots,\omega_B$, $|S^{I}(\kappa)|=O(1)$ and $|S^{S}_b(\kappa)|=O(p_b-s_b)$ (\emph{cf} \eqref{eq:defSes}) for every $b=1,\ldots,B$. Then,
    $$\lim_{n \to \infty} \sum_{\bgamma \not\in \T(\kappa)} E\left(\pi(\bgamma\mid \by, \bomega)\right)= 0$$
\end{lemma}

\begin{proof}
The first step of the proof is to use Lemma~\ref{lem:boundexpnc} to bound individual $E\left(\pi(\bgamma\mid \by, \bomega)\right)$ for every $\bgamma\not\in \T(\kappa)$ with Lemma~\ref{lem:boundexpnc} taking $\tilde{\bgamma}=\tilde{\bgamma}(\bgamma)\in \T(\kappa)$ that depends on $\bgamma$. 
Intuitively, $\tilde{\bgamma}(\bgamma)$ contains large truly non-zero parameters that are missed by $\bgamma$, and hence $\tilde{\bgamma}(\bgamma)$ should be chosen over $\bgamma$ asymptotically. More precisely,
we choose $\tilde{\bgamma}(\bgamma)\in \T(\kappa)$ such that 
$\tilde{\bgamma}(\bgamma)\setminus \bgamma\subseteq S^L(\kappa)$ and the elements in $\bgamma\setminus \tilde{\bgamma}(\bgamma)$ are either inactive or in $S^S(\kappa)$. The latter condition 
and Assumption A2 ensure that the assumptions of Lemma~\ref{lem:boundexpnc} are met. We then get a bound $E\left(\pi(\bgamma\mid \by, \bomega)\right) \leq e^{-\psi A_{\tilde{\bgamma}(\bgamma)}}$ for every large enough $n$ and any $\psi \in (0,1)$, where $A_{\tilde{\bgamma}(\bgamma)}=\nu\Delta(\bgamma,\tilde{\bgamma}(\bgamma))+\tfrac{1-\nu}{6}\mu(\Gamma(\tilde{\bgamma},\bgamma),\bgamma)$ (\emph{cf} \eqref{eq:muQM}), $\nu\in(1/2,1)$ is defined in \eqref{eq:defSes} and $\Gamma(\tilde{\bgamma},\bgamma)= \bgamma\cup \tilde{\bgamma}(\bgamma)$. The second step is to obtain a lower bound for $A_{\tilde{\bgamma}(\bgamma)}$, which gives an upper bound for $E(\pi(\bgamma\mid \by, \bomega))$.
The final step is to 
get an upper-bound $\sum_{\bgamma \in \bgamma\setminus \T(\kappa)} E\left(\pi(\bgamma\mid \by, \bomega)\right)$ that vanishes (as $n$ grows) under the assumptions of Lemma~\ref{thm:convtoT}.

For the first step of the proof, recall that the set $\T(\kappa)$ contains all the models that are the union of $S^L(\kappa)$ (large signals) and some subset of $S^{I}(\kappa)$ (intermediate signals). For any $\bgamma \not\in\T(\kappa)$, take the unique $\tilde{\bgamma}(\bgamma)\in\T(\kappa)$ such that $\tilde{\bgamma}(\bgamma) = \big[\bgamma\cap S^{I}(\kappa)\big]\cup S^L(\kappa) $, which implies $\bgamma\cap S^{I}(\kappa)=\tilde{\bgamma}(\bgamma)\cap S^{I}(\kappa)$. 
That is, $\tilde{\bgamma}(\bgamma)$ contains all the large signals plus the intermediate signals in $\bgamma$.
To show that $E(\pi(\bgamma\mid \by, \bomega)) \leq e^{-\psi A_{\tilde{\bgamma}(\bgamma)}}$ we show that $A_{\tilde{\bgamma}(\bgamma)}$ satisfies the conditions of Lemma~\ref{lem:boundexpnc}, taking $T=\tilde{\bgamma}(\bgamma)$. That is, we wish to show that three conditions hold: $A_{\tilde{\bgamma}(\bgamma)}>0$, $|\bgamma\setminus \tilde{\bgamma}(\bgamma)|=o(A_{\tilde{\bgamma}(\bgamma)})$, and $\mu(\Gamma(\tilde{\bgamma},\bgamma), \tilde{\bgamma}(\bgamma))=o(A_{\tilde{\bgamma}(\bgamma)})$.

Observe that $\Delta(\bgamma,\tilde{\bgamma}(\bgamma))$, defined in \eqref{eq:muQM}, can be rewritten as $\Delta(\bgamma,\tilde{\bgamma}(\bgamma))=\sum_{b=1}^B(|\bgamma_b\setminus \tilde{\bgamma}(\bgamma)_b|-|\tilde{\bgamma}(\bgamma)_b\setminus \bgamma_b|)\kappa_b$ and then:
$$ A_{\tilde{\bgamma}(\bgamma)}=\sum_{b=1}^B|\bgamma_b\setminus \tilde{\bgamma}(\bgamma)_b|\nu\kappa_b+\tfrac{1-\nu}{6}\mu(\Gamma(\tilde{\bgamma},\bgamma),\bgamma)-\sum_{b=1}^B|\tilde{\bgamma}(\bgamma)_b\setminus \bgamma_b|\nu\kappa_b $$
Since $\nu <1$, $(1-\nu)/6 >0$, by Lemma~\ref{lem:boundrhogen}, for every $n$ we have 
\begin{equation}\label{eq:nonoverfittingcond2}
   A_{\tilde{\bgamma}(\bgamma)}\;\geq\; \sum_{b=1}^B|\bgamma_b \setminus \tilde{\bgamma}(\bgamma)_b|\nu\kappa_b + \sum_{b=1}^B |\tilde{\bgamma}(\bgamma)_b\setminus \bgamma_b |\Big(\tfrac{1-\nu}{6} n \phi^{-1}\rho(\bX) \min_{i\in \tilde{\bgamma}(\bgamma)_b\setminus \bgamma_b}{\theta^*_i}^{2} - \nu\kappa_b\Big).
\end{equation} 
We have that $\tilde{\bgamma}(\bgamma)\subseteq \big(S^I(\kappa) \cup S^L(\kappa)\big)$ since $\tilde{\bgamma}(\bgamma)\in\T(\kappa)$ and that $\tilde{\bgamma}(\bgamma)\cap S^{I}(\kappa) = \bgamma\cap S^{I}(\kappa)$. It follows that $\tilde{\bgamma}(\bgamma)\setminus \bgamma \subseteq S^L(\kappa)$. By definition of $S^L(\kappa)$, the rightmost component in \eqref{eq:nonoverfittingcond2} is nonnegative and, if $|\tilde{\bgamma}(\bgamma)\setminus \bgamma|\neq0$, it is positive, for every $n$ large enough. By Assumption A2, component $\nu\sum_{b=1}^B|\bgamma_b \setminus \tilde{\bgamma}(\bgamma)_b|\kappa_b$ is nonnegative, and, if $|\bgamma\setminus \tilde{\bgamma}(\bgamma)|\neq0$, positive. Now, $\bgamma\neq \tilde{\bgamma}(\bgamma)$ implies that necessarily $|\bgamma\setminus \tilde{\bgamma}(\bgamma)|\neq0$ or $|\tilde{\bgamma}(\bgamma)\setminus \bgamma|\neq0$ and then, for every $n$ large enough, $A_{\tilde{\bgamma}(\bgamma)}>0$, establishing the first condition required by Lemma~\ref{lem:boundexpnc}. Regarding its second condition, if $|\bgamma\setminus \tilde{\bgamma}(\bgamma)|=0$, we immediately have $|\bgamma\setminus \tilde{\bgamma}(\bgamma)|=o(A_{\tilde{\bgamma}(\bgamma)})$. If $|\bgamma\setminus \tilde{\bgamma}(\bgamma)|\neq 0$, since the rightmost component in \eqref{eq:nonoverfittingcond2} is nonnegative, we have
$$\frac{|\bgamma\setminus \tilde{\bgamma}(\bgamma)|}{A_{\tilde{\bgamma}(\bgamma)}}\leq \bigg[ \nu\sum_{b=1}^B\frac{|\bgamma_b \setminus \tilde{\bgamma}(\bgamma)_b|}{|\bgamma\setminus \tilde{\bgamma}(\bgamma)|}\kappa_b \bigg]^{-1} \leq \bigg[ \nu\min_{b=1,\ldots,B}\kappa_b \bigg]^{-1} $$
where the last inequality follows from $\sum_{b=1}^B\frac{|\bgamma_b \setminus \tilde{\bgamma}(\bgamma)_b|}{|\bgamma\setminus \tilde{\bgamma}(\bgamma)|}=1$. By Assumption A2, $\min_{b=1,\ldots,B}\kappa_b\to \infty$ as $n\to\infty$ and hence $|\bgamma\setminus \tilde{\bgamma}(\bgamma)|=o(A_{\tilde{\bgamma}(\bgamma)})$ when $|\bgamma\setminus \tilde{\bgamma}(\bgamma)|\neq 0$ too.

Finally, consider the third condition in Lemma~\ref{lem:boundexpnc}.
Note that, since $\tilde{\bgamma}(\bgamma)\in\T(\kappa)$, we have that $S^L(\kappa)\subseteq \tilde{\bgamma}(\bgamma)$. Moreover, we have $\bgamma\cap S^{I}(\kappa)=\tilde{\bgamma}(\bgamma)\cap S^{I}(\kappa)$. It follows that $\bgamma\setminus \tilde{\bgamma}(\bgamma)\subseteq (S^{I}(\kappa) \cup S^L(\kappa))^C$, that is the elements of $\bgamma\setminus \tilde{\bgamma}(\bgamma)$ are either inactive or belong to $S^S(\kappa)$. If $\bgamma\cap S^{S}(\kappa)=\emptyset$ then $\bgamma\setminus \tilde{\bgamma}(\bgamma)\subseteq (\bgamma^*)^C$ and we immediately get $\mu(\Gamma(\tilde{\bgamma},\bgamma), \tilde{\bgamma}(\bgamma))=o(A_{\tilde{\bgamma}(\bgamma)})$ because $\btheta^*_{\bgamma \setminus \tilde{\bgamma}(\bgamma)}=\mu(\Gamma(\tilde{\bgamma},\bgamma), \tilde{\bgamma}(\bgamma))=0$. Assume now that $\bgamma\cap S^{S}(\kappa)\neq\emptyset$. Using that the rightmost component in \eqref{eq:nonoverfittingcond2} is nonnegative and Lemma~\ref{lem:upperboundnoncentralgen}, we have
$$\frac{\mu(\Gamma(\tilde{\bgamma},\bgamma), \tilde{\bgamma}(\bgamma))}{ A_{\tilde{\bgamma}(\bgamma)}} \leq \frac{\mu(\Gamma(\tilde{\bgamma},\bgamma), \tilde{\bgamma}(\bgamma))}{ \nu\sum_{b=1}^B|\bgamma_b \setminus \tilde{\bgamma}(\bgamma)_b|\kappa_b} \leq \frac{n\,\phi^{-1}\,\bar{\lambda}\,\sum_{b=1}^B|(\bgamma^*_b\cap \bgamma_b )\setminus \tilde{\bgamma}(\bgamma)_b|\,\max_{i\in (\bgamma^*_b\cap \bgamma_b )\setminus \tilde{\bgamma}(\bgamma)_b}{\theta_{i}^*}^2}{ \nu\sum_{b=1}^B|\bgamma_b \setminus \tilde{\bgamma}(\bgamma)_b|\kappa_b}.$$
Observe that for all $b=1,\ldots,B$, $\bgamma_b \setminus \tilde{\bgamma}(\bgamma)_b \supseteq (\bgamma^*_b\cap \bgamma_b )\setminus \tilde{\bgamma}(\bgamma)_b$ and then $|\bgamma_b \setminus \tilde{\bgamma}(\bgamma)_b| \geq |(\bgamma^*_b\cap \bgamma_b )\setminus \tilde{\bgamma}(\bgamma)_b|$. We then get
$$\frac{\mu(\Gamma(\tilde{\bgamma},\bgamma), \tilde{\bgamma}(\bgamma))}{ A_{\tilde{\bgamma}(\bgamma)}} \leq \frac{n\,\phi^{-1}\bar{\lambda}\,\sum_{b=1}^B|(\bgamma^*_b\cap \bgamma_b )\setminus \tilde{\bgamma}(\bgamma)_b|\,\max_{i\in(\bgamma^*_b\cap \bgamma_b )\setminus \tilde{\bgamma}(\bgamma)_b}{\theta_{i}^*}^2}{ \nu\sum_{b=1}^B|(\bgamma^*_b\cap \bgamma_b ) \setminus \tilde{\bgamma}(\bgamma)_b|\kappa_b}.$$
Moreover, since  $\bgamma\setminus \tilde{\bgamma}(\bgamma)\subseteq (S^{I}(\kappa) \cup S^L(\kappa))^C$ as discussed earlier, we have, for all $b=1,\ldots,B$, $(\bgamma^*_b\cap \bgamma_b )\setminus \tilde{\bgamma}(\bgamma)_b \subseteq S^{S}_b(\kappa)$. It follows $\max_{i\in(\bgamma^*_b\cap \bgamma_b )\setminus \tilde{\bgamma}(\bgamma)_b}{\theta_{i}^*}^2 \leq \max_{i\in S_b^S}{\theta_{i}^*}^2$ and
\begin{equation}\label{eq:onbfsa}
  \frac{\mu(\Gamma(\tilde{\bgamma},\bgamma), \tilde{\bgamma}(\bgamma))}{ A_{\tilde{\bgamma}(\bgamma)}} \leq \frac{n\,\phi^{-1}\bar{\lambda}\,\sum_{b=1}^B|(\bgamma^*_b\cap \bgamma_b )\setminus \tilde{\bgamma}(\bgamma)_b|\,\max_{i\in S^{S}_b(\kappa)}{\theta_{i}^*}^2}{ \nu\sum_{b=1}^B|(\bgamma^*_b\cap \bgamma_b ) \setminus \tilde{\bgamma}(\bgamma)_b|\kappa_b}.  
\end{equation}
Let $\bar{r}:=\sum_{b=1}^B n\phi^{-1}\bar{\lambda}\max_{i\in S^{S}_b(\kappa)}{\theta_i^*}^2/(\nu\kappa_b)$. We show next that $\bar{r}$ is an upper bound on $\mu(\Gamma(\tilde{\bgamma},\bgamma), \tilde{\bgamma}(\bgamma))/ A_{\tilde{\bgamma}(\bgamma)}$. By restricting the sum in $\bar{r}$ to the $b$ such that $|(\bgamma^*_b\cap \bgamma_b )\setminus \tilde{\bgamma}(\bgamma)_b|\neq 0$ and multiplying the numerator and denominator of the summand by $|(\bgamma^*_b\cap \bgamma_b )\setminus \tilde{\bgamma}(\bgamma)_b|$, we get 
\begin{eqnarray}\label{eq:ovrwS}
    \bar{r}
    &\geq &\sum_{b=1,|(\bgamma^*_b\cap \bgamma_b )\setminus \tilde{\bgamma}(\bgamma)_b|\neq 0}^B\frac{|(\bgamma^*_b\cap \bgamma_b )\setminus \tilde{\bgamma}(\bgamma)_b|n\phi^{-1}\bar{\lambda}\max_{i\in S^{S}_b(\kappa)}{\theta_i^*}^2}{|(\bgamma^*_b\cap \bgamma_b )\setminus \tilde{\bgamma}(\bgamma)_b|\nu\kappa_b}
\end{eqnarray}
Note that for any collections $(\alpha_j,\delta_j)\in  \R\times \R\setminus\{0\}$, $b=1,\ldots,B$, we have
\begin{equation}\label{eq:truc}
    \sum_{b=1}^B\frac{\alpha_j}{\delta_j}=
\sum_{b=1}^B\frac{\alpha_j \frac{1}{\delta_j} (\delta_j + \sum_{l \neq b} \delta_l)}{\sum_{b=1}^B \delta_j}=
    \frac{\sum_{b=1}^B\alpha_j(1+\sum_{l\neq b} \frac{\delta_l}{\delta_j})}{\sum_{b=1}^B\delta_j}
\end{equation}
Using the property in \eqref{eq:truc} in the right-hand side of \eqref{eq:ovrwS}, we get
\begin{align*}
     \bar{r}
    &\geq \frac{\sum_{b=1,|(\bgamma^*_b\cap \bgamma_b )\setminus \tilde{\bgamma}(\bgamma)_b|\neq 0 }^B|(\bgamma^*_b\cap \bgamma_b )\setminus \tilde{\bgamma}(\bgamma)_b|n\phi^{-1}\bar{\lambda}\max_{i\in S^{S}_b(\kappa)}{\theta_i^*}^2\Big(1+\sum_{l\neq b}\frac{|(\bgamma^*_l\cap \bgamma_l)\setminus \tilde{\bgamma}(\bgamma)_l|\nu\kappa_l}{|(\bgamma^*_b\cap \bgamma_b )\setminus \tilde{\bgamma}(\bgamma)_b|\nu\kappa_b}\Big)}{\sum_{b=1,|(\bgamma^*_b\cap \bgamma_b )\setminus \tilde{\bgamma}(\bgamma)_b|\neq 0}^B|(\bgamma^*_b\cap \bgamma_b )\setminus \tilde{\bgamma}(\bgamma)_b|\nu\kappa_b} \\
    &\geq \frac{\sum_{b=1}^B|(\bgamma^*_b\cap \bgamma_b )\setminus \tilde{\bgamma}(\bgamma)_b|n\phi^{-1}\bar{\lambda}\max_{i\in S^{S}_b(\kappa)}{\theta_i^*}^2}{\sum_{b=1}^B \big|(S\cap \bgamma_b)\setminus \tilde{\bgamma}(\bgamma)_b\big|\nu\kappa_b}
\end{align*}
where last inequality follows from $\Big(1+\sum_{l\neq b}\frac{|(\bgamma^*_l\cap \bgamma_l)\setminus \tilde{\bgamma}(\bgamma)_l|\nu\kappa_l}{|(\bgamma^*_b\cap \bgamma_b )\setminus \tilde{\bgamma}(\bgamma)_b|\nu\kappa_b}\Big)\geq 1$ for all $b$ and from the identity $\sum_{b=1,|(\bgamma^*_b\cap \bgamma_b )\setminus \tilde{\bgamma}(\bgamma)_b|\neq 0}^B|(\bgamma^*_b\cap \bgamma_b )\setminus \tilde{\bgamma}(\bgamma)_b|\nu\kappa_b=\sum_{b=1}^B \big|(\bgamma^*\cap \bgamma_b)\setminus \tilde{\bgamma}(\bgamma)_b\big|\nu\kappa_b$. Then, 
by \eqref{eq:onbfsa},
$$\frac{\mu(\Gamma(\tilde{\bgamma},\bgamma), \tilde{\bgamma}(\bgamma))}{ A_{\tilde{\bgamma}(\bgamma)}} \;\leq\; \bar{r} \;= \;\sum_{b=1}^B\frac{n\phi^{-1}\bar{\lambda}\max_{i\in S^{S}_b(\kappa)}{\theta_i^*}^2}{\nu\kappa_b}.$$
By definition of $S_b^{S}(\kappa)$, $n\phi^{-1}\bar{\lambda}\max_{i\in S^{S}_b(\kappa)}{\theta_i^*}^2=o(\kappa_b)$ for all $b$ and $\mu(\Gamma(\tilde{\bgamma},\bgamma), \tilde{\bgamma}(\bgamma))=o(A_{\tilde{\bgamma}(\bgamma)})$ when $\bgamma\cap S^{S}(\kappa)\neq\emptyset$ too.
We can now apply Lemma~\ref{lem:boundexpnc} and get that for every $\bgamma\in \mathcal \bgamma\setminus \T(\kappa)$, for any $\psi\in(0,1)$ and every $n$ large enough, $E(\pi(\bgamma\mid \by, \bomega))\;\leq\;e^{-\psi A_{\tilde{\bgamma}(\bgamma)}}$.

The second step of the proof is to lower-bound $A_{\tilde{\bgamma}(\bgamma)}=\nu\Delta(\bgamma,\tilde{\bgamma}(\bgamma))+\tfrac{1-\nu}{6}\mu(\Gamma(\tilde{\bgamma},\bgamma),\bgamma)$. 
Let
$$A_{\tilde{\bgamma}(\bgamma)}^*\;:=\;\nu\sum_{b=1}^B|\bgamma_b \setminus \tilde{\bgamma}(\bgamma)_b|\kappa_b + \sum_{b=1}^B |\tilde{\bgamma}(\bgamma)_b\setminus \bgamma_b |\left(\tfrac{1-\nu}{6} n \phi^{-1} \rho(\bX) \min_{j\in S^{L}_b(\kappa)}{\theta^*_j}^{2} - \nu\kappa_b\right) $$
Recall that, for every $b$, $\tilde{\bgamma}(\bgamma)_b\setminus \bgamma_b\subseteq S^{L}_b(\kappa)$. We have then $\min_{i\in \tilde{\bgamma}(\bgamma)_b\setminus \bgamma_b}{\theta^*_i}^{2}\geq \min_{j\in S^{L}_b(\kappa)}{\theta^*_j}^{2}$, and by \eqref{eq:nonoverfittingcond2}, $A_{\tilde{\bgamma}(\bgamma)} \geq A_{\tilde{\bgamma}(\bgamma)}^*$. It follows that for any $\psi\in(0,1)$ and for all $n$ large enough,
\begin{equation}\label{eq:boundexp2}
    E(\pi(\bgamma\mid \by, \bomega))\;\leq\;e^{-\psi A_{\tilde{\bgamma}(\bgamma)}^*}.
\end{equation}

To conclude the second part of the proof we lower-bound $\psi A^*_{\tilde{\bgamma}(\bgamma)}=\sum_{b=1}^B|\bgamma_b \setminus \tilde{\bgamma}(\bgamma)_b|\psi\nu\kappa_b + \sum_{b=1}^B |\tilde{\bgamma}(\bgamma)_b \setminus \bgamma_b|\psi\left(\tfrac{1-\nu}{6} n \phi^{-1}\rho(\bX) \min_{j\in S^{L}_b(\kappa)}{\theta^*_j}^{2} - \nu\kappa_b\right)$. To do this, we obtain a lower bound for $\psi \nu \kappa_b$ and for $\psi\bigg(\tfrac{1-\nu}{6} n \phi^{-1}\rho(\bX) \min_{j\in S^{L}_b(\kappa)}{\theta^*_j}^{2} - \nu\kappa_b\bigg)$.

The definition of $S^{L}_b(\kappa)$ implies that there exists some $g'_b \to \infty$ such that
\begin{equation}\label{eq:bngrws}
  \frac{(1-\nu) n \phi^{-1}\rho(\bX)}{6}\min_{j\in S^{L}_b(\kappa)}{\theta^*_j}^{2}\,- \,\kappa_b \;=\; \log(s_b) + g'_{b}.  
\end{equation}
Let $\delta \in (0,1)$ and denote $\bar{\bgamma}_b=\max\big\{\frac{2\log(p_b-s_b)}{f_b},\frac{2\log(s_b)}{g'_{b}}\big\}$, where $f_b$ is given in Assumption A2. 
Take $\psi=\max_{b=1,\ldots,B}\frac{\xi + \delta + \bar{\bgamma}_b}{1+\bar{\bgamma}_b}$ for some $\xi \in(0, 1- \delta)$ then $\psi \in (0,1)$, and we have, for every $b=1,\ldots,B$,
\begin{align}
    &\psi > \frac{\delta + \frac{2\log(p_b-s_b)}{f_b}}{1+\frac{2\log(p_b-s_b)}{f_b}} = \frac{\delta f_b/2 + \log(p_b-s_b)}{f_b/2 +\log(p_b-s_b)} \label{eq:psif2}\\
    &\psi > \frac{\delta + \frac{2\log(s_b)}{g'_b}}{1+\frac{2\log(s_b)}{g'_b}} = \frac{\delta g'_b/2 + \log(s_b)}{g'_b/2 +\log(s_b)} \geq \frac{\delta g'_b/2 + \log(s_b)}{g'_b +\log(s_b)}.\label{eq:psig2}
\end{align}
 Recall that Assumptions A2 defines $f_b= \kappa_b - \log(p_b-s_b)$ and $\nu= \frac{1}{2}(1 + \max_b \frac{\log(p_b-s_b)}{\kappa_b})$ respectively. Hence,
$$
\nu\kappa_b\;\geq\;\frac12 \Big(1+\frac{\log(p_b-s_b)}{\kappa_b}\Big)\kappa_b\;=\;\log(p_b-s_b)+\frac12(\kappa_b-\log(p_b-s_b)) \;=\; \log(p_b-s_b)+ \frac12f_b.
$$
Hence, by \eqref{eq:psif2}, we have
\begin{equation}\label{eq:ooofjfj3}
\psi\nu\kappa_b \;\geq\; \psi\big(\log(p_b-s_b)+ \frac12f_b\big) \;\geq\;  \log(p_b-s_b)+ \delta\frac12f_b.
\end{equation}
Further,
\begin{equation}\label{eq:ooofjfj4}
\psi\bigg(\tfrac{1-\nu}{6} n \phi^{-1}\rho(\bX) \min_{j\in S^{L}_b(\kappa)}{\theta^*_j}^{2} - \nu\kappa_b\bigg) \;\geq\; \psi\left(\log(s_b)+g_b'\right) \;\geq\; \log(s_b)+\delta \frac{1}{2}g'_b,
\end{equation}
where the first inequality follows from \eqref{eq:bngrws} and the second inequality from \eqref{eq:psig2}. In \eqref{eq:boundexp2}, $\psi A^*_{\tilde{\bgamma}(\bgamma)}=\sum_{b=1}^B|\bgamma_b \setminus \tilde{\bgamma}(\bgamma)_b|\psi\nu\kappa_b + \sum_{b=1}^B |\tilde{\bgamma}(\bgamma)_b \setminus \bgamma_b|\psi\left(\tfrac{1-\nu}{6} n \phi^{-1}\rho(\bX) \min_{j\in S^{L}_b(\kappa)}{\theta^*_j}^{2} - \nu\kappa_b\right)$. Then by \eqref{eq:ooofjfj3} and \eqref{eq:ooofjfj4}, we get
\begin{equation}\label{eq:fyv2}
E(\pi(\bgamma\mid \by, \bomega))\;\leq\;\exp\left\{- \sum_{b=1}^B |\bgamma_b\setminus \tilde{\bgamma}(\bgamma)_b|(\log(p_b-s_b)+\delta \tfrac{f_b}{2})-\sum_{b=1}^B |\tilde{\bgamma}(\bgamma)_b\setminus \bgamma_b|(\log(s_b)+\delta \tfrac{g'_b}{2})\right\}.
\end{equation}

For the final step of the proof, denote $\mathcal{S}= \sum_{\bgamma \not\in \T(\kappa)} E\left(\pi(\bgamma\mid \by, \bomega)\right)$ for convenience. By \eqref{eq:fyv2} we have
\begin{align*}
    \mathcal{S} &\leq \sum_{\bgamma \not\in \T(\kappa)}e^{- \sum_{b=1}^B |\bgamma_b\setminus \tilde{\bgamma}(\bgamma)_b|\big(\log(p_b-s_b)+\delta \tfrac{f_b}{2}\big)-\sum_{b=1}^B |\tilde{\bgamma}(\bgamma)_b\setminus \bgamma_b|\big(\log(s_b)+\delta \tfrac{g'_b}{2}\big)}.
\end{align*}
We split the sum in the right-hand side above into sums over models $\bgamma$ 
such that $\tilde{\bgamma}(\bgamma)=\tilde{\bgamma}$ for some $\tilde{\bgamma}\in \T(\kappa)$. Denote for any $\tilde{\bgamma}\in \T(\kappa)$, $\bGamma(\tilde{\bgamma}):=\{\bgamma \not\in \T(\kappa) \,| \tilde{\bgamma}(\bgamma)=\tilde{\bgamma}\}$, then
\begin{align}\label{eq:gepihw}
    \mathcal{S} &\leq \sum_{\tilde{\bgamma} \in \T(\kappa)}\sum_{\bgamma \in \bGamma(\tilde{\bgamma})}e^{- \sum_{b=1}^B |\bgamma_b\setminus T_b|\big(\log(p_b-s_b)+\delta \tfrac{f_b}{2}\big)-\sum_{b=1}^B |T_b\setminus \bgamma_b|\big(\log(s_b)+\delta \tfrac{g'_b}{2}\big)}.
\end{align}
The right hand-side of \eqref{eq:gepihw} is composed of a double sum over $\tilde{\bgamma} \in \T(\kappa)$ and over $\bgamma \in \bGamma(\tilde{\bgamma})$. Consider the sum over $\bgamma \in \bGamma(\tilde{\bgamma})$, add $\tilde{\bgamma}$ to it, and denote it
\begin{equation}\label{eq:defST}
  \mathcal{S}(\tilde{\bgamma})=\sum_{\bgamma \in \bGamma(\tilde{\bgamma}) \cup \tilde{\bgamma}}e^{- \sum_{b=1}^B |\bgamma_b\setminus T_b|\big(\log(p_b-s_b)+\delta \tfrac{f_b}{2}\big)-\sum_{b=1}^B |T_b\setminus \bgamma_b|\big(\log(s_b)+\delta \tfrac{g'_b}{2}\big)}.
\end{equation}
In the summand in the right-hand side of \eqref{eq:defST}, the case $\bgamma=\tilde{\bgamma}$ correspond to $|\bgamma_b\setminus \tilde{\bgamma}(\bgamma)_b|=|\tilde{\bgamma}(\bgamma)_b\setminus \bgamma_b|=0$ for all $b$ and the summand is then 1.
By \eqref{eq:gepihw}, we then get that
\begin{align}\label{eq:gepihw2}
    \mathcal{S} &\leq \sum_{\tilde{\bgamma} \in \T(\kappa)}\big(\mathcal{S}(\tilde{\bgamma})-1\big).
\end{align}
For each $\tilde{\bgamma}=\tilde{\bgamma}(\bgamma)$, we further split $\mathcal{S}(\tilde{\bgamma})$ into sums over subsets of models $\bgamma$ that have $u_b$ more parameters than $\tilde{\bgamma}(\bgamma)$ in block $b$, and are missing $w_b$ parameters from $\tilde{\bgamma}(\bgamma)$ in block $b$.
Specifically, consider models $\bgamma$ such that, for all $b$, $|\bgamma_b\setminus \tilde{\bgamma}(\bgamma)_b|=u_b$ and $|\tilde{\bgamma}(\bgamma)_b\setminus \bgamma_b|=w_b$ with $u_b \in \{0,\ldots,p_b-s_b,\ldots,p_b-s_b+|S_b^{S}(\kappa)|\}$ and $w_b \in \{0,\ldots,|S^{L}_b(\kappa)|\}$. Denote by
$$S^{\mathbf{u}}_{\mathbf{w}}(\tilde{\bgamma})=\sum_{\bgamma \in \bGamma(\tilde{\bgamma}) \cup \tilde{\bgamma}:\forall \,b |\bgamma_b\setminus \bgamma^*_b|=u_b, |\bgamma^*_b\setminus \bgamma_b|=w_b } e^{- \sum_{b=1}^B u_b\big(\log(p_b-s_b)+\delta \tfrac{f_b}{2}\big)-\sum_{b=1}^B w_b\big(\log(s_b)+\delta \tfrac{g'_b}{2}\big)}.$$
We get
\begin{equation}\label{eq:uygfew2}
    \mathcal{S}(\tilde{\bgamma})=\sum_{w_1=0}^{|S^L_1(\kappa)|}\cdots\sum_{w_b=0}^{|S^L_b(\kappa)|}\sum_{u_1=0}^{p_1-s_1+|S_1^S|}\cdots\sum_{u_b=0}^{p_b-s_b+|S_b^S|}S^{\mathbf{u}}_{\mathbf{w}}(\tilde{\bgamma}).
\end{equation}
The number of models missing, for all $b$, $w_b$ out of the $|S^L_b(\kappa)|$ large active parameters and having $u_b$ inactive or small active parameters from $B_b$ is $\prod_{b=1}^B {\binom{p_b-s_b+|S^S_b(\kappa)|}{u_b}}{\binom{|S^L_b(\kappa)|}{w_b}}$. We thus have
\begin{align*}
    S^{\mathbf{u}}_{\mathbf{w}}(\tilde{\bgamma})&= \bigg(\prod_{b=1}^B {\binom{p_b-s_b+|S^S_b(\kappa)|}{u_b}}{\binom{|S^L_b(\kappa)|}{w_b}}\bigg)e^{- \sum_{b=1}^B u_b\big(\log(p_b-s_b)+\delta \tfrac{f_b}{2}\big)-\sum_{b=1}^B w_b\big(\log(s_b)+\delta \tfrac{g'_b}{2}\big)} \\
    &= \prod_{b=1}^B {\binom{p_b-s_b+|S^S_b(\kappa)|}{u_b}}e^{- u_b\big(\log(p_b-s_b)+\delta \tfrac{f_b}{2}\big)}{\binom{|S^L_b(\kappa)|}{w_b}}e^{-w_b\big(\log(s_b)+\delta \tfrac{g'_b}{2}\big)}.
\end{align*}
Inputting the expression above in \eqref{eq:uygfew2} and factorizing over terms in $u_b$ and $w_b$ gives
\begin{align*}
    \mathcal{S}(\tilde{\bgamma})\leq\prod_{b=1}^B & \Bigg(\sum_{u_b=0}^{p_b-s_b+|S^S_b(\kappa)|}{\binom{p_b-s_b+|S^S_b(\kappa)|}{u_b}}e^{- u_b(\log(p_b-s_b)+\delta \tfrac{f_b}{2})}\Bigg)\\&.\Bigg(
    \sum_{w_b=0}^{|S^L_b(\kappa)|}{\binom{|S^L_b(\kappa)|}{w_b}}
    e^{-w_b(\log(s_b)+\delta \tfrac{g'_b}{2})}\Bigg).
\end{align*}
A standard bound on binomial coefficient for $1\leq k\leq n$ is
\begin{equation}\label{eq:upperboudbinom}
    \binom{n}{k} \leq \left(\frac{n\,e}{k}\right)^k \leq \left(n\,e\right)^k = e^{k(\log(n)+1)}.
\end{equation}
By the bound in \eqref{eq:upperboudbinom} and taking the terms in $u_b=0$ and $w_b=0$ out of the sums above, 
we have
\begin{align}\label{eq:ogr}
    \mathcal{S}(\tilde{\bgamma})\leq\prod_{b=1}^B & \left(1+\sum_{u_b=1}^{p_b-s_b+|S_b^{S}(\kappa)|}e^{-u_b\left(\delta \tfrac{f_b}{2}-\log\big(1+\frac{|S^{S}_b(\kappa)|}{p_b-s_b}\big)-1\right)}\right)\\&.\left(1+
    \sum_{w_b=1}^{|S_b^{L}(\kappa)|}    e^{-w_b\left(\delta \tfrac{g'_b}{2}+\log\big(\frac{s_b}{|S_b^{L}(\kappa)|}\big)-1\right)}\right).\nonumber
\end{align}
Denote
\begin{align}
d_b\;=\;e^{1+\log\big(1+\frac{|S^{S}_b(\kappa)|}{p_b-s_b}\big)-\delta \tfrac{f_b}{2}},\qquad h_b\;=\;e^{1-\log\big(\frac{s_b}{|S_b^{L}(\kappa)|}\big)-\delta \tfrac{g'_b}{2}}.
\label{eq:db}
\end{align}
By assumption $|S^{S}_b(\kappa)|=O(p_b-s_b)$, and by definition of $|S_b^{L}(\kappa)|$ we have $s_b\geq|S_b^{L}(\kappa)|$. By Assumption A2 and the definition of $|S_b^{L}(\kappa)$, we also have $f_b\to\infty$ and $g'_b\to\infty$, and then $\lim_{n \to \infty}= d_b= \lim_{n \to \infty} h_b=0$. Using the properties of geometric series, 
 for every $j$ we have
\begin{eqnarray*}
    &1+\sum_{u_b=1}^{p_b-s_b+|S_b^{S}(\kappa)|}e^{-u_b\left(\delta \tfrac{f_b}{2}-\log\big(1+\frac{|S^{S}_b(\kappa)|}{p_b-s_b}\big)-1\right)}
    =\frac{1-d_b^{p_b-s_b+|S_b^{S}(\kappa)|+1}}{1-d_b} \\
    &1+\sum_{w_b=1}^{|S_b^{L}(\kappa)|}    e^{-w_b\left(\delta \tfrac{g'_b}{2}+\log\big(\frac{s_b}{|S_b^{L}(\kappa)|}\big)-1\right)} =\frac{1-h_b^{|S_b^{L}(\kappa)|+1}}{1-h_b},
\end{eqnarray*}
where both expressions converge to $1$ as $n$ grows. By \eqref{eq:gepihw2} and \eqref{eq:ogr}:
$$\mathcal{S}\leq\sum_{\tilde{\bgamma}\in\T(\kappa)}\Bigg(\prod_{b=1}^B\Big(\frac{1-d_b^{p_b-s_b+|S_b^{S}(\kappa)|+1}}{1-d_b}\Big)\Big(\frac{1-h_b^{|S_b^{L}(\kappa)|+1}}{1-h_b}\Big)-1\Bigg).$$
Each of the summand vanishes as $n\to\infty$. Moreover, by assumption $|S^I(\kappa)|=O(1)$ and then $|\T(\kappa)|=2^{|S^I(\kappa)|}=O(1)$. We thus have $\lim_{n\to\infty}\mathcal{S}=\lim_{n\to\infty}\sum_{\bgamma \not\in \T(\kappa)} E\left(\pi(\bgamma\mid \by, \bomega)\right)=0$.\\
\end{proof}

\subsubsection{Convergence of the expected sum of posterior inclusion probabilities}
\begin{lemma}
    \label{prop:shatconsistence}
Assume that A1--A2 hold for $\omega_1,\ldots,\omega_B$, $|S^{I}(\kappa)|=O(1)$ and $|S^{S}_b(\kappa)|=O(p_b-s_b)$ for every $j=1,\ldots,b$. Then
$$ \,\, \frac{|S^{L}_b(\kappa)|}{p_b}\leq \lim_{n\to\infty}E\bigg(\frac{1}{p_b}\sum_{z_j=b} \pi(\bgamma_j \mid \by, \bomega)\bigg)\leq \frac{s_b-|S^{S}_b(\kappa)|}{p_b}\qquad\text{for all }j=1,\ldots,b.$$
\end{lemma}

\begin{proof}
Denote
$$A_b:=\frac{1}{p_b}\sum_{z_j=b}\sum_{\bgamma \in \T(\kappa)|j\in  \bgamma} \pi(\bgamma \mid \by, \bomega) \quad\text{and}\quad  C_b:=\frac{1}{p_b}\sum_{z_j=b}\sum_{\bgamma \not\in \T(\kappa)|j\in  \bgamma} \pi(\bgamma \mid \by, \bomega).$$
For every $j=1,\ldots,b$, we have the decomposition
\begin{gather}\label{eq:decompshat}
    \frac{1}{p_b}\sum_{z_j=b} \pi(\bgamma_j \mid \by, \bomega)
    =\frac{\sum_{z_j=b} \sum_{\bgamma  | j \in \bgamma}  \pi(\bgamma \mid \by, \bomega)}{p_b}
   =A_b\,+\, C_b.
\end{gather}
To show the lower bound, we decompose $A_b$
\begin{eqnarray}
  A_b
   &=& \sum_{\bgamma \in \T(\kappa)} \pi(\bgamma \mid \by, \bomega)\sum_{z_j=b}\frac{I(j\in  \bgamma_b)}{p_b}\nonumber\\
   &=& \sum_{\bgamma \in \T(\kappa)} \pi(\bgamma \mid \by, \bomega)\sum_{j| z_j=b, j \in S^{L}_b(\kappa)}\frac{I(j\in  \bgamma_b)}{p_b}+\sum_{\bgamma \in \T(\kappa)} \pi(\bgamma \mid \by, \bomega)\sum_{j| z_j=b, j\not\in S^{L}_b(\kappa)}\frac{I(j\in  \bgamma_b)}{p_b}\nonumber\\
   &=& \frac{|S^{L}_b(\kappa)|}{p_b}\sum_{\bgamma \in \T(\kappa)} \pi(\bgamma \mid \by, \bomega)+\sum_{\bgamma \in \T(\kappa)} \pi(\bgamma \mid \by, \bomega)\sum_{j| z_j=b, j\not\in S^{L}_b(\kappa)}\frac{I(j\in  \bgamma_b)}{p_b}\label{eq:wigr},
\end{eqnarray}
where the last equality follows from $I(j\in  \bgamma_b)=1$ for all $j\in S^L(\kappa)$ when $\bgamma\in\T(\kappa)$. The rightmost term above and $C_b$ are nonnegative, then by the linearity of the expectation
$$E\bigg(\frac{1}{p_b}\sum_{z_j=b} \pi(\bgamma_j \mid \by, \bomega)\bigg)\geq\frac{|S^{L}_b(\kappa)|}{p_b}\sum_{\bgamma \in \T(\kappa)}E( \pi(\bgamma \mid \by, \bomega)).$$
By Lemma~\ref{thm:convtoT}, $\lim_{n\to\infty}\sum_{\bgamma \in \T(\kappa)}E( \pi(\bgamma \mid \by, \bomega))=1$. It follows that $$\lim_{n\to\infty}E\bigg(\frac{1}{p_b}\sum_{z_j=b} \pi(\bgamma_j \mid \by, \bomega)\bigg)\geq \frac{|S^{L}_b(\kappa)|}{p_b}\quad\text{ for every   }j=1,\ldots,b.$$

We now prove the upper bound. Recall that $\T(\kappa)$ by definition includes models that have no small signals, i.e. all parameters are in $S^L(\kappa) \cup S^I(\kappa)$. That is,
for all $\bgamma\in\T(\kappa)$, we have that $I(j\in  \bgamma_b)=0$ for all $j\in\{j|z_j=b, j\not\in(S^{L}_b(\kappa)\cup S^{I}_b(\kappa))\}$. Hence, $A_b$ in \eqref{eq:wigr} satisfies
\begin{align*}
    A_b
   &= \frac{|S^{L}_b(\kappa)|}{p_b}\sum_{\bgamma \in \T(\kappa)} \pi(\bgamma \mid \by, \bomega)+\sum_{\bgamma \in \T(\kappa)} \pi(\bgamma \mid \by, \bomega)\sum_{j \in S^{I}_b(\kappa)}\frac{I(j\in  \bgamma_b)}{p_b}\\
   &\leq\frac{|S^{L}_b(\kappa)|}{p_b}\sum_{\bgamma \in \T(\kappa)} \pi(\bgamma \mid \by, \bomega)+ \frac{|S^{I}(\kappa)_b|}{p_b}\sum_{\bgamma \in \T(\kappa)} \pi(\bgamma \mid \by, \bomega)
\end{align*}
where the inequality follows from $\sum_{j \in S^{I}_b(\kappa)} I(j\in  \bgamma_b) \leq |S_b^{I}(\kappa)|$  for all $\bgamma$. By \eqref{eq:decompshat}, we then have
\begin{equation}\label{eq:uuuuu}
    \frac{1}{p_b}\sum_{z_j=b} \pi(\bgamma_j \mid \by, \bomega)\leq \frac{|S^{L}_b(\kappa)|+|S^{I}(\kappa)_b|}{p_b}\sum_{\bgamma \in \T(\kappa)} \pi(\bgamma \mid \by, \bomega) + C_b.
\end{equation}
Moreover, for every $j=1,\ldots,b$, $C_b$ satisfies
\begin{equation}\label{eq:uuuuu2}
    C_b = \sum_{\bgamma \not\in \T(\kappa)} \pi(\bgamma \mid \by, \bomega)\sum_{z_j=b}\frac{I(j\in  \bgamma_b)}{p_b}\leq \sum_{\bgamma \not\in \T(\kappa)} \pi(\bgamma \mid \by, \bomega)
\end{equation}
where the inequality follows from $\sum_{z_j=b}\frac{I(j\in  \bgamma_b)}{p_b}\leq 1$ for all $\bgamma$.
Taking expectations in \eqref{eq:uuuuu} and \eqref{eq:uuuuu2} gives
\begin{equation*}
    E\bigg(\frac{1}{p_b}\sum_{z_j=b} \pi(\bgamma_j \mid \by, \bomega)\bigg) 
    \leq \frac{|S^{L}_b(\kappa)|+|S^{I}(\kappa)_b|}{p_b}\sum_{\bgamma \in \T(\kappa)}E( \pi(\bgamma \mid \by, \bomega)) + \sum_{\bgamma \not\in \T(\kappa)}E( \pi(\bgamma \mid \by, \bomega)).
\end{equation*}
 By Lemma~\ref{thm:convtoT}, we have on one hand $\lim_{n\to \infty}\sum_{\bgamma \in \T(\kappa)}E( \pi(\bgamma \mid \by, \bomega))= 1$ and, on the other hand, $\lim_{n\to \infty}\sum_{\bgamma \not\in\T(\kappa)}E( \pi(\bgamma \mid \by, \bomega))= 0$. It follows that $$ \lim_{n\to \infty}E\bigg(\frac{1}{p_b}\sum_{z_j=b} \pi(\bgamma_j \mid \by, \bomega)\bigg) \leq  \frac{|S^{L}_b(\kappa)|+|S^{I}(\kappa)_b|}{p_b}= \frac{s_b-|S^{S}_b(\kappa)|}{p_b}\quad \text{for every }j=1,\ldots,b,$$ which proves the upper bound.
\end{proof}

\subsection{Proof of Theorem \ref{thm:suffcondlinearmodel}}\label{proof:suffcondlinearmodel}

Theorem \ref{thm:suffcondlinearmodel} is a particular case of Lemma \ref{thm:convtoT}. Observe that under assumption A3, $S^L(\kappa)=\bgamma^*$ and $S^S(\kappa)=S^I(\kappa)=\emptyset$. It follows that $\T(\kappa)=\bgamma^*$. Assumptions A1 and A2 are also assumed to hold and, trivially, $|S^{I}(\kappa)|=O(1)$ and $|S^{S}_b(\kappa)|=O(p_b-s_b)$ for every $b=1,\ldots,B$, then by Lemma \ref{thm:convtoT}, $\lim_{n \to \infty} \sum_{\bgamma \neq \bgamma^*} E\left(\pi(\bgamma\mid \by, \bomega)\right)= 0$.

\subsection{Oracle convergence rate for Theorem \ref{thm:suffcondlinearmodel}}
\label{ssec:oracle_conv_rate}

We give in Corollary \ref{cor:convtoT} oracle values for $\kappa_b^*$ that approximately optimize the rates at which
$\sum_{\bgamma \neq \bgamma^*} E\left(\pi(\bgamma\mid \by, \bomega)\right)$ converges to 0 under a betamin assumption.
The result follows from Lemma \ref{thm:convtoT} and some extra derivations.
We refer to $\kappa_b^*$ as oracle penalties because they depend on unknown quantities such as the sparsity $s_b$ and the smallest true signal $\theta^*_{\min,b}$ in each block.

\begin{corollary} \label{cor:convtoT}
    Assume that A1 holds, and let
\begin{align}
 \sqrt{\kappa^*_b\,} \;=\; \frac12\sqrt{\frac{(1-\nu^*) n\rho(\bX)}{6\phi}}{\theta^*_{\min,b}}+\frac12\sqrt{\frac{6\phi}{(1-\nu^*) n\rho(\bX)}}\frac{\ln({p_b-s_b})-\ln({s_b})}{\theta^*_{\min,b}}
\nonumber
\end{align}
for some fixed $\nu^* \in (1/2,1)$.
Further assume that, for all $b=1,\ldots,B$, the following betamin condition holds
\begin{equation}\label{betamincond:minimalreg}
    \lim_{n\to \infty}\frac{g(\nu^*)\sqrt{ n\rho(\bX)\phi^{-1}/6}\;\theta^*_{\min,b}}{\sqrt{\ln(p_b-s_b)} + \sqrt{\ln(s_b)}} \geq 1
\end{equation}
where $g(\nu^*)=\sqrt{(-1 + 2 \nu^*)(1-\nu^*)}/(1 + \sqrt{2 - 2 \nu^*})$.
Then there exists a constant $c>0$ such that
\begin{equation}
\sum_{\bgamma \neq \bgamma^*} E\left(\pi(\bgamma\mid \by, \bomega)\right)\;\leq\; 2c\,\sum_{b=1}^B\,\, e^{-\tfrac{\delta}{2}\big[\kappa^*_b-\ln(p_b-s_b)\big]},
\nonumber
\end{equation}
where $\delta\in(0,1)$ is a constant that can be taken arbitrarily close to 1.
\end{corollary}

\begin{proof}
The proof strategy is as follows. The first step is to obtain an upper bound for $\mathcal{S} := \sum_{\bgamma \neq \bgamma^*} E\left(\pi(\bgamma\mid \by, \bomega)\right)$ when $\kappa_b$ and $\theta_{\min,b}^*$ meet Assumptions A2 ad A3, building on the proof of Theorem \ref{thm:convtoT}. The second step in the proof is to set $\kappa_b=\kappa_b^*$, where the latter meets A2 and A3 and approximately minimize the upper bound.

Assume the $\kappa_b$ and $\theta_{\min,b}^*$, $b=1,\ldots,B$, meet Assumptions A2 ad A3. Since A3 hold, we have $S^L(\kappa)=\bgamma^*$ and $S^S(\kappa)=S^I(\kappa)=\emptyset$ such that $\T(\kappa)=\bgamma^*$. Trivially, $|S^{I}(\kappa)|=O(1)$ and $|S^{S}_b(\kappa)|=O(p_b-s_b)$ for every $b=1,\ldots,B$ and the assumptions of Lemma \ref{thm:convtoT} are also met. All steps in the proof of Theorem \ref{thm:convtoT} therefore remain valid for the $\kappa_b$.

Recall that, from \eqref{eq:gepihw2}, we have that
\begin{align}
    \mathcal{S} \leq \sum_{\tilde{\bgamma} \in \T(\kappa)}\big(\mathcal{S}(\tilde{\bgamma})-1\big) = \mathcal{S}(\bgamma^*)-1.
\nonumber
\end{align}
where $\mathcal{S}(\bgamma^*)$ is defined in \eqref{eq:defST}. Using that $S^L(\kappa)=\bgamma^*$ and $S^S(\kappa)=\emptyset$, the bound in \eqref{eq:ogr} gives 
 \begin{align}\label{eq:ogr2}
    \mathcal{S}(\bgamma^*)\leq\prod_{b=1}^B & \left(1+\sum_{u_b=1}^{p_b-s_b}e^{-u_b\left(\delta \tfrac{f_b}{2}-1\right)}\right)\\&.\left(1+
    \sum_{w_b=1}^{s_b}    e^{-w_b\left(\delta \tfrac{g'_b}{2}-1\right)}\right).\nonumber
\end{align}
where $\delta\in(0,1)$ is arbitrarily close to 1, $f_b$ is defined in Assumption A2, $g_b'$ as in \eqref{eq:bngrws} for the case that A3 holds, that is
\begin{equation}\label{eq:bngrws2}
  \frac{(1-\nu) n \phi^{-1}\rho(\bX)}{6}{\theta^*_{\min,b}}^{2}\,- \,\kappa_b \;=\; \log(s_b) + g'_{b}.  
\end{equation}
Denote
\begin{align*}
d_b\;=\;e^{1-\delta \tfrac{f_b}{2}},\qquad h_b\;=\;e^{1-\delta \tfrac{g'_b}{2}}.
\label{eq:db2}
\end{align*}
 Developing the product in the right-hand side in \eqref{eq:ogr2} and reordering the resulting terms gives
$$\mathcal{S}(\bgamma^*)-1\;\;\leq\;\;-1+1+\sum_{b=1}^B\big[S(u_b)+S(w_b)\big]+\mathcal{R}$$
where
\begin{align*}
    &S(u_b) = \sum_{u_b=1}^{p_b-s_b}e^{-u_b\left(\delta \tfrac{f_b}{2}-1\right)}
    =d_b\,\frac{1-d_b^{p_b-s_b}}{1-d_b} \\
    &S(w_b) = \sum_{w_b=1}^{s_b}    e^{-w_b\left(\delta \tfrac{g'_b}{2}-1\right)}=h_b\,\frac{1-h_b^{s_b}}{1-h_b},
\end{align*}
and all the terms in $\mathcal{R}$ are product of two or more of the sums $S(u_1),\ldots,S(u_B),S(w_1),\ldots,S(w_B)$. Given that $\delta>0$, $f_b \to \infty$ and $g_b' \to \infty$ by assumption, and hence $d_b \to 0$ and $h_b \to 0$, the $S(u_b)$ and $S(w_b)$ are smaller than 1 for all sufficiently large $n$ for all $b$. Then, each of the $2^{2B}-2B-1$ terms in $\mathcal{R}$ is bounded above by $\sum_{b=1}^B\big[S(u_b)+S(w_b)\big]$, and we get, for every $n$ large enough,
 \begin{equation} \label{eq:fierg}
 \mathcal{S}(\bgamma^*)-1
 \leq (2^{2B}-2B)\sum_{b=1}^B \bigg[d_b\,\frac{1-d_b^{p_b-s_b}}{1-d_b}+ h_b\,\frac{1-h_b^{s_b}}{1-h_b}\bigg ]. \end{equation} 
Let $r>1$ constant. Since $d_b \to 0$ and $h_b \to 0$, then $\frac{1-d_b^{p_b-s_b}}{1-d_b}\to1$ and $\frac{1-h_b^{s_b}}{1-h_b}\to1$ for every $b$, and for every $n$ large enough, 
\begin{equation} \label{eq:fierg2}r>\max_{b=1,\ldots,b}\bigg\{\frac{1-d_b^{p_b-s_b}}{1-d_b} \,,\,\frac{1-h_b^{s_b}}{1-h_b} \bigg\}.\end{equation} 
By \eqref{eq:fierg} and \eqref{eq:fierg2}, we then obtain
$$\mathcal{S}(\bgamma^*)-1
     \;\leq\; (2^{2B}-2B) r\, e \,\sum_{b=1}^B e^{-\delta \tfrac{f_b}{2}} +e^{-\delta \tfrac{g'_b}{2}}.$$
By \eqref{eq:gepihw2} and using $\T(\kappa)=\bgamma^*$, we get
$$\sum_{\bgamma \neq \bgamma^*} E\left(\pi(\bgamma\mid \by, \bomega)\right)\;\leq\;(2^{2B}-2B) r\, e \,\,\sum_{b=1}^B e^{-\delta \tfrac{f_b}{2}} +e^{-\delta \tfrac{g'_b}{2}}.$$
In the rest of the proof, we denote the constant factor on the righthand side above $c=(2^{2B}-2B) r\, e$. By the definition of $f_b$ in Assumption A2, that of $g'_b$ in \eqref{eq:bngrws2}, we get, for every $n$ large enough, that $\mathcal{S}=$
\begin{equation}\label{eq:gpkr}
    \sum_{\bgamma \neq \bgamma^*} E\left(\pi(\bgamma\mid \by, \bomega)\right)\;\leq\; c\,\sum_{b=1}^B\,\, e^{-\tfrac{\delta}{2}\big[\kappa_b-\ln(p_b-s_b)\big]} + e^{-\tfrac{\delta}{2}\big[\frac{(1-\nu) n\rho(\bX)}{6\phi}{\theta^*_{\min,b}}^{2}\,- \,\kappa_b -\ln(s_b) \big]}.
\end{equation}

We have $\frac{(1-\nu) n \rho(\bX)}{6\phi}{\theta^*_{\min,b}}^{2}\,- \,\kappa_b > \Big(\sqrt{\frac{(1-\nu) n\rho(\bX)}{6\phi}}\theta^*_{\min,b} - \sqrt{\kappa_b}\Big)^2$ and then, for any $\delta\in(0,1)$ arbitrarily close to 1, we obtain the convergence rate
\begin{equation}\label{eq:gHIPS}
  \mathcal{S} \leq c\sum_{b=1}^B e^{-\tfrac{\delta}{2}\big[\kappa_b-\ln(p_b-s_b)\big]} + e^{-\tfrac{\delta}{2}\big[\big(\sqrt{\frac{(1-\nu) n\rho(\bX)}{6\phi}}\theta^*_{\min,b} - \sqrt{\kappa_b}\big)^2-\kappa_b -\ln(s_b)\big]}.
\end{equation}
Expression \eqref{eq:gHIPS} provides an upper-bound for $\mathcal{S}= \sum_{\bgamma \neq \bgamma^*} E\left(\pi(\bgamma\mid \by, \bomega)\right)$, concluding the first part of the proof.

In the second part of the proof we derive approximately optimal values for the block penalties $\kappa_b$. In \eqref{eq:gHIPS}, each summand features a first term that decays exponentially in the penalty $\kappa_b$, and a second term that decays exponentially in 
$(\sqrt{\tfrac{(1-\nu)n\rho(\bX)}{6\phi}}\,\theta^*_{\min,b} - \sqrt{\kappa_b})^2
$. These two terms capture competing type I and II error contributions.
The choice $\kappa_b=\kappa_b^*$ such that
\begin{equation} 
\sqrt{\kappa^*_b\,} \;=\; \frac12\sqrt{\frac{(1-\nu^*) n\rho(\bX)}{6\phi}}{\theta^*_{\min,b}}+\frac12\sqrt{\frac{6\phi}{(1-\nu^*) n\rho(\bX)}}\frac{\ln({p_b-s_b})-\ln({s_b})}{\theta^*_{\min,b}},
\nonumber
\end{equation}
for some fixed $\nu^*\in(1/2,1)$, balances these contributions by equating their exponential rates, and thus approximately minimizes \eqref{eq:gHIPS}. 

To be able to apply \eqref{eq:gHIPS} to $\kappa^*_b$, we next show that if Assumption \eqref{betamincond:minimalreg} holds then the $\kappa^*_b$ and $\theta^*_{\min,b}$ satisfy Assumptions A2 and A3. Indeed, under~\eqref{betamincond:minimalreg}, for every $b$, there exists a sequence $x_b$ such that $\lim_{n\to\infty} x_b\geq1$ and $g(\nu^*)\sqrt{\frac{n}{6\phi}}\theta^*_{\min,b}=x_b\big(\sqrt{\ln(p_b-s_b)} + \sqrt{\ln(s_b)}\big)$. Let $v_b= x_b \sqrt{(1-\nu^*)}g(\nu^*)^{-1}$, then $\sqrt{\frac{(1-\nu^*)n}{6\phi}}\theta^*_{\min,b}=v_b\big(\sqrt{\ln(p_b-s_b)} + \sqrt{\ln(s_b)}\big)$. Denote $a=\sqrt{\ln(p_b-s_b)}$ and $l=\sqrt{\ln(s_b)}$.
We have
$$
\sqrt{\kappa^*_b} \;=\;\frac{v_b}{2} (a+l)+\frac{(a^2-l^2)}{2v_b(a+l)}\;=\;\frac{v_b}{2} (a+l)+\frac{a-l}{2v_b}\;=\;\frac{v_b^2+1}{2v_b}a+\frac{v_b^2-1}{2v_b}l.
$$
Since $\sqrt{(1-\nu^*)}g(\nu^*)^{-1}>1$ for $\nu^*\in(1/2,1)$ and $\lim_{n\to\infty} x_b\geq1$, for $n$ large enough, $v_b>1$. We also have $b\geq 0$. It follows that $b(v_b^2-1)/2v_b \geq 0$ for $n$ large enough and
\begin{equation}\label{eq:ojdv}
\sqrt{\kappa^*_b} =\frac{v_b^2+1}{2v_b}a+\frac{v_b^2-1}{2v_b}l  \geq  \Big(\frac{v_b^2+1}{2v_b}+1-1\Big)a = \Big(1+\frac{(v_b-1)^2}{2v_b}\Big)\sqrt{\ln(p_b-s_b)} 
\end{equation}
which implies Assumption A2 since $v_b>1$ for $n$ sufficiently large. We also have
\begin{equation}\label{eq:oabpx}
\frac{\sqrt{(1-\nu^*) n\rho(\bX)}}{6}{\theta_{\min,b}^*}-\sqrt{\kappa^*_b} = \frac{v_b^2-1}{2v_b}a+\frac{v_b^2+1}{2v_b}l \geq \Big(1+\frac{(v_b-1)^2}{2v_b}\Big)\sqrt{\ln(s_b)}.
\end{equation}
Let $\nu = \tfrac12(1+\max_b\ln(p_b-s_b)/\kappa^*_b)$ as defined in Assumption A3. By~\eqref{eq:ojdv}, $\nu = \tfrac12(1+\max_b\big(1+(v_b-1)^2/(2v_b))^{-2}$. Since $\lim_{n\to\infty} x_b\geq 1$, then for sufficiently large $n$, $v_b\geq \sqrt{(1-\nu^*)}g(\nu^*)^{-1}$. Simple algebra shows the latter inequality implies $\nu = \tfrac12(1+\max_b\big(1+(v_b-1)^2/(2v_b))^{-2} < \nu^*$ and $1-\nu > 1-\nu^*$. Then, by~\eqref{eq:oabpx}, 
\begin{equation*}
\frac{\sqrt{(1-\nu) n\rho(\bX)}}{6\phi}{\theta_{\min,b}^*}-\sqrt{\kappa^*_b} \geq \Big(1+\frac{(v_b-1)^2}{2v_b}\Big)\sqrt{\ln(s_b)},
\end{equation*}
for every sufficiently large $n$. Since $\lim_{n\to\infty} x_b\geq1$ and $\sqrt{(1-\nu^*)}g(\nu^*)^{-1}>1$ for $\nu^*\in(1/2,1)$, $v_b>1$ for $n$ sufficiently large and Assumption A3 holds.

Since under Assumption \eqref{betamincond:minimalreg}, Assumptions A2 and A3 hold for the $\kappa^*_b$, \eqref{eq:gHIPS} applies. Plugging such $\kappa_b^*$ into \eqref{eq:gHIPS}, and noting that
$\kappa_b^*$ satisfies $\kappa^*_b-\ln(p_b-s_b)=\big(\sqrt{\frac{(1-\nu^*) n\rho(\bX)}{6\phi}}{\theta^*_{\min,b}} - \sqrt{\kappa_b^*}\big)^2 -\ln(s_b)$ for all $b$, gives that 
for every $n$ large enough,
\begin{equation}\label{eq:afi}
\mathcal{S}= \sum_{\bgamma \neq \bgamma^*} E\left(\pi(\bgamma\mid \by, \bomega)\right)\;\leq\; 2c\,\sum_{b=1}^B\,\, e^{-\tfrac{\delta}{2}\big[\kappa^*_b-\ln(p_b-s_b)\big]},
\end{equation}
where $\delta\in(0,1)$ is arbitrarily close to 1, as we wished to prove.
\end{proof}

\subsection{Proof of Theorem \ref{theo:empbayesselconsist}}\label{proof:empbayesselconsist}

The proof strategy is to show that Assumptions A1, A2 and A3 hold to apply Theorem~\ref{thm:suffcondlinearmodel}. Recall that Theorem~\ref{theo:empbayesselconsist} makes Assumptions A1, A4, A5 and A6, hence it suffices to show that A2 and A3 hold.
We first derive a convenient decomposition of the empirical Bayes penalties $\kappa^{(1)}_b$. The second step of the proof is to show that, with probability going to 1, these $\kappa^{(1)}_b$ satisfy Assumption A2. The third step consists in showing that Assumption A5 implies Assumption A3 for $\kappa^{(1)}_b$. The convergence of $E\left(\pi(\bgamma^*\mid \by, \bomega^{(1)})\right)$ then follows Theorem~\ref{thm:suffcondlinearmodel}.

In Step 1, from the choice $\omega^{(0)}$ and the ensuing penalty $\kappa^{(0)}$, we have
\begin{align*}
     \frac{1}{p_b}\sum_{z_j=b}\pi(\bgamma_j\mid \by, \bomega^{(0)})
    &\;=\; \frac{1}{p_b}\sum_{z_j = b}\,\,\sum_{\bgamma}\pi(\bgamma\mid \by, \bomega^{(0)})I(j\in \bgamma)\\
   &\;=\; \frac{1}{p_b}\sum_{\bgamma }\pi(\bgamma\mid \by, \bomega^{(0)})\sum_{z_j = b}I(j\in \bgamma) \\
   &\;=\;\frac{s_b}{p_b}\pi(\bgamma^*\mid \by, \bomega^{(0)}) + \sum_{\bgamma  | \bgamma \neq \bgamma^*}\frac{|\bgamma_b|}{p_b}\pi(\bgamma\mid \by, \bomega^{(0)}).
\end{align*}
Using that $\pi(\bgamma^*\mid \by, \bomega^{(0)})=1-\sum_{\bgamma  | \bgamma \neq \bgamma^*}\pi(\bgamma\mid \by, \bomega^{(0)})$, we get
\begin{equation*}
   \frac{1}{p_b}\sum_{z_j=b}\pi(\bgamma_j\mid \by, \bomega^{(0)})
   =\frac{s_b}{p_b} + \sum_{\bgamma  | \bgamma \neq \bgamma^*}\frac{|\bgamma_b|-s_b}{p_b}\pi(\bgamma\mid \by, \bomega^{(0)}).
\end{equation*}
Consider the decomposition of the sum in the right-hand side above between the sum over models $\bgamma$ that contain more parameters than $\bgamma^*$ in block $b$ and the sum over those that contain fewer parameters than $\bgamma^*$ in block $b$. Denote
$$
    O_b^{(0)} := \sum_{\bgamma  | |\bgamma_b| > s_b}\frac{|\bgamma_b|-s_b}{p_b}\pi(\bgamma\mid \by, \bomega^{(0)})\quad\text{and}\quad
    U_b^{(0)} := \sum_{\bgamma  | |\bgamma_b| < s_b}\frac{s_b-|\bgamma_b|}{p_b}\pi(\bgamma\mid \by, \bomega^{(0)}).
$$
We have
\begin{equation}\label{eq:rewriteshat}
  \frac{1}{p_b}\sum_{z_j=b}\pi(\bgamma_j\mid \by, \bomega^{(0)})  = \frac{s_b}{p_b}\,\,+\,\, O_b^{(0)}-\,\,U_b^{(0)}.
\end{equation}
Observe that we have the following decomposition of Step 2 penalties
$$\kappa^{(1)}_b
    = \log(p_b-s_b) + \log\Big(\frac{\sqrt{1+gn}}{s_b}\Big) + \log\Bigg(\frac{p_b-\sum_{z_j=b}\pi(\bgamma_j\mid \by, \bomega^{(0)}) }{p_b-s_b}\Bigg) + \log\Bigg(\frac{s_b}{\sum_{z_j=b}\pi(\bgamma_j\mid \by, \bomega^{(0)}) }\Bigg). \nonumber$$
By \eqref{eq:rewriteshat}, it follows that
\begin{equation}\label{eq:decompstep2pen}
    \kappa^{(1)}_b = \log(p_b-s_b)+\log\Big(\frac{\sqrt{1+gn}}{s_b}\Big)+\log\Bigg(1-\frac{p_b(O^{(0)}_b-U^{(0)}_b)}{p_b-s_b}\Bigg)+\log\Bigg(\frac{s_b}{\sum_{z_j=b}\pi(\bgamma_j\mid \by, \bomega^{(0)}) }\Bigg),
\end{equation}
completing the first step of the proof.

We continue with the second step of the proof: showing that the $\kappa^{(1)}_b$'s satisfy Assumption~A2 with probability going to 1. Recall that Assumption~A2 states that there exists $f_b\to \infty$ (as $n\to \infty$) such that for every sufficiently large $n$,
\begin{equation*}
    \kappa_b \;=\;\log(p_b-s_b) + f_b.
\end{equation*}
Since $U_b^{(0)}$ is nonnegative, a lower bound on $\kappa^{(1)}_b$ is 
\begin{equation}\label{eq:decompstep2pen2}
    \kappa^{(1)}_b
    \geq  \log(p_b-s_b)+\log\Big(\frac{\sqrt{1+gn}}{s_b}\Big)+\log\Bigg(1-\frac{p_b O^{(0)}_b}{p_b-s_b}\Bigg)+\log\Bigg(\frac{s_b}{\sum_{z_j=b}\pi(\bgamma_j\mid \by, \bomega^{(0)}) }\Bigg).
\end{equation}
Plugging in the definition of $O^{(0)}_b$, we have that
$$\frac{p_b O^{(0)}_b}{p_b-s_b} =  \sum_{\bgamma  | |\bgamma_b| > s_b}\frac{|\bgamma_b|-s_b}{p_b-s_b}\pi(\bgamma\mid \by, \bomega^{(0)})\leq \sum_{\bgamma  | |\bgamma_b| > s_b}\pi(\bgamma\mid \by, \bomega^{(0)})$$
where the inequality follows from $(|\bgamma_b|-s_b)/(p_b-s_b)\leq 1$ for all $\bgamma$.
Note that if $\bgamma$ is such that $|\bgamma_b|>s_b$, then $\bgamma\not\in\T(\kappa^{(0)})$ (this follows immediately from the definition of $\T(\kappa)$ in \eqref{eq:Tkappa})  and therefore $\sum_{\bgamma  | |\bgamma_b| > s_b}\pi(\bgamma\mid \by, \bomega^{(0)})\leq \sum_{\bgamma  \not \in \T(\kappa^{(0)})}\pi(\bgamma\mid \by, \bomega^{(0)})$. Moreover, $\kappa^{(0)}$ satisfies Assumption~A2 and the assumptions of Lemma~\ref{thm:convtoT} are met for $\kappa^{(0)}$.
Then, by Lemma~\ref{thm:convtoT}, $\lim_{n\to\infty} \sum_{\bgamma \not\in \T(\kappa^{(0)})}E(\pi(\bgamma\mid \by, \bomega^{(0)}))=\lim_{n\to\infty}\sum_{\bgamma  | |\bgamma_b| > s_b}E(\pi(\bgamma\mid \by, \bomega^{(0)})) = 0$. It follows that $\frac{p_b O^{(0)}_b}{p_b-s_b}$ vanishes in probability and so does $\log\Big(1-\frac{p_b O^{(0)}_b}{p_b-s_b}\Big)$. By Assumption~A4, we also have that $\log(\sqrt{1+gn}s_b^{-1})\to\infty$. 
Then, to show that Assumption~A2 holds with probability going to 1, it is enough to show that with probability going to 1, $\log(s_b/\sum_{z_j=b}\pi(\bgamma_j\mid \by, \bomega^{(0)}) )$ is nonnegative. Observe that all assumptions in Lemma~\ref{prop:shatconsistence} are also met for $\kappa^{(0)}$. By \eqref{eq:uuuuu} and \eqref{eq:uuuuu2} in the proof of Lemma~\ref{prop:shatconsistence}, we have
\begin{align*}
    \frac{\sum_{z_j=b}\pi(\bgamma_j\mid \by, \bomega^{(0)}) }{p_b}
    \leq & \frac{|S^{L}_b(\kappa^{(0)})|+|S^{I}(\kappa^{(0)})_b|}{p_b}\sum_{\bgamma \in \T(\kappa^{(0)})}\pi(\bgamma\mid \by, \bomega^{(0)}) \\
    &+ \frac{1}{p_b} \sum_{\bgamma \not\in \T(\kappa^{(0)})}\pi(\bgamma\mid \by, \bomega^{(0)}).
\end{align*}
By Lemma~\ref{thm:convtoT}, $\sum_{\bgamma \in \T(\kappa^{(0)})}\pi(\bgamma\mid \by, \bomega^{(0)})$ and $\sum_{\bgamma \not\in \T(\kappa^{(0)})}\pi(\bgamma\mid \by, \bomega^{(0)})$ converge in probability to 1 and 0 respectively. We then have that, with probability going to 1,
\begin{equation*}
    \frac{\sum_{z_j=b}\pi(\bgamma_j\mid \by, \bomega^{(0)}) }{p_b}\leq \frac{|S^{L}_b(\kappa^{(0)})|+|S^{I}(\kappa^{(0)})_b|}{p_b},
\end{equation*}
which implies, with probability going to 1,
$$\log\Big(\frac{s_b}{\sum_{z_j=b}\pi(\bgamma_j\mid \by, \bomega^{(0)}) }\Big)\geq \log\bigg(\frac{s_b}{|S^{L}_b(\kappa^{(0)})|+|S^{I}(\kappa^{(0)})_b|} \bigg) \geq 0.$$
We then obtain that, with probability going to 1,
\begin{equation}\label{eq:eeeeeeeh}
    \kappa^{(1)}_b
    \geq  \log(p_b-s_b)+\log\Big(\frac{\sqrt{1+gn}}{s_b}\Big)
\end{equation}
and that the $\kappa^{(1)}_b$ satisfy Assumption~A2, completing the second part of the proof.\medskip

For the third and final part of the proof, 
we now show that Assumption~A5 implies Assumption~A3 for the $\kappa_b^{EB}$ with probability going to 1. Recall that Assumption~A3 for the $\kappa_b^{EB}$ states that for each block $b$ there exists $g_b\to \infty$ such that for large enough $n$,
    \begin{equation*}
        \sqrt{\frac{(1-\nu) n \phi^{-1}\rho(\bX)}{6}}{{\theta_{\min,b}^*}} \,- \,\sqrt{\kappa^{(1)}_b} \;=\; \sqrt{\log(s_b)} + g_b,
    \end{equation*}
where $\nu$ takes value
\begin{equation}\label{eq:gammakeb}
    \nu=\frac12\Big(1+\max_b\frac{\log(p_b-s_b)}{\kappa^{(1)}_b}\Big).
\end{equation}
Observe that Assumption~A5 and Assumption~A3 take the same form. To show that Assumption~A5 implies Assumption~A3 for the $\kappa_b^{EB}$ with probability going to 1, it suffices to show that the following two inequalities
\begin{align}
    \sqrt{\frac{ (1-\nu)n\phi^{-1}\rho(\bX)}{6}}{\theta_{\min,b}^*} \,&\geq \,  \sqrt{\frac{ (1-\psi)n\phi^{-1}\rho(\bX)}{6}}{\theta_{\min,b}^*}, \label{eq:oaebr} \\
    - \,\sqrt{\kappa^{(1)}_b} \,&\geq \,-\sqrt{\log\bigg(\frac{p}{|S^{L}(\kappa^{(0)})|}-1\bigg)+\frac12\log(1+gn)}\label{eq:oaebr2}
\end{align}
hold with probability going to 1 for $\nu$ as in \eqref{eq:gammakeb} and $\psi=\frac12\big(1+\max_b \frac{\log(p_b-s_b)}{\log(p_b/s_b-1)+0.5\log(1+gn)}\big)$ (defined in Assumption~A5). We first show \eqref{eq:oaebr} holds with probability going to 1 and then that \eqref{eq:oaebr2} does too.

By \eqref{eq:eeeeeeeh}, with probability going to 1,
$$ \frac{\log(p_b-s_b)}{\kappa^{(1)}_b}
    \leq \frac{\log(p_b-s_b)}{\log(p_b-s_b)+\log\Big(\frac{\sqrt{1+gn}}{s_b}\Big)} =  \frac{\log(p_b-s_b)}{\log(p_b/s_b-1)+0.5\log(1+gn)} .$$
It follows that:
\begin{equation*}
    \nu=\frac12\Big(1+\max_{b=1, \ldots,b}\frac{\log(p_b-s_b)}{\kappa^{(1)}_b}\Big) \leq \frac12\Big(1+\max_b \frac{\log(p_b-s_b)}{\log(p_b/s_b-1)+0.5\log(1+gn)}\Big)=\psi
\end{equation*}
and \eqref{eq:oaebr} holds with probability going to 1.

We now upper bound $\kappa^{(1)}_b$ to show \eqref{eq:oaebr2} holds with probability going to 1. Observe that 
$$
\log\Big(\frac{\sum_{z_j=b}\pi(\bgamma_j\mid \by, \bomega^{(0)}) }{s_b}\Big)=
 \log \left( 1 + \frac{\sum_{z_j=b}\pi(\bgamma_j\mid \by, \bomega^{(0)})  - s_b}{s_b}\right) =   
\log\Big(1+\frac{p_b(O^{(0)}_b-U^{(0)}_b)}{s_b}\Big),
$$
where the second equality follows from \eqref{eq:rewriteshat}.
Plugging this expression into \eqref{eq:decompstep2pen}, and using that $O_b^{(0)} \geq 0$, we have that  
\begin{equation}\label{eq:gggggggg}
    \kappa^{(1)}_b \leq \log\big(p_b-s_b\big)+\log\Big(\frac{\sqrt{1+gn}}{s_b}\Big)+\log\bigg(\frac{1+\frac{p_b}{p_b-s_b}U^{(0)}_b}{1-\frac{p_b}{s_b}U^{(0)}_b}\bigg).
\end{equation}
We split the sum in $U^{(0)}_b$ between models in $\T(\kappa^{(0)})$ and those not in $\T(\kappa^{(0)})$. 
$$U_b^{(0)} = \sum_{\bgamma \in \T(\kappa^{(0)}) | |\bgamma_b| < s_b}\frac{s_b-|\bgamma_b|}{p_b}\pi(\bgamma\mid \by, \bomega^{(0)})+\sum_{\bgamma \not\in \T(\kappa^{(0)}) | |\bgamma_b| < s_b}\frac{s_b-|\bgamma_b|}{p_b}\pi(\bgamma\mid \by, \bomega^{(0)}).
$$
If $\bgamma\in \T(\kappa^{(0)})$, then by definition $|\bgamma_b|\geq |S^{L}_b(\kappa^{(0)})|$ and thus $s_b-|\bgamma_b|\leq s_b-|S^{L}_b(\kappa^{(0)})|$. A bound on $s_b-|\bgamma_b|$ for $\bgamma\not\in \T(\kappa^{(0)})$ is simply $s_b-|\bgamma_b|\leq s_b$. It follows that
\begin{eqnarray*}
    U^{(0)}_b \leq \frac{s_b-|S^{L}_b(\kappa^{(0)})|}{p_b}\sum_{\bgamma \in \T(\kappa^{(0)}) | |\bgamma_b| < s_b}\pi(\bgamma\mid \by, \bomega^{(0)})+  \frac{s_b}{p_b}\sum_{\bgamma  \not\in \T(\kappa^{(0)})| |\bgamma_b| < s_b}\pi(\bgamma\mid \by, \bomega^{(0)})
\end{eqnarray*}
By Lemma~\ref{thm:convtoT}, $\sum_{\bgamma \in  \T(\kappa^{(0)})| |\bgamma_b| < s_b}\pi(\bgamma\mid \by, \bomega^{(0)})$ and $\sum_{\bgamma  \not\in \T(\kappa^{(0)})| |\bgamma_b| < s_b}\pi(\bgamma\mid \by, \bomega^{(0)})$ converge in probability to 1 and 0 respectively. We then get that, with probability going to 1,
\begin{align}
    \frac{p_b}{p_b-s_b}U^{(0)}_b &\leq \frac{s_b-|S^{L}_b(\kappa^{(0)})|}{p_b-s_b}\nonumber\\
    \frac{p_b}{s_b}U^{(0)}_b &\leq  \frac{s_b-|S^{L}_b(\kappa^{(0)})|}{s_b}. \label{eq:gewou}
\end{align}
By the bounds above and \eqref{eq:gggggggg}, we have that with probability going to 1,
\begin{align}
  \kappa^{(1)}_b 
  &\leq \log\big(p_b-s_b\big)+\log\Big(\frac{\sqrt{1+gn}}{s_b}\Big)+\log\Bigg(\frac{1+\frac{s_b-|S^{L}_b(\kappa^{(0)})|}{p_b-s_b}}{1- \frac{s_b-|S^{L}_b(\kappa^{(0)})|}{s_b}}\Bigg) \nonumber\\
  &= \log\big(p_b-s_b\big) + \log\Big(\frac{\sqrt{1+gn}}{s_b}\Big) + \log\Bigg(\frac{\frac{p_b-|S^{L}_b(\kappa^{(0)})|}{p_b-s_b}}{ \frac{|S^{L}_b(\kappa^{(0)})|}{s_b}}\Bigg) \\
&=\log\big(p_b/|S^{L}_b(\kappa^{(0)})|-1\big)+\frac{1}{2}\log(1+gn)\label{eq:finalupperboundkappaj}
\end{align}
which shows \eqref{eq:oaebr2} holds with probability going to 1 and that
Assumption~A5 implies Assumption~A3 holds for the $\kappa^{(1)}_b$ with probability going to 1. 

Since Assumptions~A2 and A3 hold with probability going to 1, by Theorem~\ref{thm:suffcondlinearmodel}, $\lim_{n\to\infty}E \left[ \pi(\bgamma^* \mid \by, \hat{\bomega}^{(1)}) \right]= 1$, as we wished to prove.

\section{Simulation study}\label{sec:simstudy_suppl}

We report numerical results complementing and extending the figures shown in the main text,
for the simulation scenarios listed next.
The scenarios differ in how informative is the external data $\bZ$. 
Specifically, additionally to the intercept, $\bZ$ contains a first meta-covariate that ranges from strongly  informative to non-informative,
as measured by the size of the hyper-parameter $\omega_1$.
The second covariate is always non-informative ($\omega_2=0$).

\begin{itemize}
\item Table \ref{tab:simstudy1} shows results for Scenario 1, where $\omega_1=2$, $\omega_2=0$.

\item Table \ref{tab:simstudy2} shows results for Scenario 2, where $\omega_1=1$, $\omega_2=0$.

\item Table \ref{tab:simstudy3} shows results for Scenario 3, where $\omega_1=0$, $\omega_2=0$.

\item Table \ref{tab:simstudy4} shows results for Scenario 4, where $\omega_1=1.5$, $\omega_2=0$.

\item Table \ref{tab:simstudy5} shows results for Scenario 5, where $\omega_1=0.75$, $\omega_2=0$.

\item  Table \ref{tab:simstudy5_sensitivity} shows results for Scenario 5 again, for different values of Zellner's g-prior parameter (the default $g=1$, $g=0.1$ and $g=10$) 
\end{itemize}

The results from Scenarios 1--3 are discussed in the main paper.

The results from Scenario 4 (Table \ref{tab:simstudy4}) are qualitatively similar to those from Scenarios 1 and 2.
When using a pMOM prior on the parameters, the empirical Bayes model prior significantly improves MSE and power over the Beta-Binomial prior.
For Zellner's prior, a similar pattern is observed. An exception is the more challenging case where $n=100$ and $p=200$: the Beta-Binomial prior shows a much smaller FDR, at the cost of a much smaller power and higher mean squared estimation error (MSE).

The results from Scenario 5 (Table \ref{tab:simstudy5}) are similar to those from Scenario 2. The main difference is that the improvement in MSE and power of the empirical Bayes model prior over the Beta-Binomial are smaller in Scenario 5, as expected given that the external information is less informative $(\omega_1=0.75)$.
 Table \ref{tab:simstudy5_sensitivity} show that the results of our framework were very robust to the choice of the $g$ dispersion parameter in Zellner's prior.

\begin{table}[h]
\begin{tabular}{@{}lcc|ccc|@{}}
\hline
Method & n & p & MSE & Power & FDR \\ 
\hline
Empirical Bayes (pMOM) & 100 & 200 & 4.87 & 0.45 & 0.11 \\ 
Empirical Bayes (pMOM, Intercept) & 100 & 200 & 6.38 & 0.25 & 0.12 \\ 
  Beta-Binomial (pMOM) & 100 & 200 & 6.57 & 0.24 & 0.11 \\ 
  Empirical Bayes (Zellner) & 100 & 200 & 4.89 & 0.44 & 0.12 \\ 
  Empirical Bayes (Zellner, Intercept) & 100 & 200 & 6.43 & 0.24 & 0.11 \\ 
  Beta-Binomial (Zellner) & 100 & 200 & 11.00 & 0.03 & 0.00 \\ 
  LASSO & 100 & 200 & 2.99 & 0.90 & 0.60 \\ 
  ALASSO & 100 & 200 & 5.62 & 0.64 & 0.27 \\ 
  SCAD & 100 & 200 & 7.39 & 0.63 & 0.59 \\ 
\hline
  Empirical Bayes (pMOM) & 200 & 200 & 0.73 & 0.82 & 0.01 \\ 
  Empirical Bayes (pMOM, Intercept) & 200 & 200 & 1.06 & 0.76 & 0.02 \\ 
  Beta-Binomial (pMOM) & 200 & 200 & 1.06 & 0.76 & 0.02 \\ 
  Empirical Bayes (Zellner) & 200 & 200 & 0.72 & 0.83 & 0.01 \\ 
  Empirical Bayes (Zellner, Intercept) & 200 & 200 & 1.06 & 0.76 & 0.02 \\ 
  Beta-Binomial (Zellner) & 200 & 200 & 1.75 & 0.49 & 0.00 \\ 
  LASSO & 200 & 200 & 0.95 & 1.00 & 0.58 \\ 
  ALASSO & 200 & 200 & 0.80 & 0.98 & 0.28 \\ 
  SCAD & 200 & 200 & 2.50 & 0.93 & 0.48 \\ 
\hline
  Empirical Bayes (pMOM) & 200 & 100 & 0.26 & 0.82 & 0.00 \\ 
  Empirical Bayes (pMOM, Intercept) & 200 & 100 & 0.32 & 0.77 & 0.00 \\ 
  Beta-Binomial (pMOM) & 200 & 100 & 0.32 & 0.77 & 0.00 \\ 
  Empirical Bayes (Zellner) & 200 & 100 & 0.26 & 0.83 & 0.00 \\ 
  Empirical Bayes (Zellner, Intercept) & 200 & 100 & 0.32 & 0.77 & 0.00 \\ 
  Beta-Binomial (Zellner) & 200 & 100 & 0.34 & 0.72 & 0.00 \\ 
  LASSO & 200 & 100 & 0.39 & 1.00 & 0.60 \\ 
  ALASSO & 200 & 100 & 0.31 & 0.98 & 0.24 \\ 
  SCAD & 200 & 100 & 0.37 & 0.98 & 0.33 \\ 
\hline
\end{tabular}
\caption{Simulation study. Scenario 1 ($\omega_1=2,\omega_2=0$). 
Mean squared estimation error (MSE), power and false discovery rate (FDR).
Priors on parameters are product MOM (pMOM) and Zellner, both with default unit prior precision.
Priors on model are that given by our empirical Bayes framework, and by a default Beta-Binomial(1,1)
}\label{tab:simstudy1}
\end{table}

\begin{table}[h]
\begin{tabular}{@{}lcc|ccc|@{}}
\hline
Method & n & p & MSE & Power & FDR \\ 
\hline
Empirical Bayes (pMOM) & 100 & 200 & 3.53 & 0.40 & 0.19 \\ 
  Empirical Bayes (pMOM, Intercept) & 100 & 200 & 3.65 & 0.33 & 0.18 \\ 
  Beta-Binomial (pMOM) & 100 & 200 & 3.45 & 0.31 & 0.13 \\ 
  Empirical Bayes (Zellner) & 100 & 200 & 3.47 & 0.40 & 0.17 \\ 
  Empirical Bayes (Zellner, Intercept) & 100 & 200 & 3.92 & 0.33 & 0.19 \\ 
  Beta-Binomial (Zellner) & 100 & 200 & 3.48 & 0.15 & 0.01 \\ 
  LASSO & 100 & 200 & 1.58 & 0.92 & 0.69 \\ 
  ALASSO & 100 & 200 & 2.10 & 0.74 & 0.33 \\ 
  SCAD & 100 & 200 & 3.24 & 0.72 & 0.58 \\ 
\hline
  Empirical Bayes (pMOM) & 200 & 200 & 0.44 & 0.77 & 0.01 \\ 
  Empirical Bayes (pMOM, Intercept) & 200 & 200 & 0.49 & 0.73 & 0.01 \\ 
  Beta-Binomial (pMOM) & 200 & 200 & 0.49 & 0.73 & 0.01 \\ 
  Empirical Bayes (Zellner) & 200 & 200 & 0.44 & 0.77 & 0.01 \\ 
  Empirical Bayes (Zellner, Intercept) & 200 & 200 & 0.49 & 0.73 & 0.01 \\ 
  Beta-Binomial (Zellner) & 200 & 200 & 0.54 & 0.65 & 0.00 \\ 
  LASSO & 200 & 200 & 0.58 & 1.00 & 0.67 \\ 
  ALASSO & 200 & 200 & 0.45 & 0.97 & 0.29 \\ 
  SCAD & 200 & 200 & 1.37 & 0.93 & 0.33 \\ 
\hline
  Empirical Bayes (pMOM) & 200 & 100 & 0.16 & 0.77 & 0.00 \\ 
  Empirical Bayes (pMOM, Intercept) & 200 & 100 & 0.18 & 0.73 & 0.00 \\ 
  Beta-Binomial (pMOM) & 200 & 100 & 0.18 & 0.73 & 0.00 \\ 
  Empirical Bayes (Zellner) & 200 & 100 & 0.16 & 0.77 & 0.00 \\ 
  Empirical Bayes (Zellner, Intercept) & 200 & 100 & 0.18 & 0.73 & 0.00 \\ 
  Beta-Binomial (Zellner) & 200 & 100 & 0.17 & 0.72 & 0.00 \\ 
  LASSO & 200 & 100 & 0.26 & 0.99 & 0.69 \\ 
  ALASSO & 200 & 100 & 0.19 & 0.96 & 0.26 \\ 
  SCAD & 200 & 100 & 0.21 & 0.97 & 0.43 \\ 
\hline
\end{tabular}
\caption{Simulation study. Scenario 2 $(\omega_1=1,\omega_2=0)$.
Mean squared estimation error (MSE), power and false discovery rate (FDR).
Priors on parameters are product MOM (pMOM) and Zellner, both with default unit prior precision.
Priors on model are that given by our empirical Bayes framework, and by a default Beta-Binomial(1,1)
}\label{tab:simstudy2}
\end{table}

\begin{table}[h]
\begin{tabular}{@{}lcc|ccc|@{}}
\hline
Method & n & p & MSE & Power & FDR \\ 
\hline
  Empirical Bayes (pMOM) & 100 & 200 & 2.28 & 0.38 & 0.14 \\ 
  Empirical Bayes (pMOM, Intercept) & 100 & 200 & 2.43 & 0.36 & 0.16 \\ 
  Beta-Binomial (pMOM) & 100 & 200 & 2.04 & 0.35 & 0.09 \\ 
  Empirical Bayes (Zellner) & 100 & 200 & 2.25 & 0.37 & 0.14 \\ 
  Empirical Bayes (Zellner, Intercept) & 100 & 200 & 2.39 & 0.35 & 0.15 \\ 
  Beta-Binomial (Zellner) & 100 & 200 & 1.77 & 0.20 & 0.01 \\ 
  LASSO & 100 & 200 & 1.11 & 0.93 & 0.72 \\ 
  ALASSO & 100 & 200 & 1.22 & 0.77 & 0.38 \\ 
  SCAD & 100 & 200 & 1.99 & 0.77 & 0.55 \\ 
\hline
  Empirical Bayes (pMOM) & 200 & 200 & 0.32 & 0.74 & 0.00 \\ 
  Empirical Bayes (pMOM, Intercept) & 200 & 200 & 0.32 & 0.72 & 0.00 \\ 
  Beta-Binomial (pMOM) & 200 & 200 & 0.32 & 0.72 & 0.00 \\ 
  Empirical Bayes (Zellner) & 200 & 200 & 0.32 & 0.73 & 0.00 \\ 
  Empirical Bayes (Zellner, Intercept) & 200 & 200 & 0.33 & 0.72 & 0.00 \\ 
  Beta-Binomial (Zellner) & 200 & 200 & 0.33 & 0.69 & 0.00 \\ 
  LASSO & 200 & 200 & 0.42 & 1.00 & 0.71 \\ 
  ALASSO & 200 & 200 & 0.31 & 0.97 & 0.28 \\ 
  SCAD & 200 & 200 & 0.73 & 0.95 & 0.26 \\ 
\hline
  Empirical Bayes (pMOM) & 200 & 100 & 0.12 & 0.72 & 0.01 \\ 
  Empirical Bayes (pMOM, Intercept) & 200 & 100 & 0.13 & 0.71 & 0.00 \\ 
  Beta-Binomial (pMOM) & 200 & 100 & 0.13 & 0.70 & 0.00 \\ 
  Empirical Bayes (Zellner) & 200 & 100 & 0.12 & 0.72 & 0.01 \\ 
  Empirical Bayes (Zellner, Intercept) & 200 & 100 & 0.13 & 0.70 & 0.00 \\ 
  Beta-Binomial (Zellner) & 200 & 100 & 0.12 & 0.70 & 0.00 \\ 
  LASSO & 200 & 100 & 0.19 & 0.99 & 0.71 \\ 
  ALASSO & 200 & 100 & 0.13 & 0.94 & 0.26 \\ 
  SCAD & 200 & 100 & 0.14 & 0.96 & 0.45 \\ 
\hline
\end{tabular}
\caption{Simulation study. Scenario 3 ($\omega_1=0,\omega_2=0$).
Mean squared estimation error (MSE), power and false discovery rate (FDR).
Priors on parameters are product MOM (pMOM) and Zellner, both with default unit prior precision.
Priors on model are that given by our empirical Bayes framework, and by a default Beta-Binomial(1,1)
}\label{tab:simstudy3}
\end{table}

\begin{table}[h]
\begin{tabular}{@{}lcc|ccc|@{}}
\hline
Method & n & p & MSE & Power & FDR \\ 
\hline
  Empirical Bayes (pMOM) & 100 & 200 & 4.35 & 0.41 & 0.15 \\ 
  Empirical Bayes (pMOM, Intercept) & 100 & 200 & 4.85 & 0.31 & 0.15 \\ 
  Beta-Binomial (pMOM) & 100 & 200 & 4.70 & 0.29 & 0.12 \\ 
  Empirical Bayes (Zellner) & 100 & 200 & 4.07 & 0.42 & 0.13 \\ 
  Empirical Bayes (Zellner, Intercept) & 100 & 200 & 4.91 & 0.32 & 0.16 \\ 
  Beta-Binomial (Zellner) & 100 & 200 & 5.54 & 0.08 & 0.01 \\ 
  LASSO & 100 & 200 & 2.12 & 0.90 & 0.65 \\ 
  ALASSO & 100 & 200 & 3.33 & 0.69 & 0.28 \\ 
  SCAD & 100 & 200 & 5.10 & 0.68 & 0.58 \\ 
\hline
  Empirical Bayes (pMOM) & 200 & 200 & 0.55 & 0.80 & 0.01 \\ 
  Empirical Bayes (pMOM, Intercept) & 200 & 200 & 0.69 & 0.75 & 0.01 \\ 
  Beta-Binomial (pMOM) & 200 & 200 & 0.69 & 0.75 & 0.01 \\ 
  Empirical Bayes (Zellner) & 200 & 200 & 0.55 & 0.80 & 0.01 \\ 
  Empirical Bayes (Zellner, Intercept) & 200 & 200 & 0.69 & 0.75 & 0.01 \\ 
  Beta-Binomial (Zellner) & 200 & 200 & 0.87 & 0.60 & 0.00 \\ 
  LASSO & 200 & 200 & 0.73 & 1.00 & 0.63 \\ 
  ALASSO & 200 & 200 & 0.58 & 0.98 & 0.28 \\ 
  SCAD & 200 & 200 & 1.81 & 0.94 & 0.43 \\ 
\hline
  Empirical Bayes (pMOM) & 200 & 100 & 0.21 & 0.81 & 0.00 \\ 
  Empirical Bayes (pMOM, Intercept) & 200 & 100 & 0.24 & 0.76 & 0.00 \\ 
  Beta-Binomial (pMOM) & 200 & 100 & 0.24 & 0.75 & 0.00 \\ 
  Empirical Bayes (Zellner) & 200 & 100 & 0.21 & 0.81 & 0.00 \\ 
  Empirical Bayes (Zellner, Intercept) & 200 & 100 & 0.24 & 0.76 & 0.00 \\ 
  Beta-Binomial (Zellner) & 200 & 100 & 0.24 & 0.74 & 0.00 \\ 
  LASSO & 200 & 100 & 0.32 & 1.00 & 0.64 \\ 
  ALASSO & 200 & 100 & 0.24 & 0.98 & 0.23 \\ 
  SCAD & 200 & 100 & 0.27 & 0.98 & 0.37 \\ 
\hline
\end{tabular}
\caption{Simulation study. Scenario 4 ($\omega_1=1.5,\omega_2=0$).
Mean squared estimation error (MSE), power and false discovery rate (FDR).
Priors on parameters are product MOM (pMOM) and Zellner, both with default unit prior precision.
Priors on model are that given by our empirical Bayes framework, and by a default Beta-Binomial(1,1)
}\label{tab:simstudy4}
\end{table}

\begin{table}[h]
\begin{tabular}{@{}lcc|ccc|@{}}
\hline
Method & n & p & MSE & Power & FDR \\ 
\hline
Empirical Bayes (pMOM) & 100 & 200 & 2.87 & 0.41 & 0.17 \\ 
  Empirical Bayes (pMOM, Intercept) & 100 & 200 & 3.05 & 0.35 & 0.18 \\ 
  Beta-Binomial (pMOM) & 100 & 200 & 2.48 & 0.33 & 0.10 \\ 
  Empirical Bayes (Zellner) & 100 & 200 & 2.82 & 0.39 & 0.16 \\ 
  Empirical Bayes (Zellner, Intercept) & 100 & 200 & 3.13 & 0.35 & 0.17 \\ 
  Beta-Binomial (Zellner) & 100 & 200 & 2.38 & 0.17 & 0.02 \\ 
  LASSO & 100 & 200 & 1.31 & 0.92 & 0.71 \\ 
  ALASSO & 100 & 200 & 1.59 & 0.75 & 0.35 \\ 
  SCAD & 100 & 200 & 2.47 & 0.75 & 0.53 \\ 
\hline
  Empirical Bayes (pMOM) & 200 & 200 & 0.37 & 0.74 & 0.01 \\ 
  Empirical Bayes (pMOM, Intercept) & 200 & 200 & 0.39 & 0.72 & 0.01 \\ 
  Beta-Binomial (pMOM) & 200 & 200 & 0.39 & 0.72 & 0.01 \\ 
  Empirical Bayes (Zellner) & 200 & 200 & 0.37 & 0.74 & 0.01 \\ 
  Empirical Bayes (Zellner, Intercept) & 200 & 200 & 0.39 & 0.72 & 0.01 \\ 
  Beta-Binomial (Zellner) & 200 & 200 & 0.42 & 0.68 & 0.00 \\ 
  LASSO & 200 & 200 & 0.49 & 0.99 & 0.69 \\ 
  ALASSO & 200 & 200 & 0.37 & 0.97 & 0.27 \\ 
  SCAD & 200 & 200 & 1.18 & 0.93 & 0.28 \\ 
\hline
  Empirical Bayes (pMOM) & 200 & 100 & 0.15 & 0.73 & 0.00 \\ 
  Empirical Bayes (pMOM, Intercept) & 200 & 100 & 0.16 & 0.72 & 0.00 \\ 
  Beta-Binomial (pMOM) & 200 & 100 & 0.16 & 0.72 & 0.00 \\ 
  Empirical Bayes (Zellner) & 200 & 100 & 0.15 & 0.73 & 0.00 \\ 
  Empirical Bayes (Zellner, Intercept) & 200 & 100 & 0.16 & 0.72 & 0.00 \\ 
  Beta-Binomial (Zellner) & 200 & 100 & 0.16 & 0.71 & 0.00 \\ 
  LASSO & 200 & 100 & 0.24 & 0.99 & 0.70 \\ 
  ALASSO & 200 & 100 & 0.17 & 0.95 & 0.27 \\ 
  SCAD & 200 & 100 & 0.19 & 0.97 & 0.46 \\ 
\hline
\end{tabular}
\caption{Simulation study. Scenario 5 ($\omega_1=0.75,\omega_2=0$).
Mean squared estimation error (MSE), power and false discovery rate (FDR).
Priors on parameters are product MOM (pMOM) and Zellner, both with default unit prior precision.
Priors on model are that given by our empirical Bayes framework, and by a default Beta-Binomial(1,1)
}\label{tab:simstudy5}
\end{table}

\begin{table}[h]
\begin{tabular}{@{}lcc|ccc|@{}}
\hline
Method & n & p & MSE & Power & FDR \\ 
\hline
Empirical Bayes (Zellner, g = 1) & 100 & 200 & 2.820 & 0.386 & 0.163 \\ 
  Empirical Bayes (Zellner, g = 0.1) & 100 & 200 & 2.725 & 0.397 & 0.146 \\ 
  Empirical Bayes (Zellner, g = 10) & 100 & 200 & 2.730 & 0.382 & 0.147 \\ 
\hline
  Empirical Bayes (Zellner, g = 1) & 200 & 200 & 0.372 & 0.743 & 0.009 \\ 
  Empirical Bayes (Zellner, g = 0.1) & 200 & 200 & 0.371 & 0.747 & 0.009 \\ 
  Empirical Bayes (Zellner, g = 10) & 200 & 200 & 0.372 & 0.743 & 0.008 \\ 
\hline
  Empirical Bayes (Zellner, g = 1) & 200 & 100 & 0.154 & 0.733 & 0.005 \\ 
  Empirical Bayes (Zellner, g = 0.1) & 200 & 100 & 0.154 & 0.733 & 0.005 \\ 
  Empirical Bayes (Zellner, g = 10) & 200 & 100 & 0.154 & 0.734 & 0.005 \\ 
\hline
\end{tabular}
\caption{Simulation study. Scenario 5 ($\omega_1=0.75,\omega_2=0$).
Mean squared estimation error (MSE), power and false discovery rate (FDR).
Priors on parameters are Zellner's prior with default precision $g=1$, and also with $g=0.1$ and $g=10$.
Priors on model are that given by our empirical Bayes framework 
}\label{tab:simstudy5_sensitivity}
\end{table}

\section{Colon cancer}\label{sec:colon_cancer_suppl}

\begin{table}[h]
\begin{tabular}{@{}l|ccc|ccc|@{}}
\hline
     & \multicolumn{3}{c|}{EBayes ($R^2=0.533$)} & \multicolumn{3}{c|}{Beta-Binomial ($R^2=0.513$)} \\
Gene  & $E(\theta_j \mid \by)$ & 95\% interval & PIP & $E(\theta_j \mid \by)$ & 95\% interval & PIP \\
\hline
CILP      & 0.22   & (0.17 , 0.28) &  0.97 & 0.22  & ( 0.15 , 0.29) & 0.92\\ 
GAS1      & 0.33   & (0.23 , 0.43) &  0.96 & 0.3   & ( 0    , 0.44) & 0.79\\ 
HIC1      & 0.25   & (0    , 0.35) &  0.85 & 0.25  & ( 0    , 0.35) & 0.76\\ 
ESM1      & 0.16   & (0    , 0.22) &  0.73 & 0.11  & ( 0    , 0.23) & 0.52\\ 
KCNJ5-AS1 & 0.16   & (0    , 0.22) &  0.70 & 0.15  & ( 0    , 0.23) & 0.69\\ 
\hline
\end{tabular}
\caption{TGFB data with Zellner's prior on coefficients. BMA estimates and 95\% intervals for genes with $\pi(\gamma_j=1 \mid \by) > 0.5$ (PIP) in either the empirical Bayes or Beta-Binomial based analysis. $R^2$ is the squared correlation between the outcome and its leave-one-out prediction}
\label{tab:tgfb_topgenes_zellner}
\end{table}

Here we report results for the colon cancer dataset when setting Zellner's unit information prior on the regression parameters.
The results remained very similar.
Table \ref{tab:tgfb_topgenes_zellner} summarizes the results for genes that had a posterior inclusion probability either under the empirical Bayes or the Beta-Binomial model priors.
These genes are the same as those reported in the main paper, and their estimated coefficients, 95\% posterior intervals and posterior inclusion probabilities are all very similar.
The leave-one-out cross-validated $R^2$ was slightly higher for Zellner's prior than for the pMOM: $R^2= 0.533$ for empirical Bayes and $0.513$ for the Beta-Binomial (relative to 0.527 and 0.488 for the pMOM, respectively).

\bibliographystyle{Chicago}

\bibliography{references}
\end{document}